\documentclass[journal]{IEEEtran}
\usepackage[cmex10]{amsmath}
\usepackage{xcolor}
\usepackage[multiple]{footmisc}

\usepackage{mdwtab}
\let\labelindent\relax
\usepackage{enumitem}
\usepackage{lipsum}
\newcounter{MYtempeqncnt}
\setitemize[itemize]{leftmargin=*,itemindent=10pt}
\usepackage{eqparbox}
\usepackage{tabularx}
\usepackage{graphicx,amssymb,rangecite,upgreek,graphicx,dsfont,mathrsfs}
\usepackage{datetime}
\usepackage{algorithm}
\usepackage{algorithmic} 
\usepackage[usenames,dvipsnames]{pstricks}
 \usepackage{epsfig}

 \usepackage{pst-grad} 
 \usepackage{pst-plot} 
\newcommand {\exe} {\stackrel{\cdot} {=}}
\newcommand {\exeq} {\stackrel{\cdot} {\leq}}

\newcommand {\bq} {\mbox{\boldmath $q$}}

\newcommand {\bxt} {\mbox{\footnotesize\boldmath $x$}}

\newcommand {\byt} {\mbox{\footnotesize\boldmath $y$}}

\newcommand {\blb} {\mbox{\boldmath $l$}}
\newcommand {\blt} {\mbox{\footnotesize\boldmath $l$}}
\newcommand {\bu} {\mbox{\boldmath $u$}}

\newcommand {\bx} {\mbox{\boldmath $x$}}
\newcommand {\by} {\mbox{\boldmath $y$}}
\newcommand {\bz} {\mbox{\boldmath $z$}}

\newcommand {\bE} {\mathbb{E}}

\newcommand {\bU} {\mbox{\boldmath $U$}}

\newcommand {\bX} {\mbox{\boldmath $X$}}
\newcommand {\bY} {\mbox{\boldmath $Y$}}
\newcommand {\bZ} {\mbox{\boldmath $Z$}}

\newcommand {\bzt} {\mbox{\boldmath \footnotesize $\bz$}}

\newcommand{\calA}{{\cal A}}
\newcommand{\calB}{{\cal B}}
\newcommand{\calC}{{\cal C}}
\newcommand{\calD}{{\cal D}}
\newcommand{\calE}{{\cal E}}
\newcommand{\calF}{{\cal F}}
\newcommand{\calG}{{\cal G}}

\newcommand{\calI}{{\cal I}}
\newcommand{\calJ}{{\cal J}}

\newcommand{\calL}{{\cal L}}
\newcommand{\calM}{{\cal M}}
\newcommand{\calN}{{\cal N}}

\newcommand{\calR}{{\cal R}}
\newcommand{\calS}{{\cal S}}

\newcommand{\calU}{{\cal U}}
\newcommand{\calV}{{\cal V}}

\newcommand{\calX}{{\cal X}}
\newcommand{\calY}{{\cal Y}}
\newcommand{\calZ}{{\cal Z}}

\newcommand{\be}{\begin{equation}}
\newcommand{\ee}{\end{equation}}
\newcommand{\beqna}{\begin{eqnarray}}
\newcommand{\eeqna}{\end{eqnarray}}

\usepackage{theorem}
\DeclareFontFamily{U}{mathx}{\hyphenchar\font45}
\DeclareFontShape{U}{mathx}{m}{n}{
      <5> <6> <7> <8> <9> <10>
      <10.95> <12> <14.4> <17.28> <20.74> <24.88>
      mathx10
      }{}
\DeclareSymbolFont{mathx}{U}{mathx}{m}{n}

\DeclareMathSymbol{\bigtimes}{1}{mathx}{"91}
\theorembodyfont{\rmfamily}

\newcommand{\abs}[1]{\left|#1\right|}

\newcommand{\defi}{{\triangleq}}

\theoremheaderfont{\itshape}

\newtheorem{theorem}{Theorem}
\newtheorem{proof}{Proof}

\newtheorem{lemma}{Lemma}

\newcommand{\p}[1]{\left(#1\right)}
\newcommand{\pp}[1]{\left[#1\right]}
\newcommand{\ppp}[1]{\left\{#1\right\}}

\newcommand{\markov}{\mbox{$-\hspace{-2.3mm}\circ\hspace{1mm}$}}

\begin{document}

\title{Random Coding Error Exponents for the Two-User Interference Channel}
\author{Wasim~Huleihel
        and~Neri~Merhav
				\\
        The Andrew \& Erna Viterbi Faculty of Electrical Engineering \\
Technion - Israel Institute of Technology \\
Haifa 3200003, ISRAEL\\
E-mail: \{wh@campus, merhav@ee\}.technion.ac.il
\thanks{$^\ast$This research was partially supported by The Israeli Science Foundation (ISF), grant no. 412/12. This paper was presented in part at the 2016 International Zurich Seminar on Communications.}
}
\maketitle

\IEEEpeerreviewmaketitle

\begin{abstract}
\boldmath This paper is about deriving lower bounds on the error exponents for the two-user interference channel under the random coding regime for several ensembles. Specifically, we first analyze the standard random coding ensemble, where the codebooks are comprised of independently and identically distributed (i.i.d.) codewords. For this ensemble, we focus on optimum decoding, which is in contrast to other, suboptimal decoding rules that have been used in the literature (e.g., joint typicality decoding, treating interference as noise, etc.). The fact that the interfering signal is a codeword, rather than an i.i.d. noise process, complicates the application of conventional techniques of performance analysis of the optimum decoder. Also, unfortunately, these conventional techniques result in loose bounds. Using analytical tools rooted in statistical physics, as well as advanced union bounds, we derive single-letter formulas for the random coding error exponents. We compare our results with the best known lower bound on the error exponent, and show that our exponents can be strictly better. Then, in the second part of this paper, we consider more complicated coding ensembles, and find a lower bound on the error exponent associated with the celebrated Han-Kobayashi (HK) random coding ensemble, which is based on superposition coding. 

\end{abstract}
\begin{IEEEkeywords}
Random coding, error exponent, interference channels, superposition coding, Han-Kobayashi scheme, statistical physics, optimal decoding, multiuser communication.  
\end{IEEEkeywords}
\section{Introduction}

\PARstart{T}{he} two-user interference channel (IFC) models a general scenario of communication between two transmitters and two receivers (with no cooperation at either side), where each receiver decodes its intended message from an observed signal, which is interfered by the other user, and corrupted by channel noise. The information-theoretic analysis of this model has begun over more than four decades ago and has recently witnessed a resurgence of interest. Most of the previous work on multiuser communication, and specifically on the IFC, has focused on obtaining inner and outer bounds to the capacity region (see, for example, \cite[Ch. II.7]{GamalKim}). In a nutshell, the study of this kind of channel started in \cite{ShannonTwoWay} and continued in \cite{Ahlswede}, where simple inner and outer bounds to the capacity region were given. Then, in \cite{Carleial}, by using the well-known superposition coding technique, the inner bound of \cite{Ahlswede} was strictly improved. In \cite{Sato}, various inner and outer bounds were obtained by transforming the IFC model into some multiple-access or broadcast channel. Unfortunately, the capacity region for the general interference channel is still unknown, although it has been solved for some special cases \cite{Carleial2,Benzel}. The best known inner bound is the Han-Kobayashi (HK) region, established in \cite{HanKobayashi}, and which will also be considered in this paper.

To our knowledge, \cite{Etkin,ChangEtkin} are the only previous works which treat the error exponents for the IFC under optimal decoding. Specifically, \cite{Etkin} derives lower bounds on error exponents of random codebooks comprised of i.i.d. codewords uniformly distributed over a given type class, under maximum likelihood (ML) decoding at each user, that is, optimal decoding. Contrary to the error exponent analysis of other multiuser communication systems, such as the multiple access channel \cite{scarletNew}, the difficulty in analyzing the error probability of the optimal decoder for the IFC is due to statistical dependencies induced by the interfering signal. Indeed, for the IFC, the marginal channel determining each receiver's ML decoding rule is induced also by the codebook of the interfering user. This extremely complicates the analysis, mostly because the interfering signal is a codeword and not an i.i.d. process. Another important observation, which was noticed in \cite{Etkin}, is that the usual bounding techniques (e.g., Gallager's bounding technique) on the error probability fail to give tight results. To alleviate this problem, the authors of \cite{Etkin} combined some of the ideas from Gallager's bounding technique \cite{Gallager} to get an upper bound on the average probability of decoding error under ML decoding, the method of types \cite{CsisKro}, and used the method of type class enumerators, in the spirit of \cite{MerhavDistance}, which allows to avoid the use of Jensen's inequality in some steps.

The main purpose of this paper is to extend the study of achievability schemes to the more refined analysis of error exponents achieved by the two users, similarly as in \cite{Etkin}. Specifically, we derive single-letter expressions for the error exponents associated with the average error probability, for the finite-alphabet two-user IFC, under several random coding ensembles. The main contributions of this paper are as follows:

\begin{itemize}[leftmargin=*]
\item Similarly as in recent works (see, e.g., \cite{scarletNew,MerhavEr1,MerhavEr2,MerhavEr3,MerhavEr4} and references therein) on the analysis of error exponents, we derive single-letter lower bounds for the random coding error exponents. For the standard random coding ensemble, considered in Subsection \ref{sub:simp}, we analyze the optimal decoder for each receiver, which is interested solely in its intended message. This is in contrast to usual decoding techniques analyzed for the IFC, in which each receiver decodes, in addition to its intended message, also part of (or all) the interfering codeword (that is, the other user's message), or other conventional achievability arguments \cite[Ch. II.7]{GamalKim}, which are based on joint-typicality decoding, with restrictions on the decoder (such as, ``treat interference as noise" or to ``decode the interference"). This enables us to understand whether there is any significant degradation in performance due to the sub-optimality of the decoder. Also, since \cite{Etkin} analyzed the optimal decoder as well, we compare our formulas with those of \cite{Etkin}, and show that our error exponent can be strictly better, which implies that the bounding technique in \cite{Etkin} is not tight. It is worthwhile to mention that the analytical formulas of our error exponents are simpler than the lower bound of \cite{Etkin}. 
\item As was mentioned earlier, in \cite{Etkin} only random codebooks comprised of i.i.d. codewords (uniformly distributed over a type class) were considered. These ensembles are much simpler than the superposition codebooks of \cite{HanKobayashi}. Unfortunately, it is very tedious to analyze superposition codebooks using the methods of \cite{Etkin}. In this paper, however, the new tools that we have derived enable us to analyze more involved random coding ensembles. Indeed, we can consider the coding ensemble used in the HK achievability scheme \cite{HanKobayashi} and derive the respective error exponents. We also discuss an ensemble of hierarchical/tree codes \cite{MerhavEr5}. 
\item The analysis of the error exponents, carried out in this paper, turns out to be much more difficult than in previous works on point-to-point and multiuser communication problems, see, e.g., \cite{scarletNew,MerhavEr1,MerhavEr2,MerhavEr3,MerhavEr4}. Specifically, we encounter two main difficulties in our analysis: First, typically, when analyzing the probability of error, the first step is to apply the union bound. Usually, for point-to-point systems, under the random coding regime, the average error probability can be written as a union of pairwise independent error events. Accordingly, in this case, it is well known that the truncated union bound is exponentially tight \cite[Lemma A.2]{ShluRate}. This is no longer the case, however, when considering multiuser systems, and in particular, the IFC. For the IFC, the events comprising the union are strongly dependent, especially due to the fact that we are considering the optimal decoder. To alleviate this difficulty, following the ideas of \cite{scarletNew}, we derived new upper bounds on the probability of a union of events, which take into account the dependencies among the events. The second difficulty that we have encountered in our analysis is that in contrast to previous works, applying the type class enumerator method \cite{MerhavDistance} is not simple, due to the reason mentioned above. Using some methods from large deviations theory, we were able to tackle this difficulty.  
\item Recently, in \cite{ScarlettCog,ScarlettCogThesis}, the authors independently suggested lower bounds on the error exponents of both standard and cognitive multiple-access channels (MACs), assuming suboptimal successive decoding scheme, and using the standard random coding ensemble (considered in Subsection \ref{sub:simp}). Although the motivation in \cite{ScarlettCog} is different, the codebook construction and the decoding rule are the same as in the first part of this paper, and thus, essentially, their results apply also for the IFC. It is important to emphasize that while we believe that our error exponent analysis is somewhat simpler, at least conceptually, there is strong resemblance between our analysis and \cite{ScarlettCog}, as they both based on type enumeration techniques. Note, however, that while in \cite{ScarlettCog} the standard union bound was used, here, the new upper bounds mentioned above, provide some potential gain over \cite{ScarlettCog}, even for the ordinary ensemble. Also, as was mentioned above, we consider also the more complicated ensemble pertaining to the HK scheme. The derivation of the lower bound on the error exponent of this ensemble is built upon the derivation of the lower bound on the error exponent of the standard random coding ensemble, and thus it makes useful and convenient to start with the analysis of the latter ensemble. We emphasize that the extension of \cite{ScarlettCog} to the HK ensemble is non-trivial. Finally, we mention that the focus in \cite{ScarlettCog} was on achievable rate region, rather than error exponents, and thus no comparison to \cite{Etkin} was provided.
\item We believe that by using the techniques and tools derived in this paper, other multiuser systems, such as the IFC with mismatched decoding, the MAC \cite{scarletNew}, the broadcast channel, the relay channel, etc., and accordingly, other coding schemes, such as binning \cite{MerhavEr1}, and hierarchical codes \cite{MerhavEr5}, can be analyzed.
\end{itemize}

The paper is organized as follows. In Section \ref{sec:not}, we establish notation conventions. In Section \ref{sec:main}, we formalize the problem and assert the main theorems. Specifically, in Subsections \ref{sub:simp} and \ref{sub:HK}, we give the resulting error exponents under the standard random coding ensemble and the HK coding ensemble, respectively. Finally, Section \ref{sec:proof1} is devoted to the proofs of our main results.
\allowdisplaybreaks

\section{Notation Conventions}\label{sec:not}

Throughout this paper, scalar random variables (RVs) will be denoted by capital letters, their sample values will be denoted by the respective lower case letters, and their alphabets will be denoted by the respective calligraphic letters, e.g. $X$, $x$, and $\calX$, respectively. A similar convention will apply to random vectors of dimension $n$ and their sample values, which will be denoted with the same symbols in the boldface font. We also use the notation $X_i^j$ $(j>i)$ to designate the sequence of RVs $(X_i,X_{i+1},\ldots,X_j)$. The set of all $n$-vectors with components taking values in a certain finite alphabet, will be denoted by the same alphabet superscripted by $n$, e.g., $\calX^n$. Generic channels will be usually denoted by the letters $P$, $Q$, or $W$. We shall mainly consider joint distributions of two RVs $(X,Y)$ over the Cartesian product of two finite alphabets $\calX$ and $\calY$. For brevity, we will denote any joint distribution, e.g. $Q_{XY}$, simply by $Q$, the marginals will be denoted by $Q_X$ and $Q_Y$, and the conditional distributions will be denoted by $Q_{X\vert Y}$ and $Q_{Y\vert X}$. The joint distribution induced by $Q_{X}$ and $Q_{Y|X}$ will be denoted by $Q_{X}\times Q_{Y|X}$, and a similar notation will be used when the roles of $X$ and $Y$ are switched. 

The expectation operator will be denoted by $\bE\ppp{\cdot}$, and when we wish to make the dependence on the underlying distribution $Q$ clear, we denote it by $\bE_Q\ppp{\cdot}$. Information measures induced by the generic joint distribution $Q_{XY}$, will be subscripted by $Q$, for example, $I_Q(X;Y)$ will denote the corresponding mutual information, etc. The divergence (or, Kullback-Liebler distance) between two probability measures $Q$ and $P$ will be denoted by $D(Q||P)$. The weighted divergence between two channels, $Q_{Y|X}$ and $P_{Y|X}$, with weight $P_X$, is defined as
\begin{align}
D(Q_{Y|X}||P_{Y|X}\vert P_X)&\triangleq \sum_{x\in\calX}P_X(x)\nonumber\\
&\ \ \ \sum_{y\in\calY}Q_{Y|X}(y\vert x)\log\frac{Q_{Y|X}(y\vert x)}{P_{Y|X}(y\vert x)}.
\end{align}
For a given vector $\bx$, let $\hat{Q}_{\bxt}$ denote the empirical distribution, that is, the vector $\{\hat{Q}_{\bxt}(x),~x\in{\cal X}\}$, where $\hat{Q}_{\bxt}(x)$ is the relative frequency of the letter $x$ in the vector $\bx$. Let $T(P_X)$ denote the type class associated with $P_X$, that is, the set of all sequences $\bx$ for which $\hat{Q}_{\bxt}=P_X$. Similarly, for a pair of vectors $(\bx,\by)$, the empirical joint distribution will be denoted by $\hat{Q}_{\bxt\byt}$, or simply by $\hat{Q}$, for short. All previously defined notation rules for regular distributions will also be used for empirical distributions.

The cardinality of a finite set $\calA$ will be denoted by $\abs{\calA}$, its complement will be denoted by $\calA^c$. The probability of an event $\calE$ will be denoted by $\Pr\ppp{\calE}$. The indicator function of an event $\calE$ will be denoted by $\calI\ppp{\calE}$. For two sequences of positive numbers, $\ppp{a_n}$ and $\ppp{b_n}$, the notation $a_n\exe b_n$ means that $\ppp{a_n}$ and $\ppp{b_n}$ are of the same exponential order, i.e., $n^{-1}\log a_n/b_n\to0$ as $n\to\infty$, where logarithms are defined with respect to (w.r.t.) the natural basis, that is, $\log\p{\cdot} = \ln\p{\cdot}$. 
Finally, for a real number $x$, we denote $\pp{x}_+ \triangleq \max\ppp{0,x}$.

\section{Problem Formulation and Main Results}\label{sec:main}
We divide this section into three subsections. In the first, we present the model and formulate the problem. In the second, we present a lower bound on the IFC error exponent, assuming a simple random coding ensemble where random codebooks comprised of i.i.d. codewords are uniformly distributed over a type class. It is well-known \cite{ChangEtkin} that this coding scheme can be improved by using superposition coding and introducing the notion of ``private" and ``common" messages (to be defined in the sequel). Accordingly, in the third subsection, we consider the HK coding scheme \cite{HanKobayashi}, and derive lower bounds on the error exponents. Finally, we discuss other ensembles that can be analyzed using the same methods.

\subsection{The IFC Model}
Consider a two-user interference channel of two senders, two receivers, and a discrete memoryless channel (DMC), defined by a set of single-letter transition probabilities, $W_{Y_1Y_2\vert X_1X_2}\p{y_1,y_2\vert x_1,x_2}$, with finite input alphabets, $\calX_1,\calX_2$, and finite output alphabets, $\calY_1,\calY_2$. Here, each sender, $k\in\ppp{1,2}$, communicates an independent message $m_k\in\{1,2,\ldots,M_k\triangleq 2^{nR_k}\}$ at rate $R_k$, and each receiver, $l\in\ppp{1,2}$, decodes its respective message. Specifically, a $(2^{nR_1},2^{nR_2},n)$ code $\calC_n$ consists of:
\begin{itemize}
\item Two message sets $\calM_1\triangleq\ppp{0,\ldots,2^{nR_1}-1}$ and $\calM_2\triangleq\ppp{0,\ldots,2^{nR_2}-1}$ for the first and second users, respectively.
\item Two encoders, where for each $k\in\ppp{1,2}$, the $k$-th encoder assigns a codeword $\bx_{k,i}$ to each message $i\in\calM_k$.
\item Two decoders, where each decoder $l\in\ppp{1,2}$ assigns an estimate $\hat{m}_{l}$ to $m_l$.
\end{itemize}

We assume that the message pair $\p{m_1,m_2}$ is uniformly distributed over $\calM_1\times\calM_2$. It is clear that the \emph{optimal decoder} of the first user, for this problem, is given by
\begin{align}
\hat{m}_1 &= \arg\max_{i\in\calM_1}P\p{\by_1\vert \bx_{1,i}}\\
&= \arg\max_{i\in\calM_1}\frac{1}{M_2}\sum_{j=1}^{M_2-1}P\p{\by_1\vert\bx_{1,i},\bx_{2,j}}\label{optDec}
\end{align}
where $P\p{\by_1\vert\bx_{1,i},\bx_{2,j}}$ is the marginal channel defined as
\begin{align}
P\p{\by_1\vert\bx_{1,i},\bx_{2,j}}\triangleq \prod_{k=1}^nW_{Y_1\vert X_1X_2}(y_{1k}\vert x_{1,i,k},x_{2,j,k}),
\end{align}
and
\begin{align}
&W_{Y_1\vert X_1X_2}(y_{1,k}\vert x_{1,i,k},x_{2,j,k})\nonumber\\
&\triangleq \sum_{y_{2,k}\in\calY_2}W_{Y_1Y_2\vert X_1X_2}(y_{1,k},y_{2,k}\vert x_{1,i,k},x_{2,j,k}).
\end{align}
The optimal decoder of the second user is defined similarly. Since there is no cooperation between the two receivers, the error probabilities for the code $\calC_n$, are defined as
\begin{align}
P_{e,i}\p{\calC_n} &\triangleq 2^{-n(R_1+R_2)}\nonumber\\
&\cdot\sum_{\tilde m_1,\tilde m_2}\Pr\ppp{\hat{m}_i\p{Y_i^n}\neq \tilde m_i\vert m_1=\tilde m_1,m_2=\tilde m_2},
\end{align} 
for $i=1,2$.

\subsection{The Ordinary Random Coding Ensemble}\label{sub:simp}
In this subsection, we consider the ordinary random coding ensemble: For each $k\in\ppp{1,2}$, we select independently $M_k$ codewords $\ppp{\bx_{k,i}}$, for $i\in\calM_k$, under the uniform distribution across the type class $T\p{P_{X_k}}$, for a given distribution $P_{X_k}$ on $\calX_k$. Our goal is to assess the exponential rate of $\bar{P}^{(n)}_{e,1} \triangleq \bE \ppp{P_{e,1}\p{\calC_n}}$, where the average is over the code ensemble, that is,
\begin{align}
E_1^*(R_1,R_2) \triangleq \liminf_{n\to\infty}-\frac{1}{n}\log \bar{P}^{(n)}_{e,1},\label{ErrorUser1}
\end{align}
and similarly for the second user. Before stating the main result, we define some quantities. Given a joint distribution $Q_{X_1X_2Y_1}$ over $\calX_1\times\calX_2\times\calY_1$, consider the definitions in \eqref{fdfd}, shown at the top of the next page.
\begin{figure*}[!t]
\normalsize
\setcounter{MYtempeqncnt}{\value{equation}}
\setcounter{equation}{7}
\begin{subequations}\label{fdfd}
\begin{IEEEeqnarray}{rCl}
& &f\p{Q_{X_1X_2Y_1}} \triangleq \bE_Q\pp{\log W_{Y_1\vert X_1X_2}(Y_1\vert X_1,X_2)},\\
& &t_0	(Q_{X_1Y_1}) \triangleq R_2+\max_{\substack{\hat{Q}:\;\hat{Q}_{X_2}=P_{X_2},\;\hat{Q}_{X_1Y_1} = Q_{X_1Y_1}\\ I_{\hat{Q}}(X_2;X_1,Y_1)\leq R_2}}\pp{f(\hat{Q})-I_{\hat{Q}}(X_2;X_1,Y_1)},\label{t0def}\\
& & E_1(\tilde{Q}_{X_1X_2Y_1},Q_{X_1X_2Y_1}) \triangleq \min_{\substack{\hat{Q}:\;\hat{Q}_{X_2}=P_{X_2},\;\hat{Q}_{X_1Y_1} = \tilde{Q}_{X_1Y_1}\\ \hat{Q}\in\calL(\tilde{Q}_{X_1X_2Y_1},Q_{X_1X_2Y_1})}}\pp{I_{\hat{Q}}(X_2;X_1,Y_1)-R_2}_+,\label{E1QtileQdef}\\
& &E_2(\tilde{Q}_{X_1X_2Y_1},Q_{X_1X_2Y_1}) \triangleq \min_{\substack{\hat{Q}:\;\hat{Q}_{X_2}=P_{X_2},\;\hat{Q}_{X_1Y_1} = \tilde{Q}_{X_1Y_1}\\ \hat{Q}\in\hat{\calL}(\tilde{Q}_{X_1X_2Y_1},Q_{X_1X_2Y_1})}}\pp{I_{\hat{Q}}(X_2;Y_1)-R_2}_+,\label{E2QtileQdef}\\
& &\calL(\tilde{Q}_{X_1X_2Y_1},Q_{X_1X_2Y_1}) \triangleq \left\{\hat{Q}:\vphantom{\max\pp{f(\tilde{Q}_{X_1X_2Y_1}),f(\hat{Q})+\pp{R_2-I_{\hat{Q}}(X_2;X_1,Y_1)}_+}}\;\max\pp{t_0(Q_{X_1X_2Y_1}),f(Q_{X_1X_2Y_1})}\right.\nonumber\\
& &\hspace{5cm}\left.\leq\max\pp{f(\tilde{Q}_{X_1X_2Y_1}),f(\hat{Q})+\pp{R_2-I_{\hat{Q}}(X_2;X_1,Y_1)}_+}\right\},\label{LsetCases0}\\
& &\hat{\calL}(\tilde{Q}_{X_1X_2Y_1},Q_{X_1X_2Y_1}) \triangleq \left\{\hat{Q}:\vphantom{\max\pp{f(\tilde{Q}_{X_1X_2Y_1}),f(\hat{Q})+\pp{R_2-I_{\hat{Q}}(X_2;X_1,Y_1)}_+}}\;\max\pp{t_0(Q_{X_1X_2Y_1}),f(Q_{X_1X_2Y_1})}\right.\nonumber\\
& &\hspace{5cm}\left.\leq\max\pp{f(\tilde{Q}_{X_1X_2Y_1}),f(\hat{Q})+\pp{R_2-I_{\hat{Q}}(X_2;Y_1)}_+}\right\},\label{LsetCases}\\
& &\hat{E}_1(Q_{X_1X_2Y_1},R_2) \triangleq \min_{\tilde{Q}:\;\tilde{Q}_{X_1}=P_{X_1},\;\tilde{Q}_{X_2Y_1} = Q_{X_2Y_1}}\pp{I_{\tilde{Q}}(X_1;X_{2},Y_1)+E_1(\tilde{Q}_{X_1X_2Y_1},Q_{X_1X_2Y_1})},\label{E1hatQdef}\\
& &\hat{E}_2(Q_{X_1X_2Y_1},R_2) \triangleq \min_{\tilde{Q}:\;\tilde{Q}_{X_1}=P_{X_1},\;\tilde{Q}_{X_2Y_1} = Q_{X_2Y_1}}E_2(\tilde{Q}_{X_1X_2Y_1},Q_{X_1X_2Y_1}),\label{E2hatQdef}\\
& &E(Q_{X_1X_2Y_1},R_1,R_2) \triangleq \max\ppp{\pp{\hat{E}_1(Q_{X_1X_2Y_1},R_2)-R_1}_+,\hat{E}_2(Q_{X_1X_2Y_1},R_2)},\label{EQdef}\\
& &\tilde E_1(R_1,R_2) \triangleq \min_{\substack{Q_{Y_1|X_1X_2}:\\ Q_{X_1} = P_{X_1},Q_{X_2} =P_{X_2}}}\pp{D(Q_{Y_1|X_{1}X_{2}}||W_{Y_1|X_{1}X_{2}}\vert P_{X_{1}}\times P_{X_{2}})+E(Q_{X_1X_2Y_1},R_1,R_2)}.\label{Estar12}
\end{IEEEeqnarray}
\end{subequations}
\hrulefill
\vspace*{4pt}
\end{figure*}
We devote Appendix \ref{app:2} for a discussion on aspects of the computation of \eqref{Estar12}. We have the following result.
\begin{theorem}\label{th:1}
Let $R_1$ and $R_2$ be given, and let $E^*(R_1,R_2)$ be defined as in \eqref{ErrorUser1}. Consider the ensemble of fixed composition codes of types $P_{X_1}$ and $P_{X_2}$, for the first and second users, respectively. For a discrete memoryless two-user IFC, we have 
\begin{align}
E_1^*(R_1,R_2) \geq \tilde E_1(R_1,R_2),
\end{align}
for any $R_1,R_2\geq0$.
\end{theorem}

Several remarks on Theorem \ref{th:1} are in order.
\begin{itemize}
\item Due to symmetry, the error exponent for the second user, that is, $\tilde E_2(R_1,R_2)$ is simply obtained from Theorem \ref{th:1} by swapping the roles of $X_1$, $Y_1$, and $R_1$, with those of $X_2$, $Y_2$, and $R_2$, respectively.
\item An immediate byproduct of Theorem \ref{th:1} is finding the set of rates $(R_1,R_2)$ for which $\tilde E_1(R_1,R_2)>0$, namely, the rates for which the probability of error vanishes exponentially as $n\to\infty$. We show in Appendix \ref{app:3}, that this set is given by:
\begin{align}
&\calR_{\text{ordinary},1} = \ppp{R_1<I\p{X_1;Y_1}}\cup\nonumber\\
&\ppp{\ppp{R_1+R_2<I\p{X_1,X_2;Y_1}}\cap\ppp{R_1<I\p{X_1;Y_1\vert X_2}}}\label{ordinregion}
\end{align}
evaluated with $P_{X_1X_2Y_1} = P_{X_1}\times P_{X_2}\times W_{Y_1|X_1X_2}$. Note that this region can be obtained also by using standard typicality-based achievability arguments (see, e.g., \cite{Bandemer}). Fig. \ref{fig:1} demonstrates a qualitative description of this region. The interpretation is as follows: The corner point $\p{I\p{X_1;Y_1\vert X_2},I\p{X_2;Y_1}}$ is achieved by first decoding the interference (the second user), canceling it, and then decoding the first user. The sum-rate constraint can be achieved by joint decoding the two users (similarly to MAC), and thus, obviously, also by our optimal decoder. Finally, the region $R_1<I\p{X_1;Y_1}$ and $R_2\geq I\p{X_2;Y_1\vert X_1}$ means that we decode the first user while treating the interference as noise. Evidently, from the perspective of the first decoder, which is interested only in the message transmitted by the first sender, the second sender can use any rate, and thus there is no bound on $R_2$ whenever $R_1<I\p{X_1;Y_1}$. Now, it was shown in \cite{ChangEtkin} that the error exponent achievable for the first user under the ordinary random coding regime is zero outside the closure of $\calR_{\text{ordinary},1}$. Whence, this fact and the above conclusion, characterize the rate region where the attainable exponent with ordinary random coding is positive. Notice that $\calR_{\text{ordinary},1}$ is well-known to be contained in the HK region \cite{ChangEtkin,Bandemer}. 
\begin{figure}[!t]
\begin{minipage}[b]{1.0\linewidth}
  \centering
	\centerline{\includegraphics[width=8cm,height = 5cm]{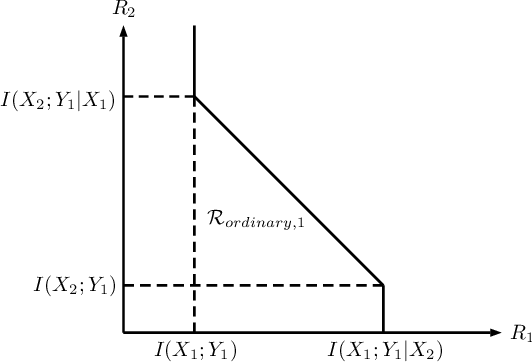}}
	\end{minipage}
\caption{Rate region $\calR_{\text{ordinary},1}$ for which $\tilde E_1(R_1,R_2)>0$.}
\label{fig:1}
\end{figure}
\item \emph{Existence of a single code:} our result holds true on the average, where the averaging is done over the random choice of codebooks. It can be shown (see, for example, \cite[p. 2924]{WengError}) that there exists deterministic sequence of fixed composition codebooks of increasing block length $n$ for which the same asymptotic error performance can be achieved for \emph{both} users simultaneously. 
\item \emph{About the proof:} it is instructive to discuss (in some more detail than earlier) one of the main difficulties in proving Theorem \ref{th:1}, which is customary to multiuser systems, such as the IFC. Without loss of generality, we assume throughout, that the transmitted codewords are $\bx_{1,0}$ and $\bx_{2,0}$. Accordingly, the average probability of error associated with the decoder \eqref{optDec} is given by \eqref{CompUnion}, shown at the top of the next page, 
\begin{figure*}[!t]
\normalsize
\setcounter{MYtempeqncnt}{\value{equation}}
\setcounter{equation}{10}
\begin{align}
\bar{P}^{(n)}_{e,1} &= \Pr\pp{\bigcup_{i=1}^{M_1-1}\ppp{\sum_{j=0}^{M_2-1}P\p{\bY_1\vert \bX_{1,i},\bX_{2,j}}\geq\sum_{j=0}^{M_2-1}P\p{\bY_1\vert \bX_{1,0},\bX_{2,j}}}}\nonumber\\
& = \bE\ppp{\left.\Pr\pp{\bigcup_{i=1}^{M_1-1}\ppp{\sum_{j=0}^{M_2-1}P\p{\bY_1\vert \bX_{1,i},\bX_{2,j}}\geq\sum_{j=0}^{M_2-1}P\p{\bY_1\vert \bX_{1,0},\bX_{2,j}}}\right\vert\calF_0}}\label{CompUnion}
\end{align}
\hrulefill
\vspace*{4pt}
\end{figure*}
where $\calF_0 \triangleq \p{\bX_{1,0},\bX_{2,0},\bY_1}$. In contrast to previous works, applying the type class enumerator\footnote{For a given $y^n\in\calY^n$, and a given joint probability distribution $Q_{XY}$ on $\calX\times\calY$, the \emph{type class enumerator}, $N(Q_{XY})$, is the number of codewords $\ppp{x^n_i}$ in $\calC_n$ whose conditional empirical joint distribution with $y^n$ is $Q_{XY}$, namely, $N(Q_{XY}) = \abs{x^n\in\calC_n:\;\hat{Q}_{x^ny^n} = Q_{XY}}$, where $\hat{Q}_{x^ny^n}$ is the empirical joint distribution of $x^n$ and $y^n$, and $\abs{\calA}$ designates the cardinality of a finite set $\calA$. Type class enumeration method refers to the process of converting a sum of exponentially many terms (usually likelihood functions) into polynomial number of type class enumerators, which are easier to analyze.} method \cite{MerhavDistance}, is not a simple task. Since we are interested in the optimal decoder, each event of the union in \eqref{CompUnion}, depends on the whole codebook of the second user. One may speculate that this problem can be tackled by conditioning on the codebook of the second user, and then \eqref{UnionShul}. However, the cost of this conditioning is a very complicated (if not intractable) large deviations analysis of some quantities. The consequence of this situation is that in order to analyze the probability of error, it is required to analyze the joint distribution of type class enumerators, and not just rely on their marginal distributions, as is usually done, e.g., \cite{MerhavEr1,MerhavEr2,MerhavEr3,MerhavEr4}.

\hspace{0.4cm} Another difficulty is handling the union in \eqref{CompUnion}. By the union bound and Shulman's inequality \cite[Lemma A.2]{ShluRate}, we know that for a sequence of pairwise independent events, $\ppp{\calA_i}_{i=1}^N$, the following holds
\begin{align}
\frac{1}{2}\min\ppp{1,\sum_{i=1}^N\Pr\ppp{\calA_i}}&\leq\Pr\ppp{\bigcup_{i=1}^N\calA_i}\nonumber\\
&\leq\min\ppp{1,\sum_{i=1}^N\Pr\ppp{\calA_i}},\label{UnionShul}
\end{align}
which is a useful result when assessing the exponential behavior of such probabilities. Equation \eqref{UnionShul} is one of the building blocks of tight exponential analysis of previously considered point-to-point systems (see, e.g., \cite{MerhavEr1,MerhavEr2,MerhavEr3,MerhavEr4}, and many references therein). However, in our case the various events are not pairwise independent, and therefore this result cannot be applied. To alleviate this problem, following the techniques of \cite{scarletNew}, we derive new upper bounds on the probability of a union of events, which takes into account such dependencies among the events. 

\item As was mentioned in the Introduction, in \cite{ScarlettCog}, lower bounds on the error exponents of both standard and cognitive multiple-access channels (MACs) were suggested. Although the motivation in \cite{ScarlettCog} is different, their results apply also for the IFC. Now, while in \cite{ScarlettCog} the standard truncated union bound was used, here our new upper bound on the probability of a union of events, provides some potential gain over \cite{ScarlettCog}. Specifically, the lower bound in \cite{ScarlettCog} is the same as \eqref{Estar12} but without the $\hat{E}_2(Q,R_2)$ term, i.e., it is given by $\min_Q\big\{D(Q||W)+\big[\hat{E}_1(Q,R_2)-R_1\big]_+\big\}$, and thus, in general, our result may be tighter. It should be stressed, however, that we have not identified specific examples where the new term, namely, $\hat{E}_2(Q,R_2)$, dominated the maximum in \eqref{EQdef}.

\item The lower bound in \cite{Etkin} is extremely complicated, and it is very difficult to compare it analytically to the lower bound in Theorem \ref{th:1}. Nonetheless, we can still claim (in general) that our lower bound is at least as good as the lower bound in \cite{Etkin}. Indeed, the first step in the analysis of the error exponent in both our paper and in \cite{Etkin} is applying the union bound (actually, here, we employ a tighter union bound). However, it will be seen that every other passage in our analysis is exponentially \emph{exact}, while in \cite{Etkin}, some steps are associated with inequalities that may cause gaps in the exponential scale, and thus in general, $\tilde{E}(R_1,R_2)\geq E_{\text{[9]}}(R_1,R_2)$, for any $(R_1,R_2)\in\mathbb{R}^2_+$, where $E_{\text{[9]}}(R_1,R_2)$ is the lower bound in \cite{Etkin}.

\item \emph{Comparison with \cite{Etkin}}: Similarly to \cite{Etkin}, we present results for the following channel: $Y_1 = X_1\cdot X_2\oplus Z$ and $Y_2 = X_2$, where $X_1,X_2,Y_1,Y_2\in\ppp{0,1}$, $Z\sim\text{Bern}(p)$, ``$\cdot$" is multiplication, and ``$\oplus$" is modulo-2 addition. In the numerical calculations, we fix $p = 0.01$. Fig. \ref{fig:OurvsEtkin} presents the lower bound on the error exponent under optimal decoding, derived in this paper, compared to the lower bound $E_{\text{[9]}}(R_1,R_2)$ of \cite{Etkin}, as a function of $R_1$, for different values of $P_{X_1}$, $P_{X_2}$, and $R_2$. It can be seen that our exponents are strictly better than those of \cite{Etkin}. 
\begin{figure}[!t]
\begin{minipage}[b]{1.0\linewidth}
  \centering
	\centerline{\includegraphics[width=10cm,height = 7.5cm]{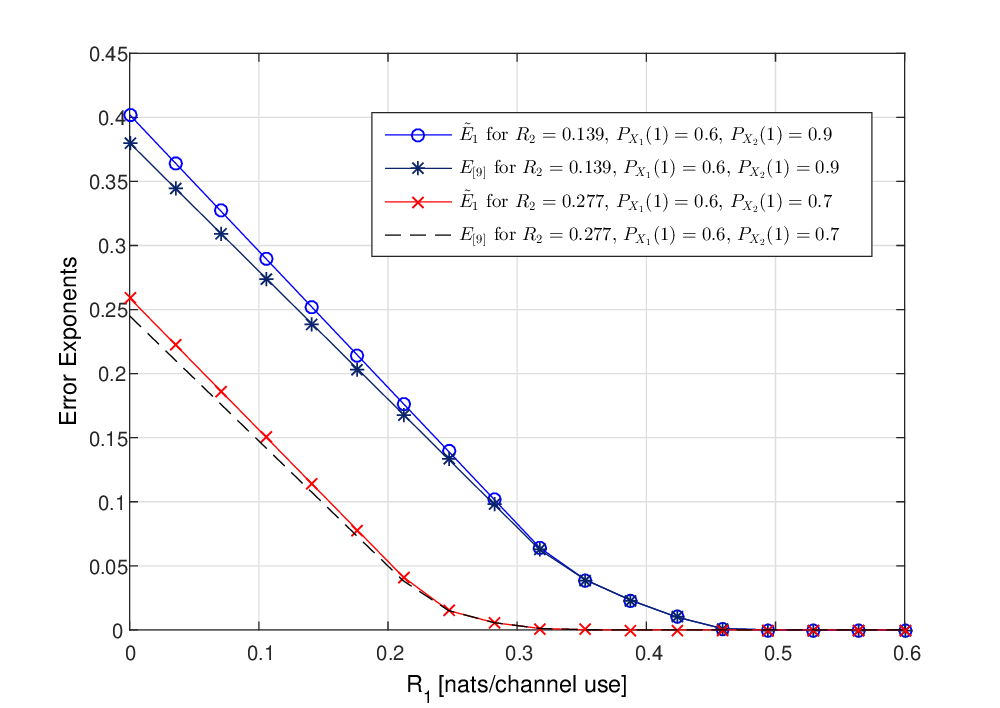}}
	\end{minipage}
\caption{Comparison between $\tilde E_1(R_1,R_2)$ and $E_{\text{[9]}}(R_1,R_2)$ of \cite{Etkin}, as a function of $R_1$ for two different values of $R_2$ and fixed choices of $P_{X_1}$ and $P_{X_2}$.}
\label{fig:OurvsEtkin}
\end{figure}
\end{itemize}

\subsection{The Han-Kobayashi Coding Scheme}\label{sub:HK}

Consider the channel model of Subsection \ref{sub:simp}. The best known inner bound on the capacity region is achieved by the HK coding scheme \cite{HanKobayashi}. The idea of this scheme is to split the message $m_1$ into ``private" and ``common" messages, $m_{11}$ and $m_{12}$ at rates $R_{11}$ and $R_{12}$, respectively, such that $R_1 = R_{11}+R_{12}$. Similarly, $m_2$ is split into $m_{21}$ and $m_{22}$ at rates $R_{21}$ and $R_{22}$, with $R_2 = R_{21}+R_{22}$. The intuition behind this splitting is based on the receiver behavior at low and high signal-to-noise-ratio (SNR). Specifically, it is well-known \cite{GamalKim} that: (1) when the SNR is low, treating the interference as noise is an optimal strategy, and (2) when the SNR is high, decoding and then canceling the interference is the optimal strategy. Accordingly, the above splitting captures the general intermediate situation where the first decoder, for example, is interested only in partial information from  the second user, in addition to its own intended message. 

Next, we describe explicitly the coding strategy of \cite{HanKobayashi}. Fix a distribution $P_{Z_{11}}P_{Z_{12}}P_{Z_{21}}P_{Z_{22}}P_{X_1|Z_{11}Z_{12}}P_{X_2|Z_{21}Z_{22}}$, where the latter two conditional distributions represent deterministic mappings. For each $k,k'\in\ppp{1,2}$, randomly and conditionally independently generate a sequence $\bz_{k,k'}(m_{k,k'})$ under the uniform distribution across the type class $T(P_{Z_{kk'}})$ for a given $P_{Z_{k,k'}}$. To communicate a message pair $(m_{11},m_{12})$, sender 1 transmits $\bx_1(\bz_{11},\bz_{12})$, and analogously for sender 2. All our results can be extended to the setting in which the codewords are generated conditionally on a time-sharing sequence $\bq$. However, this leads to more complicated notation. Thus, we focus primarily on the case without time-sharing.

Let us now describe the operation of each receiver. Receiver $k = 1,2$, recovers its intended message $m_k$ and the common message from the other sender (although it is not required to). This scheme is illustrated in Fig. \ref{fig:HK}. Note that this decoding operation is the one that was used in \cite{HanKobayashi}, but there, the sub-optimal non-unique simultaneous joint typical decoder \cite[Ch. II.7]{GamalKim} was used. Here, by contrast, we use sub-optimal ML decoding (the sub-optimality is due to the fact that our decoder recovers also the common message from the other sender). As will be explained in the sequel, analyzing the optimal ML decoder is a challenging task, and therefore we will focus on sub-optimal ML decoding.

We wish to find a lower bound on the error exponent, achieved by the HK encoding functions, in conjunction with the above described decoding functions. To this end, note that by combining the channel and the deterministic mappings as indicated by the dashed box in Fig. \ref{fig:HK}, the channel $(Z_{11},Z_{12},Z_{21},Z_{22})\mapsto(Y_1,Y_2)$ is just a four-sender, two-receiver, DMC interference channel, with virtual inputs. Note that this formulation induces the Markovian structure $(Z_{11},Z_{12},Z_{21},Z_{22})\markov (X_1,X_2)\markov (Y_1,Y_2)$, where the (virtual) input distributions, i.e., $P_{Z_{k,k'}}$ for $k,k'\in\ppp{1,2}$, can be optimized.
\begin{figure}[!t]
\begin{minipage}[b]{1.0\linewidth}
  \centering
	\centerline{\includegraphics[width=9cm,height = 3.6cm]{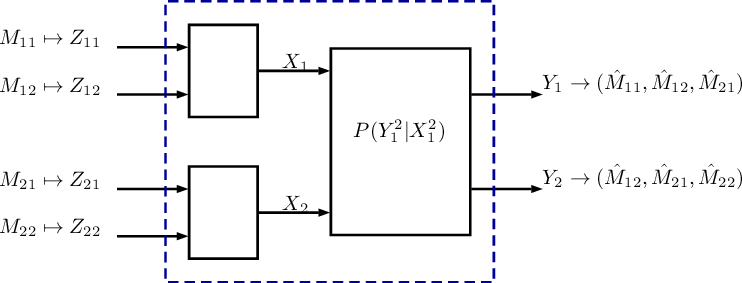}}
	\end{minipage}
\caption{Han-Kobayashi coding scheme.}
\label{fig:HK}
\end{figure}
We assume that the message quadruple $\p{m_{11},m_{12},m_{21},m_{22}}$ is uniformly distributed over $\calM_{11}\times\calM_{12}\times\calM_{21}\times\calM_{22}$. Following the above descriptions, our decoder for this problem is given by
\begin{align}
&(\hat{m}_{11},\hat{m}_{12},\hat{m}_{21})= \nonumber\\
&\arg\max_{(i,j,k)\in\calM_{11}\times\calM_{12}\times\calM_{21}}P\p{\by_1\vert \bz_{11,i},\bz_{12,j},\bz_{21,k}}\\
&= \arg\max_{(i,j,k)\in\calM_{11}\times\calM_{12}\times\calM_{21}}\nonumber\\
&\hspace{1.5cm}\frac{1}{M_{22}}\sum_{l=0}^{M_{22}-1}P\p{\by_1\vert \bz_{11,i},\bz_{12,j},\bz_{21,k},\bz_{22,l}}.\label{optDecHK}
\end{align}
Accordingly, the probability of error for the code $\calC_n$ and for the first user, is defined as 
\begin{align}
P_{e,1}\p{\calC_n}\triangleq\Pr\ppp{(\hat{m}_{11},\hat{m}_{12})\neq (m_{11},m_{12})},
\end{align}
and similarly for the second user. Our goal is to assess the exponential rate of $\bar{P}_{e,1}^{(n)} \triangleq \bE\ppp{P_{e,1}\p{\calC_n}}$, where the average is over the code ensemble, namely,
\begin{align}
E_{\text{HK}}^*(R_1,R_2) \triangleq \liminf_{n\to\infty}-\frac{1}{n}\log \bar{P}^{(n)}_{e,1},\label{ErrorHK}
\end{align}
and similarly for the second user. In order to facilitate the presentation of the following result, we move the technical definitions to Appendix~\ref{app:0}. Our second main result is the following. 
\begin{theorem}\label{th:2}
Let $E_{\text{HK}}^*(R_1,R_2)$ be defined as in \eqref{ErrorHK}. Consider the HK encoding scheme described above. For a discrete memoryless two-user IFC, we have: 
\begin{align}
E_{\text{HK}}^*(R_1,R_2) \geq  \max_{\substack{(R_{11},R_{12},R_{21},R_{22}):\;\\R_{11}+R_{12}=R_1\\ R_{21}+R_{22}=R_2}}\tilde E_{\text{HK}}(R_{11},R_{12},R_{21},R_{22}),
\end{align}
for any $R_1,R_2\geq0$, where $\tilde E_{\text{HK}}(R_{11},R_{12},R_{21},R_{22})$ is given in \eqref{appendixA}.
\end{theorem}

Several remarks on Theorem \ref{th:2} are in order.
\begin{itemize}
\item As before, an immediate byproduct of Theorem \ref{th:2} is finding the set of rates $(R_{11},R_{12},R_{21},R_{22})$ for which $\tilde E_{\text{HK}}(R_{11},R_{12},R_{21},R_{22})>0$, namely, for which the probability of error vanishes exponentially as $n\to\infty$. It can be shown that this set is given by the HK region, that is,
\begin{subequations}
\label{Region2}
\begin{IEEEeqnarray}{rCl}
R_{11}&\leq& I(Z_{1};Y_1\vert Z_2,Z_3),\\
R_{12}&\leq& I(Z_{2};Y_1\vert Z_1,Z_3),\\
R_{21}&\leq& I(Z_{3};Y_1\vert Z_1,Z_2),\\
R_{11}+R_{12}&\leq& I(Z_{1},Z_2;Y_1\vert Z_3),\\
R_{11}+R_{21}&\leq& I(Z_{1},Z_3;Y_1\vert Z_2),\\
R_{12}+R_{21}&\leq& I(Z_{2},Z_3;Y_1\vert Z_1),\\
R_{11}+R_{12}+R_{21}&\leq& I(Z_{1},Z_2,Z_3;Y_1),
\end{IEEEeqnarray}
\end{subequations}
evaluated with 
$$P_{Z_1^4Y_1} = P_{Z_1} P_{Z_2} P_{Z_3} P_{Z_4} P_{X_1|Z_1Z_2}P_{X_2|Z_3Z_4}W_{Y_1|X_1X_2}$$
and similarly for the second user, where $P_{X_1|Z_1Z_2}$ and $P_{X_2|Z_3Z_4}$ represent deterministic mappings. As was mentioned earlier, it is possible to introduce a time-sharing sequence $\bq$, and accordingly, \eqref{Region2} remains almost the same, but with a time-sharing RV $Q$ (with alphabet size bounded by eight \cite{HanKobayashi}), appearing at the conditioning of each the above mutual information terms. 
Finally, we mention that in \cite{Bandemer} it was shown that by using the optimal ML decoder (given in \eqref{optdecHK1}) instead of the non-unique simultaneous joint typical decoder \cite{HanKobayashi}, we cannot improve the achievable region. This observation do not for the error exponent. 
\item It can be shown\footnote{By definition, the ordinary ensemble is a simple instance of the HK ensemble, and thus the latter is indeed better upon optimization of the auxiliary RVs $\ppp{Z_{ij}}$. To see that the ordinary ensemble is a special case of the HK ensemble, we take $Z_{11}=X_1$, $Z_{12}=Z_{21}=\emptyset$, and $Z_{22}=X_2$.} that the error exponent in Theorem~\ref{th:2} is no worse than the error exponent in Theorem~\ref{th:1}, namely, $\tilde{E}_{\text{HK}}(R_{11},R_{12},R_{21},R_{22})\geq\tilde E_1(R_1,R_2)$ for any $(R_1,R_2)$ such that $R_1 = R_{11}+R_{12}$ and $R_2 = R_{21}+R_{22}$. Moreover, it is well-known that upon optimizing the auxiliary RVs, $\ppp{Z_{ij}}$, the HK region in \eqref{Region2} is \emph{strictly} better than $\calR_{\text{ordinary},1}$. Therefore, this necessarily implies that for a certain region of high rates, the HK error exponent in Theorem~\ref{th:2} will be positive while the standard random coding error exponent in Theorem~\ref{th:1} will be zero. On the other extreme, it is easy to show that for $(R_1,R_2)=(0,0)$ the error exponent in Theorem~\ref{th:2} equals to the error exponent in Theorem~\ref{th:1}, so for small rates there is no improvement in the error exponents.
\item Contrary to the ordinary ensemble, described in Subsection \ref{sub:simp}, the HK ensemble depends on some auxiliary RVs which should be optimized. For a give pair of rates $(R_1,R_2)$, our error exponent formula provides a criterion for the choice of the optimal auxiliary RVs: maximize the lower bound on the error exponent in Theorem~\ref{th:2}, w.r.t. the auxiliaries, subject to some relevant constraints. As a matter of fact, for a \emph{given} pair of rates $(R_1,R_2)$, it is very likely that the optimal choice of these auxiliaries will be different from the optimal choice for the same pair in the achievable region. Indeed, even in the single-user case, the capacity achieving distribution is usually different from the optimal distribution in the error exponent sense.
\item Using the same techniques and tools derived in this paper, we can consider other random coding ensembles. For example, we can analyze the error exponents resulting from the \emph{hierarchical code ensemble}. Specifically, in this ensemble, the message $m_1$ is split into common and private messages $m_{11}$, $m_{12}$ at rates $R_{11}$ and $R_{12}$, respectively, such that $R_1 = R_{11}+R_{12}$. Similarly $m_2$ is split into $m_{21}$, $m_{22}$ at rates $R_{21}$ and $R_{22}$, respectively, such that $R_2 = R_{21}+R_{22}$. Then, we first randomly draw a rate $R_{11}$ codebook of block length $n$ according to a given distribution. Then, for each such codeword, we randomly and conditionally independently generate a rate $R_{12}$ codebook of block length $n$. In other words, the code has a tree structure with two levels, where the first serves for ``cloud centers", and the second for the ``satellites". We do the same for the second user. Under this ensemble, we can analyze the optimal decoder. Note, however, that this ensemble is different from the product ensemble considered in Theorem \ref{th:2}. Indeed, while for the former for each first stage codeword (cloud center) we independently draw a new codebook (satellites), for the latter, for each cloud center we have the same satellite codebook. Loosely speaking, this means that the product ensemble is ``less random". From the point of view of achievable region, however, the hierarchical ensemble is equivalent to the product ensemble used in HK scheme \cite[Ch. II.7]{GamalKim}. Nonetheless, the error exponents associated with these ensembles could be different.
\item In Theorem \ref{th:2} we assumed the sub-optimal decoder given in \eqref{optDecHK}. Indeed, the optimal decoder for our problem is given by:
\begin{align}
&(\hat{m}_{11},\hat{m}_{12}) = \arg\max_{i,j}P\p{\by_1\vert \bz_{11,i},\bz_{12,j}}\\
&\hspace{1.64cm}= \arg\max_{i,j}\nonumber\\
&\frac{1}{M_{21}M_{22}}\sum_{k=0}^{M_{21}-1}\sum_{l=0}^{M_{22}-1}P\p{\by_1\vert \bz_{11,i},\bz_{12,j},\bz_{21,k},\bz_{22,l}}.\label{optdecHK1}
\end{align}
Unfortunately, it turns out that analyzing the HK scheme (in conjunction with \eqref{optdecHK1}) is much more difficult, and requires some more delicate tools from large deviations theory. Specifically, the main difficulty in the derivations, is to analyze the large deviations behavior of a two-dimensional sum (due to the double summation in \eqref{optdecHK1}) involving binomial RVs which are strongly dependent (contrary to the standard one-dimensional version, see, e.g., \cite[p. 6027-6028]{MerhavEr1}). Nonetheless, we note that for the hierarchical code ensemble described above, the optimal decoder can be analyzed. Indeed, for this ensemble, it is clear that the optimal decoder is given by
\begin{align}
&(\hat{m}_{11},\hat{m}_{12}) = \arg\max_{i,j}P\p{\by_1\vert \bx_1(i,j)}\\
&= \arg\max_{i,j}\frac{1}{M_{21}M_{22}}\sum_{k=0}^{M_{21}-1}\sum_{l=0}^{M_{22}-1}P\p{\by_1\vert \bx_1(i,j),\bx_2(k,l)}
\end{align}
where $\bx_1(i,j) \triangleq f_1(\bx_1'(i),\bx_1''(i,j))$ and $\bx_2(i,j) \triangleq f_2(\bx_2'(i),\bx_2''(i,j))$ due to the hierarchical structure. Now, while here too, we will deal with two-dimensional summation, the summands will be independent, given the cloud centers codebook, and the proof can be carried out smoothly.
\end{itemize}


\section{Proofs of Main Results}\label{sec:proof1}
\subsection{Proof of Theorem \ref{th:1}}\label{subsec:proof1}
Without loss of generality, we assume throughout, that the transmitted codewords are $\bx_{1,0}$ and $\bx_{2,0}$, and due to the fact that we analyze the first decoder, for convenience, we use $\by$ instead of $\by_1$. Accordingly, the average probability of error associated with the optimal decoder \eqref{optDec}, is given by \eqref{innerPro}, shown at the top of the next page, 
\begin{figure*}[!t]
\normalsize
\setcounter{MYtempeqncnt}{\value{equation}}
\setcounter{equation}{22}
\begin{align}
\bar P^{(n)}_{e,1} &= \Pr\pp{\bigcup_{i=1}^{M_1-1}\ppp{\sum_{j=0}^{M_2-1}P\p{\bY\vert \bX_{1,i},\bX_{2,j}}\geq\sum_{j=0}^{M_2-1}P\p{\bY\vert \bX_{1,0},\bX_{2,j}}}}\label{OrignalProb}\\
& = \bE\ppp{\left.\Pr\pp{\bigcup_{i=1}^{M_1-1}\ppp{\sum_{j=0}^{M_2-1}P\p{\bY\vert \bX_{1,i},\bX_{2,j}}\geq\sum_{j=0}^{M_2-1}P\p{\bY\vert \bX_{1,0},\bX_{2,j}}}\right\vert\calF_0}}\label{innerPro}
\end{align}
\hrulefill
\vspace*{4pt}
\end{figure*}
where $\calF_0 \triangleq \p{\bX_{1,0},\bX_{2,0},\bY}$. 
In the following, we propose new upper bound on the probability of a union of events, which are suitable for some structured dependency between the events, as above. 

In order to give some motivation for this new bound, we first rewrite \eqref{OrignalProb} in another (equivalent) form. Specifically, we express \eqref{innerPro} in terms of the joint types of $\p{\bX_{1,0},\bX_{2,0},\bY}$ and $\ppp{\p{\bY,\bX_{1,i},\bX_{2,j}}}_{i,j}$. First, for a given joint distribution $Q_{X_1X_2Y}$ of $(\bx_1,\bx_2,\by)$, let
\begin{align}
f\p{Q_{X_1X_2Y}} &\triangleq \frac{1}{n}\log P\p{\by\vert\bx_1,\bx_2}\\
& = \bE_Q \pp{\log W_{Y\vert X_1X_2}(Y\vert X_1X_2)}.
\end{align}
For a given joint type $Q_{X_{1,0}X_{2,0}Y}$ of the random vectors $\p{\bX_{1,0},\bX_{2,0},\bY}$, define the set $T_I\p{Q_{X_{1,0}X_{2,0}Y}}$, given in \eqref{IFCTypical}, shown at the top of the next page.
\begin{figure*}[!t]
\normalsize
\setcounter{MYtempeqncnt}{\value{equation}}
\setcounter{equation}{26}
\begin{align}
& T_I\p{Q_{X_{1,0}X_{2,0}Y}}\defi\left\{\tilde{Q}^0_{X_1X_{2,0}Y}\in\calS_0,\p{\ppp{\tilde{Q}^k_{X_1X_2Y}}_{k=1}^{M_2-1},\ppp{\hat{Q}^k_{X_{1,0}X_2Y}}_{k=1}^{M_2-1}}\in\calS_1:\;\right.\nonumber\\
&\left.\ \ \ \ \ \ \ \ \ \ \ \ \ \ \ \ \ \ \ \ \ \ \ e^{nf(\tilde{Q}_{X_1X_{2,0}Y}^0)}+\sum_{k=1}^{M_2-1}\pp{e^{nf(\tilde{Q}_{X_1X_2Y}^k)}-e^{nf(\hat{Q}_{X_1X_2Y}^k)}}\geq e^{nf(Q_{X_{1,0}X_{2,0}Y})}\right\},\label{IFCTypical}\\
&\calS_0(Q_{X_{1,0}X_{2,0}Y}) \triangleq \ppp{\tilde{Q}^0_{X_1X_{2,0}Y}:\;\tilde{Q}^0_{X_1} = P_{X_1},\tilde{Q}^0_{X_2} = P_{X_2},\tilde{Q}^0_{X_{2,0}Y} = Q_{X_{2,0}Y}},\label{margcons1}\\
\calS_1(Q_{X_{1,0}X_{2,0}Y})&\triangleq \left\{\ppp{\tilde{Q}^k_{X_1X_2Y}}_{k=1}^{M_2-1},\ppp{\hat{Q}^k_{X_{1,0}X_2Y}}_{k=1}^{M_2-1}:\;\tilde{Q}^k_{X_1}= P_{X_1},\tilde{Q}^k_{X_2}= P_{X_2},\tilde{Q}^k_{Y}= Q_{Y},\right.\nonumber\\
&\ \ \ \ \ \ \ \ \ \ \ \hat{Q}^k_{X_{1,0}}= P_{X_1},\hat{Q}^k_{X_2}= P_{X_2},\hat{Q}^k_{X_{1,0}Y}= Q_{X_{1,0}Y},\;\forall 1\leq k\leq M_2-1\nonumber\\
&\left.\ \ \ \ \ \ \ \ \ \ \ \ \ \ \ \ \ \ \ \ \ \ \ \ \ \ \ \ \ \ \ \ \ \ \ \ \ \ \ \ \ \ \ \ \ \ \vphantom{\ppp{\tilde{Q}^k_{X_1X_2Y}}_{k=1}^{M_2-1},\ppp{\hat{Q}^k_{X_1X_2Y}}_{k=1}^{M_2-1}:}\tilde{Q}_{X_2Y}^k=\hat{Q}_{X_2Y}^k,\tilde{Q}_{X_1Y}^k=\tilde{Q}_{X_1Y}^m, \forall k,m\right\}.\label{margcons2}
\end{align}
\hrulefill
\vspace*{4pt}
\end{figure*}
The set $ T_I(Q_{X_{1,0}X_{2,0}Y})$ is the set of all possible types of $\p{\bX_{1,i},\calC_2}$, where $\calC_2$ denotes the codebook of the second user, which lead to a decoding error when $\p{\bX_{1,0},\bX_{2,0},\bY}\in T(Q_{X_{1,0}X_{2,0}Y})$ is transmitted. The various marginal constraints in \eqref{margcons1} and \eqref{margcons2} arise from the fact that we are assuming constant-composition random coding and, of course, fixed marginals due to the given fixed joint distribution $Q_{X_{1,0}X_{2,0}Y}$. Finally, the constraint
\begin{align}
&e^{nf(\tilde{Q}_{X_1X_{2,0}Y}^0)}+\sum_{k=1}^{M_2-1}\pp{e^{nf(\tilde{Q}_{X_1X_2Y}^k)}-e^{nf(\hat{Q}_{X_1X_2Y}^k)}}\nonumber\\
&\hspace{2cm}\geq e^{nf(Q_{X_{1,0}X_{2,0}Y})}\label{ifand}
\end{align}
in \eqref{IFCTypical}, represents a decoding error event, that is, it holds if and only if 
\begin{align}
\sum_{j=0}^{M_2-1}P\p{\by\vert \bx_{1,i},\bx_{2,j}}\geq\sum_{j=0}^{M_2-1}P\p{\by\vert \bx_{1,0},\bx_{2,j}},\label{onlyif}
\end{align}
or, equivalently,
\begin{align}
&P\p{\by\vert \bx_{1,i},\bx_{2,0}}+\sum_{j=1}^{M_2-1}\pp{P\p{\by\vert \bx_{1,i},\bx_{2,j}}-P\p{\by\vert \bx_{1,0},\bx_{2,j}}}\nonumber\\
&\hspace{2cm}\geq P\p{\by\vert \bx_{1,0},\bx_{2,0}},
\end{align}
for $\p{\bx_{1,0},\bx_{2,0},\by}\in T(Q_{X_{1,0}X_{2,0}Y})$, $\p{\bx_{1,i},\bx_{2,0},\by}\in T(\tilde{Q}^0_{X_1X_{2,0}Y})$, $\p{\bx_{1,i},\bx_{2,j},\by}\in T(\tilde{Q}^j_{X_1X_2Y})$, and $\p{\bx_{1,0},\bx_{2,j},\by}\in T(\hat{Q}^j_{X_{1,0}X_2Y})$, for $j=1,2,\ldots,M_2-1$. Now, with these definitions, fixing $Q_{X_{1,0}X_{2,0}Y}$, and letting $\p{\bx_{1,0},\bx_{2,0},\by}$ be an arbitrary triplet of sequences such that $\p{\bx_{1,0},\bx_{2,0},\by}\in T(Q_{X_{1,0}X_{2,0}Y})$, it follows, by definition, that the error event
\begin{align}
\bigcup_{i=1}^{M_1-1}\ppp{\sum_{j=0}^{M_2-1}P\p{\bY\vert \bX_{1,i},\bX_{2,j}}\geq\sum_{j=0}^{M_2-1}P\p{\bY\vert \bX_{1,0},\bX_{2,j}}}\label{EquivErrorEvent00}
\end{align}
can be rewritten, in terms of types, as follows
\begin{align}
&\bigcup_{i=1}^{M_1-1}\bigcup_{\ppp{\tilde{Q}_{X_1X_2Y}^j,\hat{Q}_{X_1X_2Y}^j}_j\in T_I(Q_{X_{1,0}X_{2,0}Y})}\nonumber\\
&\left\{ \begin{array}{ccc}
\p{\bX_{1,i},\bx_{2,0},\by}\in T(\tilde{Q}^0_{X_1X_{2,0}Y}), \\
\ppp{\p{\bX_{1,i},\bX_{2,j},\by}\in T(\tilde{Q}^j_{X_1X_2Y})}_{j=1}^{M_2-1}, \\
\ppp{\p{\bx_{1,0},\bX_{2,j},\by}\in T(\hat{Q}^j_{X_{1,0}X_2Y})}_{j=1}^{M_2-1} \end{array} \right\}.\label{EquivErrorEvent}
\end{align}
We wish to analyze the probability of the event in \eqref{EquivErrorEvent}, conditioned on $\calF_0$. Note that the inner union in \eqref{EquivErrorEvent} is over vectors of types (an exponential number of them). Finally, for the sake of convenience, we simplify the notations of \eqref{EquivErrorEvent}, and write it equivalently as
\begin{align}
\bigcup_{i=1}^{M_1-1}\bigcup_{\blt}\left\{ \begin{array}{ccc}
\bX_{1,i}\in \calA_{\blt,0}, \\
\p{\bX_{1,i},\bX_{2,j}} \in \calA_{\blt,j},\ \text{for }j=1,\ldots,M_2-1, \\
\bX_{2,j}\in \tilde{\calA}_{\blt,j},\ \ \ \text{for }j=1,\ldots,M_2-1 \end{array} \right\}\label{EquivErrorEvent1}
\end{align}
where, again, the index ``$\blb$" in the inner union runs over the combinations of types (namely, $\blb = \{\tilde{Q}_{X_1X_2Y}^j,\hat{Q}_{X_1X_2Y}^j\}_j$) that belong to $T_I(Q_{X_{1,0}X_{2,0}Y})$, and the various sets $\{\calA_{\blt,j},\tilde{\calA}_{\blt,j}\}_{\blt,j}$ correspond to the typical sets in \eqref{EquivErrorEvent} (recall that $(\bx_{1,0},\bx_{2,0},\by)$ are given at this stage). Next, following the ideas of \cite{scarletNew}, we provide a new upper bound on a generic probability which has the form of \eqref{EquivErrorEvent1}. The proof of this lemma is relegated to Appendix \ref{app:1}. 

\begin{lemma}\label{lem:Union3}
Let $\ppp{V_1\p{i}}_{i=1}^{L_1},V_2,V_3,\ldots,V_K$ be independent sequences of independently and identically distributed (i.i.d.) RVs on the alphabets $\calV_1\times\calV_2\times\ldots\times\calV_K$, respectively, with $V_1\p{i}\sim P_{V_1},V_2\sim P_{V_2},\ldots,V_K\sim P_{V_K}$. Fix a sequence of sets $\ppp{\calA_{i,1}}_{i=1}^N,\ppp{\calA_{i,2}}_{i=1}^N,\ldots,\ppp{\calA_{i,K-1}}_{i=1}^N$, where $\calA_{i,j}\subseteq\calV_1\times\calV_{j+1}$, for $1\leq j\leq K-1$ and for all $1\leq i\leq N$. Also, fix a set $\ppp{\calA_{i,0}}_{i=1}^N$ where $\calA_{i,0}\subseteq\calV_1$ for all $1\leq i\leq N$, and another sequence of sets $\ppp{\calG_{i,2}}_{i=1}^N,\ppp{\calG_{i,3}}_{i=1}^N,\ldots,\ppp{\calG_{i,K}}_{i=1}^N$, where $\calG_{i,j}\subseteq\calV_{j}$, for $2\leq j\leq K$ and for all $1\leq i\leq N$. Define
\begin{align}
\calB_{m,1}\defi&\left\{v_1:\;v_1\in\calA_{m,0},\;\bigcap_{j=1}^{K-1}\p{v_1,v_{j+1}}\in\calA_{m,j},\right.\nonumber\\
&\left.\ \ \ \ \bigcap_{j=2}^{K}v_j\in\calG_{m,j}\ \ \text{for some }\ppp{v_j}_{j=2}^K\right\},\label{calB1df0}
\end{align}
and
\begin{align}
\calB_{m,2}&\defi\left\{\ppp{v_j}_{j=2}^K:\;v_1\in\calA_{m,0},\;\bigcap_{j=1}^{K-1}\p{v_1,v_{j+1}}\in\calA_{m,j},\right.\nonumber\\
&\left.\ \ \ \ \bigcap_{j=2}^{K}v_j\in\calG_{m,j}\ \ \text{for some }v_1\right\},\label{calB2df0}
\end{align}
for $m=1,2,\ldots,N$. Then, a general upper bound is given in \eqref{EquivErrorEvent01}, shown at the top of the next page,
\begin{figure*}[!t]
\normalsize
\setcounter{MYtempeqncnt}{\value{equation}}
\setcounter{equation}{37}
\begin{align}
&\Pr\ppp{\bigcup_{i}\ppp{\bigcup_{m=1}^N\ppp{V_1(i)\in\calA_{m,0},\;\bigcap_{k=1}^{K-1}\p{V_1(i),V_{k+1}}\in\calA_{m,k},\;\bigcap_{k=2}^{K}V_k\in\calG_{m,k}}}}\nonumber\\
&\ \ \ \ \ \ \ \ \ \ \ \ \ \ \ \ \ \leq\min\left\{1,L_1\Pr\ppp{\bigcup_{m=1}^N\ppp{V_1\in\calB_{m,1}}},\Pr\ppp{\bigcup_{m=1}^N\ppp{\ppp{V_j}_{k=2}^K\in\calB_{m,2}}},\right.\nonumber\\
&\left. \ \ \ \ \ \ \ \ \ \ \ \ \ \ \ \ \ \ \ \ L_1\Pr\ppp{\bigcup_{m=1}^N\ppp{V_1\in\calA_{m,0},\;\bigcap_{k=1}^{K-1}\p{V_1,V_{k+1}}\in\calA_{m,k},\;\bigcap_{k=2}^{K}V_k\in\calG_{m,k}}}\right\}\label{EquivErrorEvent01}
\end{align}
\hrulefill
\vspace*{4pt}
\end{figure*}
with $\p{V_1,\ldots,V_K}\sim P_{V_1}\cdots\times P_{V_K}$.
\end{lemma}

Next, we apply Lemma \ref{lem:Union3} to the problem at hand. To this end, we choose the following parameters in accordance to the notations used in Lemma \ref{lem:Union3}. Recall that we deal with
\begin{align}
\bigcup_{i=1}^{M_1-1}\bigcup_{\blt}\left\{ \begin{array}{ccc}
\bX_{1,i}\in \calA_{\blt,0}, \\
\p{\bX_{1,i},\bX_{2,j}} \in \calA_{\blt,j},\ \text{for }j=1,\ldots,M_2-1, \\
\bX_{2,j}\in \tilde{\calA}_{\blt,j},\ \ \ \text{for }j=1,\ldots,M_2-1 \end{array} \right\}\label{EquivErrorEvent2}
\end{align}
and in Lemma \ref{lem:Union3} we have considered:
\begin{align}
\bigcup_{i=1}^{L_1}\bigcup_{m=1}^N\left\{ \begin{array}{ccc}
V_1(i)\in\calA_{m,0},\\
\p{V_1(i),V_{j+1}}\in\calA_{m,j},\ \text{for }j=1,\ldots,K-1 \\
V_2\in\calG_{m,2},\ldots,V_K\in\calG_{m,K} \end{array} \right\}.\label{EquivErrorEvent2comp}
\end{align}
Thus, comparing \eqref{EquivErrorEvent2} and \eqref{EquivErrorEvent2comp}, we readily notice the following parallels: 
\begin{itemize} 
\item The numbers of events in the unions over $i$ is $L_1 = M_1-1$. Also, we have $K = M_2$ independent random vectors $V_1\p{i} = \bX_{1,i}$ and $V_{l+1} = \bX_{2,l}$, for $1\leq i\leq M_1-1$ and $1\leq l\leq M_2-1$.
\item The union over $m$ corresponds to a union over $\blb$, which as was mentioned before, is actually a union over a vector of types. Accordingly, we have:
\begin{enumerate}
\item $\calA_{m,i} = \calA_{\blt,i},\ \text{for }0\leq i\leq M_2-1$,
\item $\calG_{m,i} = \tilde{\calA}_{\blt,i-1},\ \text{for }2\leq i\leq M_2$.
\end{enumerate}
These sets correspond to each of the typical sets $ T(\tilde{Q}^0_{X_1X_{2,0}Y})$, $\{ T(\tilde{Q}^k_{X_1X_2Y})\}_{k=1}^{M_2-1}$, and $\{T(\hat{Q}^k_{X_{1,0}X_2Y})\}_{k=1}^{M_2-1}$.
\item According to \eqref{calB1df0} and \eqref{calB2df0} we need to define $\calB_{m,1} = \calB_1(\tilde{Q}^0_{X_1X_{2,0}Y},\{\tilde{Q}^j_{X_1X_2Y},\hat{Q}^j_{X_1X_2Y}\}_j)$ and $\calB_{m,2} = \calB_2(\tilde{Q}^0_{X_1X_{2,0}Y},\{\tilde{Q}^j_{X_1X_2Y},\hat{Q}^j_{X_1X_2Y}\}_j)$.
Using \eqref{calB1df0} and \eqref{calB2df0}, we get \eqref{B1setLem} and \eqref{bl2set}, given at the top of the next page.
\begin{figure*}[!t]
\normalsize
\setcounter{MYtempeqncnt}{\value{equation}}
\setcounter{equation}{40}
\begin{align}
\calB_{m,1}=
\left\{ \begin{array}{ccc}
&\p{\bx_1,\bx_{2,0},\by}\in T(\tilde{Q}^0_{X_1X_{2,0}Y}), \\
\bx_1:&\ppp{\p{\bx_1,\bx_{2,j},\by}\in T(\tilde{Q}^j_{X_1X_2Y})}_{j=1}^{M_2-1}, \\
&\ppp{\p{\bx_{1,0},\bx_{2,j},\by}\in T(\hat{Q}^j_{X_{1,0}X_2Y})}_{j=1}^{M_2-1}\ \text{for some }\ppp{\bx_{2,j}}_j \end{array} \right\}\label{B1setLem}\\
\calB_{m,2}=
\left\{ \begin{array}{ccc}
&\p{\bx_1,\bx_{2,0},\by}\in T(\tilde{Q}^0_{X_1X_{2,0}Y}), \\
\ppp{\bx_{2,j}}_{j\geq1}:&\ppp{\p{\bx_1,\bx_{2,j},\by}\in T(\tilde{Q}^j_{X_1X_2Y})}_{j=1}^{M_2-1}, \\
&\ppp{\p{\bx_{1,0},\bx_{2,j},\by}\in T(\hat{Q}^j_{X_{1,0}X_2Y})}_{j=1}^{M_2-1}\ \text{for some }\bx_1 \end{array} \right\}\label{bl2set}
\end{align}
\hrulefill
\vspace*{4pt}
\end{figure*}
\end{itemize} 

Thus, invoking Lemma \ref{lem:Union3}, we have \eqref{Last3Prob}, 
\begin{figure*}[!t]
\normalsize
\setcounter{MYtempeqncnt}{\value{equation}}
\setcounter{equation}{42}
\begin{align}
\tilde{P}_{e,1}^{(n)} &\triangleq \left.\Pr\pp{\bigcup_{i=1}^{M_1-1}\ppp{\sum_{j=0}^{M_2-1}P\p{\bY\vert \bX_{1,i},\bX_{2,j}}\geq\sum_{j=0}^{M_2-1}P\p{\bY\vert \bX_{1,0},\bX_{2,j}}}\right\vert\calF_0}\\
&\leq \min\left\{1, M_1\cdot\Pr\pp{\bigcup_{\ppp{\tilde{Q}_{X_1X_2Y}^j,\hat{Q}_{X_1X_2Y}^j}_j\in T_I(Q_{X_{1,0}X_{2,0}Y})}\bX_{1,1}\in\calB_1\p{\tilde{Q}^0_{X_1X_{2,0}Y},(\tilde{Q}^j_{X_1X_2Y},\hat{Q}^j_{X_1X_2Y})_j)}},\right. \nonumber\\
&\left.\ \ \ \ \ \Pr\pp{\bigcup_{\ppp{\tilde{Q}_{X_1X_2Y}^j,\hat{Q}_{X_1X_2Y}^j}_j\in T_I(Q_{X_{1,0}X_{2,0}Y})}\ppp{\bX_{2,j}}_{j\geq1}\in\calB_2\p{\tilde{Q}^0_{X_1X_{2,0}Y},(\tilde{Q}^j_{X_1X_2Y},\hat{Q}^j_{X_1X_2Y})_j)}},\nonumber\right.\\
&\left.\hspace{1cm}M_1\cdot\Pr\pp{\bigcup_{\ppp{\tilde{Q}_{X_1X_2Y}^j,\hat{Q}_{X_1X_2Y}^j}_j\in T_I(Q_{X_{1,0}X_{2,0}Y})}\left\{ \begin{array}{ccc}
\p{\bX_{1,1},\bx_{2,0},\by}\in T(\tilde{Q}^0_{X_1X_{2,0}Y}), \\
\ppp{\p{\bX_{1,1},\bX_{2,j},\by}\in T(\tilde{Q}^j_{X_1X_2Y})}_{j=1}^{M_2-1}, \\
\ppp{\p{\bx_{1,0},\bX_{2,j},\by}\in T(\hat{Q}^j_{X_{1,0}X_2Y})}_{j=1}^{M_2-1} \end{array} \right\}}\right\}\label{Last3Prob}
\end{align}
\hrulefill
\vspace*{4pt}
\end{figure*}
where each of the probabilities at the r.h.s. of \eqref{Last3Prob} are conditioned on $\calF_0$. Therefore, we were able to simplify the problematic union over the codebook of the first user. Note, however, that we cannot (directly) apply here the method of types due to the fact that the union is over an exponential number of types, and thus a more refined analysis is needed. We start by analyzing the last term at the r.h.s. of \eqref{Last3Prob}. To this end, we will invoke the type enumeration method, but first, the main observation here is that similarly to the passage from \eqref{EquivErrorEvent00} to \eqref{EquivErrorEvent}, the last term at the r.h.s. of \eqref{Last3Prob} can be rewritten as \eqref{lastProRight}, shown at the next page. 
\begin{figure*}[!t]
\normalsize
\setcounter{MYtempeqncnt}{\value{equation}}
\setcounter{equation}{44}
\begin{align}
&\Pr\pp{\bigcup_{\ppp{\tilde{Q}_{X_1X_2Y}^j,\hat{Q}_{X_1X_2Y}^j}_j\in T_I(Q_{X_{1,0}X_{2,0}Y})}\left\{ \begin{array}{ccc}
\p{\bX_{1,1},\bx_{2,0},\by}\in T(\tilde{Q}^0_{X_1X_{2,0}Y}), \\
\ppp{\p{\bX_{1,1},\bX_{2,j},\by}\in T(\tilde{Q}^j_{X_1X_2Y})}_{j=1}^{M_2-1}, \\
\ppp{\p{\bx_{1,0},\bX_{2,j},\by}\in T(\hat{Q}^j_{X_{1,0}X_2Y})}_{j=1}^{M_2-1} \end{array} \right\}} \nonumber\\
&\ \ \ \ \ \ \ \ \ \ \ \ =\Pr\pp{\left.\ppp{\sum_{j=0}^{M_2-1}P\p{\bY\vert \bX_{1,1},\bX_{2,j}}\geq\sum_{j=0}^{M_2-1}P\p{\bY\vert \bX_{1,0},\bX_{2,j}}}\right\vert\calF_0}\\
&\ \ \ \ \ \ \ \ \ \ \ \ =\bE\ppp{\left.\Pr\pp{\left.\ppp{\sum_{j=0}^{M_2-1}P\p{\bY\vert \bX_{1,1},\bX_{2,j}}\geq\sum_{j=0}^{M_2-1}P\p{\bY\vert \bX_{1,0},\bX_{2,j}}}\right\vert\calF_0,\bX_{1,1}}\right\vert\calF_0}.\label{lastProRight}
\end{align}
\hrulefill
\vspace*{4pt}
\end{figure*}
That is, we returned back to the structure of the original probability in \eqref{innerPro}, but now, without the union over the codebook of the first user. Note that the conditioning on the random vector $\bX_{1,1}$ in \eqref{lastProRight}, is due to the fact that $\bX_{1,1}$ is common to all the summands in the inner summation over the codebook of the second user. We next evaluate the exponential behavior of the probability in \eqref{lastProRight}. For a given realization of $\bY = \by$, $\bX_{1,0}= \bx_{1,0}$, $\bX_{1,1}= \bx_{1,1}$, and $\bX_{2,0}= \bx_{2,0}$, let us define 
\begin{align}
s\triangleq \frac{1}{n}\log P\p{\by\vert \bx_{1,0},\bx_{2,0}},\label{sdefinition}
\end{align}
and
\begin{align}
r\triangleq \frac{1}{n}\log P\p{\by\vert \bx_{1,1},\bx_{2,0}}.\label{rdefinition}
\end{align}
For a given $(\by,\bx_{1,0},\bx_{1,1},\bx_{2,0})$, and a given joint probability distribution $Q_{X_1X_2Y}$ on $\calX_1\times\calX_2\times\calY$, let $N_1\p{Q_{X_1X_2Y}}$ designate the number of codewords $\ppp{\bX_{2,j}}_j$ (excluding $\bx_{2,0}$) whose conditional empirical distribution with $\by$ and $\bx_{1,1}$ is $Q_{X_1X_2Y}$, that is,
\begin{align}
N_1\p{Q_{X_1X_2Y}}\triangleq\sum_{j=1}^{M_2-1}\calI\ppp{\p{\bx_{1,1},\bX_{2,j},\by}\in T\p{Q_{X_1X_2Y}}},\label{N1QDef}
\end{align}
and let $N_2\p{Q_{X_1X_2Y}}$ designate the number of codewords $\ppp{\bX_{2,j}}_j$ (excluding $\bx_{2,0}$) whose conditional empirical distribution with $\by$ and $\bx_{1,0}$ is $Q_{X_1X_2Y}$, that is
\begin{align}
N_2\p{Q_{X_1X_2Y}}\triangleq\sum_{j=1}^{M_2-1}\calI\ppp{\p{\bx_{1,0},\bX_{2,j},\by}\in T\p{Q_{X_1X_2Y}}}.\label{N2QDef}
\end{align}
Also, recall that
\begin{align}
f\p{Q_{X_1X_2Y}} &= \frac{1}{n}\log P\p{\by\vert\bx_1,\bx_2}\\
& = \bE_Q\pp{\log W_{Y\vert X_1X_2}\p{Y\vert X_1,X_2}}\label{fQdef}
\end{align}
where $Q_{X_1X_2Y}$ is understood to be the joint empirical distribution of $\p{\bx_1,\bx_2,\by}\in\calX_1^n\times\calX_2^n\times\calY^n$. Thus, in terms of the above notations, we may write:
\begin{align}
&\sum_{j=0}^{M_2-1}P\p{\by\vert \bx_{1,1},\bX_{2,j}} = e^{nr}\nonumber\\
&\ \ \ \ +\sum_{Q_{X_2\vert X_1Y}\in\calS(Q_{X_1Y})}N_1\p{Q_{X_1X_2Y}}e^{nf(Q_{X_1X_2Y})}\\
&\hspace{3.3cm}\triangleq e^{nr}+\calN_1(Q_{X_1Y}).\label{calN1QDef}
\end{align}
where for a given $Q_{X_1Y}$, $\calS(Q_{X_1Y})$ is defined as the set of all distributions $\ppp{Q_{X_2\vert X_1Y}}$, such that $\sum_{\p{x_1,y}\in\calX_1\times\calY}Q_{X_1Y}\p{x_1,y}Q_{X_2\vert X_1Y}\p{x_2\vert x_1,y} = P_{X_2}\p{x_2}$ for all $x_2\in\calX_2$, namely,
\begin{align}
\calS(Q_{X_1Y}) = \ppp{Q_{X_1X_2Y}':\;Q_{X_1Y}' = Q_{X_1Y},\; Q_{X_2}' = P_{X_2}}.\label{Sdef}
\end{align}
Similarly,
\begin{align}
&\sum_{j=0}^{M_2-1}P\p{\by\vert \bx_{1,0},\bX_{2,j}} = e^{ns}\nonumber\\
&+\sum_{Q_{X_2\vert X_{1,0}Y}\in\calS(Q_{X_{1,0}Y})}N_2\p{Q_{X_{1,0}X_2Y}}e^{nf(Q_{X_{1,0}X_2Y})}\\
&\hspace{3.3cm}\triangleq e^{ns}+\calN_2(Q_{X_{1,0}Y}).\label{calN2QDef}
\end{align}
where for a given $Q_{X_{1,0}Y}$, $\calS(Q_{X_{1,0}Y})$ is defined as the set of all distributions $\ppp{Q_{X_2\vert X_{1,0}Y}}$, such that $\sum_{\p{x_1,y}\in\calX_1\times\calY}Q_{X_{1,0}Y}\p{x_1,y}Q_{X_2\vert X_{1,0}Y}\p{x_2\vert x_1,y} = P_{X_2}\p{x_2}$ for all $x_2\in\calX_2$ (similarly as in \eqref{Sdef}). For simplicity of notation, in the following, we use $Q$ and $\tilde{Q}$ to denote $Q_{X_1X_2Y}$ and $Q_{X_{1,0}X_2Y}$, respectively. Therefore, with these definitions in mind, we wish to calculate (given $\p{\calF_0,\bX_{1,1}}$)
\begin{align}
&\Pr\pp{\sum_{j=0}^{M_2-1}P\p{\bY\vert \bx_{1,1},\bX_{2,j}}\geq\sum_{j=0}^{M_2-1}P\p{\bY\vert \bx_{1,0},\bX_{2,j}}}\nonumber\\
&=\Pr\left[\calN_1(Q_{X_1Y})-\calN_2(Q_{X_{1,0}Y})\geq e^{ns}-e^{nr} \right]\label{anaBef}
\end{align}
where $s$, $r$, $\calN_1(Q)$ and $\calN_2(Q)$ are given in \eqref{sdefinition}, \eqref{rdefinition}, \eqref{calN1QDef}, and \eqref{calN2QDef}, respectively. Let $\varepsilon>0$ be arbitrarily small, and define $i_1\triangleq\left\lfloor \frac{1}{n\epsilon}\log P(\by\vert \bx_{1,0},\bx_{2,0})\right\rfloor$. Then,
\begin{align}
&\Pr\left[\calN_1(Q_{X_1Y})-\calN_2(Q_{X_{1,0}Y})\geq e^{ns}-e^{nr} \right] \nonumber\\
& = \sum_{i=i_1}^{\left\lceil R_2/\varepsilon\right\rceil} \Pr\left\{e^{ni\varepsilon}\leq \calN_2(Q_{X_{1,0}Y})\leq e^{n(i+1)\varepsilon},\right.\nonumber\\
&\left.\hspace{2cm}\vphantom{e^{ni\varepsilon}\leq \calN_2(Q_{X_{1,0}Y})\leq e^{n(i+1)\varepsilon}}\calN_1(Q_{X_1Y})-\calN_2(Q_{X_{1,0}Y})\geq e^{ns}-e^{nr}\right\}\nonumber\\
& \leq \sum_{i=i_1}^{\left\lceil R_2/\varepsilon\right\rceil} \Pr\left\{e^{ni\varepsilon}\leq \calN_2(Q_{X_{1,0}Y})\leq e^{n\p{i+1}\varepsilon},\right.\nonumber\\
&\left.\hspace{2cm}\vphantom{e^{ni\varepsilon}\leq \calN_2(Q_{X_{1,0}Y})\leq e^{n(i+1)\varepsilon}}\calN_1(Q_{X_1Y})\geq e^{ni\varepsilon}+e^{ns}-e^{nr}\right\}\\
& = \sum_{i=i_1}^{\left\lceil R_2/\varepsilon\right\rceil} \Pr\left\{e^{ni\varepsilon}\leq \calN_2(Q_{X_{1,0}Y})\leq e^{n(i+1)\varepsilon}\right\}\nonumber\\
&\ \ \ \ \ \ \ \ \ \times\Pr\left\{\calN_1(Q_{X_1Y})\geq e^{ni\varepsilon}+e^{ns}-e^{nr}\right.\nonumber\\
&\ \ \ \ \ \ \ \ \ \ \ \ \ \ \ \ \ \ \left.\left.\right\vert e^{ni\varepsilon}\leq \calN_2(Q_{X_{1,0}Y})\leq e^{n(i+1)\varepsilon}\right\}.
\label{lastExposs}
\end{align}
It is not difficult to show that (see, e.g., \cite[p. 6028]{MerhavEr1})
\begin{align}
&\Pr\left\{e^{nt}\leq \calN_2(Q_{X_{1,0}Y})\leq e^{n(t+\varepsilon)}\right\}\nonumber\\
&\exe \begin{cases}
0\ \ \ \ \ \ \ \ \ \ \ \ \ \ \ \ \ \ \ \ \ \ \ \ \ \ \ \ \ \ \ \ \ \ \ \ t<t_0(Q_{X_{1,0}Y})-\varepsilon\\
1\ \ \ \ \ \ \ \ \ \ \ \ \ \ \ \ \ \ \ \ \ \ \ \ t_0(Q_{X_{1,0}Y})-\varepsilon\leq t\leq t_0(Q_{X_{1,0}Y})\\
\exp\pp{-nE(t,Q_{X_{1,0}Y})}\ \ \ \ \ \ \ \ \ \ t> t_0(Q_{X_{1,0}Y})
\end{cases}\label{larggedeivupperlower}
\end{align}
where
\begin{align}
&t_0(Q_{X_{1,0}Y})\triangleq R_2+\nonumber\\
&\max_{\tilde{Q}\in\calS(Q_{X_{1,0}Y}):\;I_{\tilde{Q}}(X_2;X_{1,0},Y)\leq R_2}\pp{f(\tilde{Q})-I_{\tilde{Q}}(X_2;X_{1,0},Y)},
\end{align}
in which $\calS(Q)$ is defined in \eqref{Sdef}, $f(Q)$ is given in \eqref{fQdef}, and
\begin{align}
&E(t,Q_{X_{1,0}Y})\triangleq \min\left\{\pp{I_{\tilde{Q}}(X_2;X_{1,0},Y)-R_2}_+:\right.\nonumber\\
&\left.\ \ \ \ \ \ \ \ \ \ \ \ \ \ f(\tilde{Q})+\pp{R_2-I_{\tilde{Q}}(X_2;X_{1,0},Y)}_+\geq t\right\}.
\end{align}
Now, in the exponential scale, the term at the r.h.s. of \eqref{lastExposs} is dominated by one of the summands, and we claim that the dominant contribution to the sum over $i$ is due to the first term\footnote{Note that according to \eqref{larggedeivupperlower}, $\Pr\{e^{ni\varepsilon}\leq \calN_2(Q_{X_{1,0}Y})\leq e^{n(i+1)\varepsilon}\}$ vanishes (in the exponential scale) for $i<t_0(Q_{X_{1,0}Y})/\varepsilon$. Thus, to asses the exponential scale of \eqref{lastExposs} we consider only the indices correspond to $i\geq t_0(Q_{X_{1,0}Y})/\varepsilon$.}, $i = t_0(Q_{X_{1,0}Y})/\varepsilon$. Indeed, let $\mathscr{A}_k \triangleq \ppp{e^{nk\varepsilon}\leq \calN_2(Q_{X_{1,0}Y})\leq e^{n\p{k+1}\varepsilon}}$ and $\mathscr{B}_k \triangleq \ppp{\calN_1(Q_{X_1Y})\geq e^{nk\varepsilon}+e^{ns}-e^{nr}}$, and notice that the summands in \eqref{lastExposs} correspond to $\Pr\ppp{\mathscr{A}_k\cap\mathscr{B}_k}$. According to \eqref{larggedeivupperlower}, $\Pr\ppp{\mathscr{A}_{t_0}}\to1$ (the exponent $E(k\varepsilon,Q_{X_{1,0}Y})$ vanishes), and note that $\Pr\ppp{\mathscr{B}_k}$ is monotonically decreasing with $k$. Therefore,
\begin{align}
\Pr\ppp{\mathscr{A}_{t_0}\cap\mathscr{B}_{t_0}}&\leq\max_{k\geq t_0}\Pr\ppp{\mathscr{A}_k\cap\mathscr{B}_k}\nonumber\\
&\leq \max_{k\geq t_0}\Pr\ppp{\mathscr{B}_k} = \Pr\ppp{\mathscr{B}_{t_0}}.\label{ineqineq0}
\end{align}
On the other hand,
\begin{align}
\Pr\ppp{\mathscr{A}_{t_0}\cap\mathscr{B}_{t_0}}&= \Pr\ppp{\mathscr{B}_{t_0}}-\Pr\ppp{\mathscr{A}_{t_0}^c\cap\mathscr{B}_{t_0}}\nonumber\\
&\geq \Pr\ppp{\mathscr{B}_{t_0}}-\Pr\ppp{\mathscr{A}_{t_0}^c}.
\end{align}
Thus, due to the fact that $\Pr\ppp{\mathscr{A}_{t_0}}\to1$ super-exponentially fast \cite[p. 6028]{MerhavEr1}, we may conclude that 
\begin{align}
\Pr\ppp{\mathscr{A}_{t_0}\cap\mathscr{B}_{t_0}}\exe \Pr\ppp{\mathscr{B}_{t_0}}.\label{ineqineq1}
\end{align}
Combining \eqref{lastExposs}, \eqref{ineqineq0} and \eqref{ineqineq1}, and the fact that $\varepsilon$ is arbitrarily small, we get \eqref{lastExposs3} (shown at the top of the next page) by using standard large deviations techniques (see, e.g., \cite[p. 6027]{MerhavEr1}), 
\begin{figure*}[!t]
\normalsize
\setcounter{MYtempeqncnt}{\value{equation}}
\setcounter{equation}{66}
\begin{align}
\Pr\left[\calN_1(Q_{X_1Y})-\calN_2(Q_{X_{1,0}Y})\geq e^{ns}-e^{nr} \right] &\exe \Pr\ppp{\mathscr{B}_{t_0}}\nonumber\\
&=\Pr\left\{\calN_1(Q_{X_1Y})\geq e^{nt_0(Q_{X_{1,0}Y})}+e^{ns}-e^{nr}\right\}\nonumber\\
&\exe\max_{Q\in\calS(Q_{X_1Y})}\Pr\ppp{N_1(Q)\geq e^{n\pp{t_0(Q_{X_{1,0}Y})-f(Q)}}+e^{n\pp{s-f(Q)}}-e^{n\pp{r-f(Q)}}}\label{rrrss}\\
&\exe\max_{Q\in\calS(Q_{X_1Y})}\begin{cases}
1\ &r>\max\pp{t_0,s}\\
e^{-n\pp{I_Q(X_2;X_1,Y)-R_2}_+}\ &r\leq \max\pp{t_0,s},\;Q\in\tilde\calL\\
0 \ &r\leq \max\pp{t_0,s},\;Q\in\tilde\calL^c
\end{cases}\label{lastExposs3}
\end{align}
\hrulefill
\vspace*{4pt}
\end{figure*}
where $N_1(Q)$ and $\calS(Q)$ are defined in \eqref{N1QDef} and \eqref{Sdef}, respectively, and
\begin{align}
\tilde\calL \triangleq \ppp{Q:\;\max\pp{t_0,s}-f(Q)\leq \pp{R_2-I_Q(X_2;X_1,Y)}_+}.\label{tildeLdef}
\end{align}
Thus,
\begin{align}
&\Pr\left[\calN_1(Q_{X_1Y})-\calN_2(Q_{X_{1,0}Y})\geq e^{ns}-e^{nr} \right]\nonumber\\
&\ \ \ \ \ \ \ \ \exe \exp\ppp{-nE_1(Q_{X_1X_{2,0}Y},Q_{X_{1,0}X_{2,0}Y})}\label{E1befHat}
\end{align}
where $E_1(\cdot,\cdot)$ is defined in \eqref{E1QtileQdef}. Note that when $r>\max\pp{t_0,s}$, the r.h.s. term of the inequality in the probability in \eqref{rrrss} is negative, and due to the fact that the enumerator is nonnegative, the overall probability is unity. Finally, we average over $\bX_{1,1}$ given $\calF_0$. Using the method of types, we readily obtain \eqref{toppageEq}, given at the top of the next page, 
\begin{figure*}[!t]
\normalsize
\setcounter{MYtempeqncnt}{\value{equation}}
\setcounter{equation}{70}
\begin{align}
&\bE\ppp{\left.\Pr\pp{\left.\ppp{\sum_{j=0}^{M_2-1}P\p{\bY\vert \bX_{1,1},\bX_{2,j}}\geq\sum_{j=0}^{M_2-1}P\p{\bY\vert \bX_{1,0},\bX_{2,j}}}\right\vert\calF_0,\bX_{1,1}}\right\vert\calF_0}\label{sameFormRCE}\\
&\exe \exp\ppp{-n\min_{Q_{X_1\vert X_{2,0}Y}\in\hat{\calS}(Q_{X_{2,0}Y})}\pp{I_Q(X_1;X_{2,0},Y)+E_1(Q_{X_1X_{2,0}Y},Q_{X_{1,0}X_{2,0}Y})}}\label{methodoftypes}\\
&\triangleq \exp\ppp{-n\hat{E}_1(Q_{X_{1,0}X_{2,0}Y},R_2)}\label{toppageEq}
\end{align}
\hrulefill
\vspace*{4pt}
\end{figure*}
where 
\begin{align}
\hat{\calS}(Q_{X_2Y}) \triangleq \ppp{Q_{X_1X_2Y}':\;Q_{X_2Y}' = Q_{X_2Y},\; Q_{X_1}' = P_{X_1}},\label{hatSdef}
\end{align}
and $\hat{E}_1(\cdot,\cdot)$ is defined in \eqref{E1hatQdef}. This completes the analysis of the last term at the r.h.s. of \eqref{Last3Prob}. 

Next, we analyze the second and third terms at the r.h.s. of \eqref{Last3Prob}. Recall that the latter is given by \eqref{weirdposs}.
\begin{figure*}[!t]
\normalsize
\setcounter{MYtempeqncnt}{\value{equation}}
\setcounter{equation}{74}
\begin{align}
P_{e,3}\triangleq\Pr\pp{\bigcup_{\ppp{\tilde{Q}_{X_1X_2Y}^j,\hat{Q}_{X_1X_2Y}^j}_j\in T_I(Q_{X_{1,0}X_{2,0}Y})}\ppp{\bX_{2,j}}_{j\geq1}\in\calB_2\p{\tilde{Q}^0_{X_1X_{2,0}Y},(\tilde{Q}^j_{X_1X_2Y},\hat{Q}^j_{X_1X_2Y})_j)}}.\label{weirdposs}
\end{align}
\hrulefill
\vspace*{4pt}
\end{figure*}
Accordingly, in the spirit of \eqref{lastProRight}, we note that $P_{e,3}$ can be equivalently rewritten as \eqref{lastProRight2}, shown at the top of the next page, 
\begin{figure*}[!t]
\normalsize
\setcounter{MYtempeqncnt}{\value{equation}}
\setcounter{equation}{75}
\begin{align}
P_{e,3} &= \Pr\left[\bigcup_{Q_{X_1\vert X_{2,0}Y}}P\p{\by\vert \bx_{1,1},\bx_{2,0}}+\sum_{j=1}^{M_2-1}P\p{\by\vert \bx_{1,1},\bX_{2,j}}\geq P\p{\by\vert \bx_{1,0},\bx_{2,0}}\right.\nonumber\\
&\left.\left.\ \ \ \ \ \ \ \ \ \ \ \ \ \ \ \ \ +\sum_{j=1}^{M_2-1}P\p{\by\vert \bx_{1,0},\bX_{2,j}},\;\text{for some }\bx_{1,1}\in T(Q_{X_1X_{2,0}Y})\right\vert\calF_0\right],\\
&\exe \max_{Q_{X_1\vert X_{2,0}Y}\in\hat{\calS}(Q_{X_{2,0}Y})} \Pr\left[P\p{\by\vert \bx_{1,1},\bx_{2,0}}+\sum_{j=1}^{M_2-1}P\p{\by\vert \bx_{1,1},\bX_{2,j}}\geq P\p{\by\vert \bx_{1,0},\bx_{2,0}}\right.\nonumber\\
&\left.\left.\ \ \ \ \ \ \ \ \ \ \ \ \ \ \ \ \ \ \ \ \ \ \ \ \ \ \ \ \ \ \ +\sum_{j=1}^{M_2-1}P\p{\by\vert \bx_{1,0},\bX_{2,j}},\;\text{for some }\bx_{1,1}\in T(Q_{X_1X_{2,0}Y})\right\vert\calF_0\right]\label{lastProRight2mid}\\
& = \max_{Q_{X_1\vert X_{2,0}Y}\in\hat{\calS}(Q_{X_{2,0}Y})} \Pr\left[\calN_1(Q_{X_1Y})-\calN_2(Q_{X_{1,0}Y})\geq e^{ns}-e^{nr},\;\text{for some }\bx_{1,1}\in T(Q_{X_1X_{2,0}Y})\vert\calF_0\right]\label{lastProRight2}
\end{align}
\hrulefill
\vspace*{4pt}
\end{figure*}
where $s$, $r$, $\calN_1(Q)$, $\calN_2(Q)$, and $\hat{\calS}(Q)$, are given in \eqref{sdefinition}, \eqref{rdefinition}, \eqref{calN1QDef}, \eqref{calN2QDef}, and \eqref{hatSdef}, respectively, and the second passage follows by using the method of types. 
Now, due to the fact that only $\calN_1$ (and not $\calN_2$) in \eqref{lastProRight2} depends on $\bx_{1,1}$, and since the analysis in \eqref{anaBef}-\eqref{rrrss} is independent of $\bx_{1,1}$, it can be repeated here, and we obtain \eqref{lastProRight2after}, shown at the top of page~\pageref{lastProRight2after},
\begin{figure*}[!t]
\normalsize
\setcounter{MYtempeqncnt}{\value{equation}}
\setcounter{equation}{78}
\begin{align}
P_{e,3}&\exe \max_{Q_{X_1\vert X_{2,0}Y}\in\hat{\calS}(Q_{X_{2,0}Y})}\sum_{i=i_1}^{\left\lceil R_2/\varepsilon\right\rceil} \Pr\left\{e^{ni\varepsilon}\leq \calN_2(Q_{X_{1,0}Y})\leq e^{n(i+1)\varepsilon}\right\}\nonumber\\
&\ \ \ \ \ \ \ \ \ \ \ \ \ \ \ \ \ \ \ \ \ \ \ \ \ \ \ \times\Pr\left\{\left.\bigcup_{\bxt_{1,1}}\calN_1(Q_{X_1Y})\geq e^{ni\varepsilon}+e^{ns}-e^{nr}\right\vert e^{ni\varepsilon}\leq \calN_2(Q_{X_{1,0}Y})\leq e^{n(i+1)\varepsilon}\right\}\nonumber\\
&\exe\max_{Q_{X_1\vert X_{2,0}Y}\in\hat{\calS}(Q_{X_{2,0}Y})} \Pr\left\{\calN_1(Q_{X_1Y})\geq e^{nt_0(Q_{X_{1,0}Y})}+e^{ns}-e^{nr},\;\text{for some }\bx_{1,1}\right\}\nonumber\\
&\exe \max_{Q_{X_1\vert X_{2,0}Y}\in\hat{\calS}(Q_{X_{2,0}Y})}\max_{Q\in\calS(Q_{X_1Y})}\Pr\left\{N_1(Q)\geq e^{n\gamma},\;\text{for some }\bx_{1,1}\right\},\label{lastProRight2after}
\end{align}
\hrulefill
\vspace*{4pt}
\end{figure*}
where $N_1(Q)$ is given in \eqref{N1QDef}, and we have defined $e^{n\gamma}\triangleq e^{n\pp{t_0(Q_{X_{1,0}Y})-f(Q)}}+e^{n\pp{s-f(Q)}}-e^{n\pp{r-f(Q)}}$. Recall \eqref{N1QDef}, and let $\tilde{N}_1(Q) \triangleq \sum_{j=1}^{M_2-1}\calI\ppp{(\bX_{2,j},\by)\in T(Q_{X_2Y})}$. We claim that \eqref{lastProRight2after} can be rewritten as\footnote{It is easy to see that \eqref{444} is an upper bound on \eqref{lastProRight2after}. The other direction follows from:
\begin{align}
&\max_{Q,\hat{Q}\in\calS(Q)}\Pr\ppp{\tilde{N}_1(\hat{Q})\geq e^{n\gamma(Q)}}= \max_{\hat{Q}\in\calS(Q^*)}\Pr\ppp{\tilde{N}_1(\hat{Q})\geq e^{n\gamma(Q^*)}}\nonumber\\
&\ \ \ \ \ =\max_{\hat{Q}\in\calS(Q^*)}\Pr\ppp{{N}_1(\hat{Q})\geq e^{n\gamma(Q^*)},\;\text{for some }\bx_{1,1}\in T(Q^*)}\nonumber\\
&\ \ \ \ \ \leq\max_Q\max_{\hat{Q}\in\calS(Q)}\Pr\ppp{{N}_1(\hat{Q})\geq e^{n\gamma(Q)},\;\text{for some }\bx_{1,1}\in T(Q)}\nonumber,
\end{align}
where in the first equality we designate $Q^*$ as the maximizer, and the second equality follows from the fact that $T(Q_{X_2Y})=\ppp{\bx_2:\;(\bx_{1,1},\bx_2,\by)\in T(Q_{X_1X_2Y}),\;\text{for some }\bx_{1,1}}$.
}
\begin{align}
P_{e,3}\exe \max_{Q_{X_1\vert X_{2,0}Y}\in\hat{\calS}(Q_{X_{2,0}Y})}\max_{Q\in\calS(Q_{X_1Y})}\Pr\left\{\tilde{N}_1(Q)\geq e^{n\gamma}\right\},\label{444}
\end{align}
which follows from the fact that the set $\ppp{\bx_2:\;(\bx_{1,1},\bx_2,\by)\in T(Q_{X_1X_2Y}),\;\text{for some }\bx_{1,1}}$ equals $T(Q_{X_2Y})$ (see, e.g., \cite[eqs. (24)-(25)]{scarletNew}). Thus, by using standard large deviations techniques (see, e.g., \cite[p. 6027]{MerhavEr1})
\begin{align}
P_{e,3}&\exe \max_{Q_{X_1\vert X_{2,0}Y}\in\hat{\calS}(Q_{X_{2,0}Y})}\max_{Q\in\calS(Q_{X_1Y})}\\
&\begin{cases}
1\ &r>\max\pp{t_0,s}\\
e^{-n\pp{I_Q(X_2;Y)-R_2}_+}\ &r\leq \max\pp{t_0,s},\;Q\in\hat\calL\\
0 \ &r\leq \max\pp{t_0,s},\;Q\in\hat\calL^c
\end{cases}
\end{align}
where
\begin{align}
\hat\calL \triangleq \ppp{Q:\;\max\pp{t_0,s}-f(Q)\leq \pp{R_2-I_Q(X_2;Y)}_+}.
\end{align}
Therefore,
\begin{align}
P_{e,3} &\exe \exp\ppp{-n\hat{E}_2(Q_{X_{1,0}X_{2,0}Y},R_2)}\label{lastProRight3ere}
\end{align}
where $\hat{E}_2(Q_{X_{1,0}X_{2,0}Y},R_2)$ is defined in \eqref{E2hatQdef}. This completes the analysis of the third term at the r.h.s. of \eqref{Last3Prob}. Finally, recall that the second term at the r.h.s. of \eqref{Last3Prob} is given by
\begin{align}
&A \triangleq M_1\cdot\Pr\left[\bigcup_{ T_I(Q_{X_{1,0}X_{2,0}Y})}\right.\nonumber\\
&\left.\vphantom{\bigcup_{ T_I(Q_{X_{1,0}X_{2,0}Y})}}\ \ \ \ \bX_{1,1}\in\calB_1\p{\tilde{Q}^0_{X_1X_{2,0}Y},(\tilde{Q}^j_{X_1X_2Y},\hat{Q}^j_{X_1X_2Y})_j)}\right]
\end{align}
and is equivalent to \eqref{lastProRightfin0}, given at the top of page~\pageref{lastProRightfin0}.
\begin{figure*}[!t]
\normalsize
\setcounter{MYtempeqncnt}{\value{equation}}
\setcounter{equation}{85}
\begin{align}
&A=M_1\cdot\Pr\pp{\bigcup_{ T_I(Q_{X_{1,0}X_{2,0}Y})}\left\{ \begin{array}{ccc}
\p{\bX_{1,1},\bx_{2,0},\by}\in T(\tilde{Q}^0_{X_1X_{2,0}Y}), \\
\ppp{\p{\bX_{1,1},\bx_{2,j},\by}\in T(\tilde{Q}^j_{X_1X_2Y})}_{j=1}^{M_2-1},\text{ for some }\ppp{\bx_{2,j}} \\
\ppp{\p{\bx_{1,0},\bx_{2,j},\by}\in T(\hat{Q}^j_{X_{1,0}X_2Y})}_{j=1}^{M_2-1} \end{array} \right\}}.
\label{lastProRightfin0}
\end{align}
\hrulefill
\vspace*{4pt}
\end{figure*}
This term can be analyzed as before, but, we claim that it is actually larger than the fourth term at the r.h.s. of \eqref{Last3Prob}, and thus, essentially, does not affect the minimum in \eqref{Last3Prob}. Indeed, recall that the fourth term is given by \eqref{lastProRightfin}, shown at the top of page~\pageref{lastProRightfin},
\begin{figure*}[!t]
\normalsize
\setcounter{MYtempeqncnt}{\value{equation}}
\setcounter{equation}{86}
\begin{align}
B \triangleq M_1\cdot\Pr\pp{\bigcup_{ T_I(Q_{X_{1,0}X_{2,0}Y})}\left\{ \begin{array}{ccc}
\p{\bX_{1,1},\bx_{2,0},\by}\in T(\tilde{Q}^0_{X_1X_{2,0}Y}), \\
\ppp{\p{\bX_{1,1},\bX_{2,j},\by}\in T(\tilde{Q}^j_{X_1X_2Y})}_{j=1}^{M_2-1}, \\
\ppp{\p{\bx_{1,0},\bX_{2,j},\by}\in T(\hat{Q}^j_{X_{1,0}X_2Y})}_{j=1}^{M_2-1} \end{array} \right\}},
\label{lastProRightfin}
\end{align}
\hrulefill
\vspace*{4pt}
\end{figure*}
and since the factor $M_1$ is common to both $A$ and $B$, we just need to compare the probabilities in these terms. However, it is obvious that the probability term in $B$ is smaller than the probability in $A$, due to the fact that events in the former are contained in the events in the latter. Indeed, this is equivalent to comparing $\Pr\ppp{(Z_1,Z_2)\in\calZ}$ and $\Pr\ppp{(Z_1,z_2)\in\calZ,\text{ for some }z_2\in\calZ_2}$, where $Z_1$ and $Z_2$ are RVs that are defined over the alphabets $\calZ_1$ and $\calZ_2$, respectively, and $\calZ\subseteq\calZ_1\times\calZ_2$. Let $\calV\triangleq\tilde{\calV}\times\calZ_2$, in which
\begin{align}
\tilde{\calV}\triangleq\ppp{z_1\in\calZ_1:\;(z_1,z_2)\in\calZ,\text{ for some }z_2\in\calZ_2}.
\end{align}
Then, it is obvious that $\calZ\subseteq\calV$, and thus
\begin{align}
&\Pr\ppp{(Z_1,Z_2)\in\calZ} = \sum_{(z_1,z_2)\in\calZ}P(z_1,z_2)\\
&\leq \sum_{(z_1,z_2)\in\calV}P(z_1,z_2)\\
&= \sum_{z_1\in\tilde{\calV}}P(z_1)= \Pr\ppp{(Z_1,z_2)\in\calZ,\text{ for some }z_2}.\label{biggerThan}
\end{align}   
Wrapping up, using \eqref{innerPro}, \eqref{Last3Prob}, and the last results, after averaging w.r.t. $\calF_0$, we get \eqref{proboferror}, shown at the top of the next page, as required.
\begin{figure*}[!t]
\normalsize
\setcounter{MYtempeqncnt}{\value{equation}}
\setcounter{equation}{92}
\begin{align}
\bar P^{(n)}_{e,1} &\exeq \bE\ppp{\min\ppp{1,e^{-n(\hat{E}_1(Q_{X_{1,0}X_{2,0}Y},R_2)-R_1)},e^{-n\hat{E}_2(Q_{X_{1,0}X_{2,0}Y},R_2)}}}\\
& = \bE\ppp{\min\ppp{e^{-n\pp{\hat{E}_1(Q_{X_{1,0}X_{2,0}Y},R_2)-R_1}_+},e^{-n\hat{E}_2(Q_{X_{1,0}X_{2,0}Y},R_2)}}}\\
& = \bE\ppp{\exp\pp{-n\max\ppp{\pp{\hat{E}_1(Q_{X_{1,0}X_{2,0}Y},R_2)-R_1}_+,\hat{E}_2(Q_{X_{1,0}X_{2,0}Y},R_2)}}}\\
&\exe \exp\ppp{-n\pp{\min_{Q_{Y|X_{1,0}X_{2,0}}}\pp{D(Q_{Y|X_{1,0}X_{2,0}}||W_{Y|X_{1,0}X_{2,0}}\vert P_{X_{1,0}}\times P_{X_{2,0}})+E(Q,R_1,R_2)}}}\label{proboferror}\\
E(Q,R_1,R_2) &\triangleq\max\left\{\pp{\hat{E}_1(Q_{X_{1,0}X_{2,0}Y},R_2)-R_1}_+,\hat{E}_2(Q_{X_{1,0}X_{2,0}Y},R_2)\right\}.
\end{align}
\hrulefill
\vspace*{4pt}
\end{figure*}

\subsection{Proof of Theorem \ref{th:2}}\label{subsec:proof2}
 
Without loss of generality, we assume throughout, that the transmitted codewords are $\bx_{1,0}$ and $\bx_{2,0}$ which correspond to $\bz_{11,0},\bz_{12,0},\bz_{21,0}$ and $\bz_{22,0}$. Here, we distinguish between several types of errors. Recall that the overall error probability is given by
\begin{align}
\bar P^{(n)}_{e,1} = \Pr\ppp{(\hat{m}_{11},\hat{m}_{12})\neq (0,0)},\label{overalProb}
\end{align}
which can be divided into six possible types of errors: $(\hat{m}_{11}\neq0,\hat{m}_{12}=0,\hat{m}_{21}=0)$, $(\hat{m}_{11}=0,\hat{m}_{12}\neq0,\hat{m}_{21}=0)$, $(\hat{m}_{11}\neq0,\hat{m}_{12}\neq0,\hat{m}_{21}=0)$, $(\hat{m}_{11}\neq0,\hat{m}_{12}=0,\hat{m}_{21}\neq0)$, $(\hat{m}_{11}=0,\hat{m}_{12}\neq0,\hat{m}_{21}\neq0)$, and  $(\hat{m}_{11}\neq0,\hat{m}_{12}\neq0,\hat{m}_{21}\neq0)$. Note that the event $(\hat{m}_{11}=0,\hat{m}_{12}=0,\hat{m}_{21}\neq0)$ will not result in an error, and thus ignored. Obviously, the exponent of the overall error probability in \eqref{overalProb} is given by the minimum between the error exponents corresponding to each type of error individually. We start with analyzing the last error event, which is also the most involved one. For this event, the average probability of error, associated with the decoder in \eqref{optDecHK}, is given by \eqref{innerProHK}, given at the top of the next page, 
\begin{figure*}[!t]
\normalsize
\setcounter{MYtempeqncnt}{\value{equation}}
\setcounter{equation}{98}
\begin{align}
P_e^{(7)} &\triangleq \Pr\pp{\bigcup_{i=1}^{M_{11}-1}\bigcup_{j=1}^{M_{12}-1}\bigcup_{k=1}^{M_{21}-1}\ppp{\sum_{l=0}^{M_{22}-1}P(\bY\vert \tilde{\bZ}_{ijk},\bZ_{22,l})\geq\sum_{l=0}^{M_{22}-1}P(\bY\vert \tilde{\bZ}_{0},\bZ_{22,l})}}\label{OrignalProbHK}\\
& = \bE\ppp{\left.\Pr\pp{\bigcup_{i=1}^{M_{11}-1}\bigcup_{j=1}^{M_{12}-1}\bigcup_{k=1}^{M_{21}-1}\ppp{\sum_{l=0}^{M_{22}-1}P(\bY\vert \tilde{\bZ}_{ijk},\bZ_{22,l})\geq\sum_{l=0}^{M_{22}-1}P(\bY\vert \tilde{\bZ}_{0},\bZ_{22,l})}\right\vert\calF_0}}\label{innerProHK}
\end{align}
\hrulefill
\vspace*{4pt}
\end{figure*}
where $\tilde{\bZ}_{ijk}\triangleq (\bZ_{11,i},\bZ_{12,j},\bZ_{21,k})$, $\tilde{\bZ}_0\triangleq(\bZ_{11,0},\bZ_{12,0},\bZ_{21,0})$, and $\calF_0 \triangleq(\tilde{\bZ}_0,\bZ_{22,0},\bY)$. We will assess the exponential behavior of \eqref{innerProHK} in the same manner as we did for \eqref{innerPro}. Specifically, we start with expressing \eqref{innerProHK} in terms of types. First, for a given joint distribution $Q_{Z_1^4Y}$, we let
\begin{align}
f(Q_{Z_1^4Y}) &\triangleq \frac{1}{n}\log P\p{\by\vert\bx_1(\bz_1,\bz_2),\bx_2(\bz_3,\bz_4)}.
\end{align}
Now, for a given joint type $Q_{Z_{1,0}^{4}Y}$ of the random vectors $\p{\bZ_{1,0},\bZ_{2,0},\bZ_{3,0},\bZ_{4,0},\bY}$, we define the set $T_I(Q_{Z_{1,0}^{4}Y})$ in \eqref{IFCTypicalHK} given at the top of the next page.
\begin{figure*}[!t]
\normalsize
\setcounter{MYtempeqncnt}{\value{equation}}
\setcounter{equation}{101}
\begin{align}
 T_I(Q_{Z_{1,0}^{4}Y})&\defi\left\{\tilde{Q}^0_{Z_1^3Z_{4,0}Y}\in\calS_0,\p{\ppp{\tilde{Q}^l_{Z_1^4Y}}_{l=1}^{M_{22}-1},\ppp{\hat{Q}^l_{Z_{1,0}^3Z_4Y}}_{l=1}^{M_{22}-1}}\in\calS_1:\;\right.\nonumber\\
&\left.\ \ \ \ \ \ \ \ \ \ \ \ \ \ \ \ \ \ \ \ \ \ \ e^{nf(\tilde{Q}_{Z_1^3Z_{4,0}Y}^0)}+\sum_{l=1}^{M_{22}-1}\pp{e^{nf(\tilde{Q}^l_{Z_1^4Y})}-e^{nf(\hat{Q}^l_{Z_{1,0}^3Z_4Y})}}\geq e^{nf(Q_{Z_{1,0}^{4}Y})}\right\}\label{IFCTypicalHK}\\
\calS_0(Q_{Z_{1,0}^{4}Y}) &\triangleq \ppp{\tilde{Q}_{Z_1^3Z_{4,0}Y}^0:\;\tilde{Q}_{Z_i}^0 = P_{Z_i},\tilde{Q}^0_{Z_{4,0}Y} = Q_{Z_{4,0}Y},\;\forall 1\leq i\leq 4},\label{margcons1HK}\\
\calS_1(Q_{Z_{1,0}^{4}Y})&\triangleq \left\{\ppp{\tilde{Q}^l_{Z_1^4Y}}_{l=1}^{M_{22}-1},\ppp{\hat{Q}^l_{Z_{1,0}^3Z_4Y}}_{l=1}^{M_{22}-1}:\;\tilde{Q}^l_{Z_i}= P_{Z_i},\tilde{Q}^l_{Y}= Q_{Y},\right.\nonumber\\
&\ \ \ \ \ \ \ \ \ \ \ \hat{Q}^l_{Z_i} = P_{Z_i},\hat{Q}^l_{Z_{1,0}^3Y}= Q_{Z_{1,0}^3Y},\;\forall 1\leq i\leq 4,\;\forall 1\leq l\leq M_{22}-1\nonumber\\
&\left.\ \ \ \ \ \ \ \ \ \ \ \ \ \ \ \ \ \ \ \ \ \ \ \ \ \ \ \ \ \ \ \ \ \ \ \ \ \ \ \ \ \ \ \ \ \ \tilde{Q}_{Z_4Y}^l=\hat{Q}_{Z_4Y}^l,\tilde{Q}^l_{Z_1^3Y}=\tilde{Q}^m_{Z_1^3Y}, \forall l,m\right\}.\label{margcons2HK}
\end{align}
\hrulefill
\vspace*{4pt}
\end{figure*}
Now, with these definitions, fixing $Q_{Z_{1,0}^{4}Y}$, it follows, by definition, that the error event
\begin{align}
&\bigcup_{i=1}^{M_{11}-1}\bigcup_{j=1}^{M_{12}-1}\bigcup_{k=1}^{M_{21}-1}\nonumber\\
&\ppp{\sum_{l=0}^{M_{22}-1}P(\bY\vert \tilde{\bZ}_{ijk},\bZ_{4,l})\geq\sum_{l=0}^{M_{22}-1}P(\bY\vert \tilde{\bZ}_{0},\bZ_{4,l})}\label{EquivErrorEvent00HK}
\end{align}
can be rewritten, in terms of types, as follows
\begin{align}
&\bigcup_{i=1}^{M_{11}-1}\bigcup_{j=1}^{M_{12}-1}\bigcup_{k=1}^{M_{21}-1}\bigcup_{T_I(Q_{Z_{1,0}^{4}Y})}\\
&\left\{ \begin{array}{ccc}
(\tilde{\bZ}_{ijk},\bz_{4,0},\by)\in T(\tilde{Q}^0_{Z_1^3Z_{4,0}Y}), \\
\ppp{(\tilde{\bZ}_{ijk},\bZ_{4,l},\by)\in T(\tilde{Q}^l_{Z_1^4Y})}_{l=1}^{M_{22}-1}, \\
\ppp{\p{\tilde{\bz}_{0},\bZ_{4,l},\by}\in T(\hat{Q}^l_{Z^3_{1,0}Z_4Y})}_{l=1}^{M_{22}-1} \end{array} \right\}.\label{EquivErrorEventHK}
\end{align}
We next analyze the probability of \eqref{EquivErrorEventHK}, conditioned on $\calF_0$. Note that the inner union in \eqref{EquivErrorEventHK} is over vectors of types (an exponential number of them). Finally, as before, we simplify the notations of \eqref{EquivErrorEventHK}, and write it equivalently as
\begin{align}
&\bigcup_{i=1}^{M_{11}-1}\bigcup_{j=1}^{M_{12}-1}\bigcup_{k=1}^{M_{21}-1}\bigcup_{\blt}\\
&\left\{ \begin{array}{ccc}
\tilde{\bZ}_{ijk}\in \calA_{\blt,0}, \\
(\tilde{\bZ}_{ijk},\bZ_{4,m})\in \calA_{\blt,m},\ \ \ \text{for }m=1,\ldots,M_{22}-1 \\
\bZ_{4,m}\in \tilde{\calA}_{\blt,m},\ \ \ \text{for }m=1,\ldots,M_{22}-1 \end{array} \right\}\label{EquivErrorEvent1HK}
\end{align}
where, again, the index ``$\blb$" in the inner union runs over the combinations of types (namely, $\blb = \{\tilde{Q}^l_{Z_1^4Y},\hat{Q}^l_{Z_{1,0}^3Z_4Y}\}_l$) that belong to $T_I(Q_{Z_{1,0}^{4}Y})$, and the various sets $\{\calA_{\blt,j},\tilde{\calA}_{\blt,j}\}_{\blt,j}$ correspond to the typical sets in \eqref{EquivErrorEventHK} (recall that $(\bz_{1,0},\bz_{2,0},\bz_{3,0},\bz_{4,0},\by)$ are given in this stage). Similarly as in the proof of Theorem \ref{th:1}, we derive upper bound on a generic probability which have the form of \eqref{EquivErrorEvent1HK}. In the following, we give a generalization of Lemma \ref{lem:Union3} to the probability of a union indexed by $K$ values. The proof is very similar to the proof of Lemma \ref{lem:Union3}, and thus omitted for brevity. For a given subset $\calJ = \ppp{j_1,\ldots,k_{\abs{\calJ}}}$ of $\ppp{1,\ldots,J}$ we write $\bZ_{\calJ}$ as a shorthand for $(Z_{j_1},\ldots,Z_{j_{\abs{\calJ}}})$.
\begin{lemma}\label{lem:Union4}
Let $\ppp{Z_1\p{i}}_{i=1}^{N_1},\ldots,\ppp{Z_J\p{i}}_{i=1}^{N_J}$ and $\ppp{V_1\p{i}}_{i=1}^{N_{J+1}},\ppp{V_2\p{i}}_{i=1}^{N_{J+1}},\ldots,\ppp{V_K\p{i}}_{i=1}^{N_{J+1}}$ be independent sequences of independently and identically distributed (i.i.d.) RVs on the alphabets $\calZ_1\times\ldots\times\calZ_J\times\calV_1\times\ldots\times\calV_K$, respectively, with $Z_1\p{i}\sim P_{Z_1},\ldots,Z_J\p{i}\sim P_{Z_J},V_1\p{i}\sim P_{V_1},\ldots,V_K\p{i}\sim P_{V_K}$. Fix a sequence of sets $\ppp{\calA_{i,1}}_{i=1}^N,\ppp{\calA_{i,2}}_{i=1}^N,\ldots,\ppp{\calA_{i,K}}_{i=1}^N$, where $\calA_{i,j}\subseteq\calZ_1\times\ldots\times\calZ_J\times\calV_{j}$, for $1\leq j\leq K$ and for all $1\leq i\leq N$. Also, fix a set $\ppp{\calA_{i,0}}_{i=1}^N$ where $\calA_{i,0}\subseteq\calZ_1\times\ldots\times\calZ_J$ for all $1\leq i\leq N$, and another sequence of sets $\ppp{\calG_{i,1}}_{i=1}^N,\ppp{\calG_{i,2}}_{i=1}^N,\ldots,\ppp{\calG_{i,K}}_{i=1}^N$, where $\calG_{i,j}\subseteq\calV_{j}$, for $1\leq j\leq K$ and for all $1\leq i\leq N$. Let 
$\bU = (Z_1,Z_2,\ldots,Z_J,U_{J+1})$ with $U_{J+1} \triangleq (V_1,\ldots,V_K)$. Finally, define $\calB_{l,\calJ}$ given in \eqref{calB1df0HK}, 
\begin{figure*}[!t]
\normalsize
\setcounter{MYtempeqncnt}{\value{equation}}
\setcounter{equation}{109}
\begin{align}
\calB_{l,\calJ}\defi\ppp{\bu_{\calJ}:\;z_1^J\in\calA_{l,0},\;\bigcap_{j=1}^{K}\p{z_1^J,v_{j}}\in\calA_{l,j},\;\bigcap_{j=1}^{K}v_j\in\calG_{l,j}\ \ \text{for some }\bu_{\calJ^c}},\label{calB1df0HK}
\end{align}
\hrulefill
\vspace*{4pt}
\end{figure*}
for $l=1,2,\ldots,N$, and $\calZ(i_1^J) = (Z_1(i_1),\ldots,Z_J(i_J))$. Then, a general upper bound is given by \eqref{EquivErrorEvent01HK}, shown at the top of the next page.
\begin{figure*}[!t]
\normalsize
\setcounter{MYtempeqncnt}{\value{equation}}
\setcounter{equation}{110}
\begin{align}
&\Pr\ppp{\bigcup_{i_1^J,j}\ppp{\bigcup_{l=1}^N\ppp{\calZ(i_1^J)\in\calA_{l,0},\;\bigcap_{k=1}^{K}\p{\calZ(i_1^J),V_{k}(j)}\in\calA_{l,k},\;\bigcap_{k=1}^{K}V_k(j)\in\calG_{l,k}}}}\nonumber\\
&\ \ \ \ \ \ \ \ \ \ \ \ \ \ \ \leq\min\left\{1,\min_{\calJ\subseteq\ppp{1,\ldots,J+1}\calJ\neq\emptyset}\p{\prod_{j\in\calJ}N_j}\Pr\ppp{\bigcup_{l=1}^N\bU_{\calJ}\in\calB_{l,\calJ}}\right\}.\label{EquivErrorEvent01HK}
\end{align}
\hrulefill
\vspace*{4pt}
\end{figure*}
\end{lemma}

Applying Lemma \ref{lem:Union4} on \eqref{EquivErrorEventHK} (or, \eqref{EquivErrorEvent1HK}) we obtain \eqref{EquivErrorEvent2HK}, shown at the top of the next page,  
\begin{figure*}[!t]
\normalsize
\setcounter{MYtempeqncnt}{\value{equation}}
\setcounter{equation}{111}
\begin{align}
&\Pr\ppp{\left.\bigcup_{i=1}^{M_{11}-1}\bigcup_{j=1}^{M_{12}-1}\bigcup_{k=1}^{M_{21}-1}\ppp{\sum_{l=0}^{M_{22}-1}P(\bY\vert \tilde{\bZ}_{ijk},\bZ_{4,l})\geq\sum_{l=0}^{M_{22}-1}P(\bY\vert \tilde{\bZ}_{0},\bZ_{4,l})}\right\vert\calF_0}\nonumber\\
&\ \ \ \ \ \ \ \ \ \ \ \ \ \ \exeq \min\left\{1,\min_{\calJ\subseteq\ppp{1,\ldots,4}\calJ\neq\emptyset}\p{\prod_{j\in\calJ}N_j}\Pr\ppp{\bigcup_{\blt}\bU_{\calJ}\in\calB_{\blt,\calJ}}\right\}\label{EquivErrorEvent2HK}
\end{align}
\hrulefill
\vspace*{4pt}
\end{figure*}
where $N_1 = M_{11}, N_2 = M_{12}, N_3 = M_{21}, N_4 = 1$, and 
\begin{align}
\bU = (\bZ_{11},\bZ_{12},\bZ_{21},\bU_4)
\end{align}
in which $\bU_4 = (\bZ_{4,1},\ldots,\bZ_{4,M_{22}-1})$, and $\calB_{\blt,\calJ}$ is given in \eqref{BjSetdef}, also shown at the top of the next page. 
\begin{figure*}[!t]
\normalsize
\setcounter{MYtempeqncnt}{\value{equation}}
\setcounter{equation}{113}
\begin{align}
\calB_{\blt,\calJ} = \left\{ \begin{array}{ccc}
(\tilde{\bz}_{111},\bz_{4,0},\by)\in T(\tilde{Q}^0_{X_1X_{2,0}Y}), \\
\bu_{\calJ}:\ \ppp{(\tilde{\bz}_{111},\bz_{4,l},\by)\in T(\tilde{Q}^l_{Z_1^4Y})}_{l=1}^{M_{22}-1}, \\
\ppp{\p{\tilde{\bz}_{0},\bz_{4,l},\by}\in T(\hat{Q}^j_{X_{1,0}X_2Y})}_{l=1}^{M_{22}-1},\ \text{for some }\bu_{\calJ^c}\end{array} \right\}\label{BjSetdef}
\end{align}
\hrulefill
\vspace*{4pt}
\end{figure*}
The various possibilities for the set $\calJ$ are,
\begin{align}
\left\{ \begin{array}{ccc}
1;2;3;4;\\
12;13;14;23;24;34;\\
123;124;134;234;\\
1234
\end{array} \right\},
\end{align}
namely, we have 15 possibilities. We claim that possibilities $\ppp{1,2,3,12,13,23,123}$ do not affect the outer minimum in \eqref{EquivErrorEvent2HK}, and so we left with possibilities $\ppp{4,14,24,34,124,134,234,1234}$. This observation follows from the same arguments used in \eqref{lastProRightfin0}-\eqref{biggerThan} for the second term at the r.h.s. of \eqref{Last3Prob}. For example, possibilities $\ppp{1,2,3}$ do not affect the outer minimum due to the fact that the probabilities that correspond to possibilities $\ppp{14,24,34}$, respectively, are smaller. Indeed, the multiplicative factors in \eqref{EquivErrorEvent2HK} for each of the pairs $(1,14)$, $(2,24)$, and $(3,34)$, are the same, but the respective probabilities in \eqref{EquivErrorEvent2HK} are smaller for $\ppp{14,24,34}$ (due to the same reason used in \eqref{biggerThan}). Similarly, possibilities $\ppp{12,13,23,123}$ do not affect the outer minimum due to possibilities $\ppp{124,134,234,1234}$, respectively.

In the following, we analyze the remaining terms. For example, the term that corresponds to possibility ``$1234$", is given by
\begin{align}
P_{e,1234}\triangleq M_{11}M_{12}M_{21}\Pr\ppp{\bigcup_{\blt}\bU\in\calB_{\blt,1234}},
\end{align}
which similarly to the passage from \eqref{EquivErrorEvent00HK} to \eqref{EquivErrorEventHK}, can be rewritten as in \eqref{EquivErrorEvent03HK}, shown at the top of the next page.
\begin{figure*}[!t]
\normalsize
\setcounter{MYtempeqncnt}{\value{equation}}
\setcounter{equation}{116}
\begin{align}
P_{e,1234}&=M_{11}M_{12}M_{21}\Pr\ppp{\left.\sum_{l=0}^{M_{22}-1}P(\bY\vert \tilde{\bZ}_{111},\bZ_{4,l})\geq\sum_{l=0}^{M_{22}-1}P(\bY\vert \tilde{\bZ}_{0},\bZ_{4,l})\right\vert\calF_0}\nonumber\\
&=M_{11}M_{12}M_{21}\bE\ppp{\left.\left.\Pr\ppp{\sum_{l=0}^{M_{22}-1}P(\bY\vert \tilde{\bZ}_{111},\bZ_{4,l})\geq\sum_{l=0}^{M_{22}-1}P(\bY\vert \tilde{\bZ}_{0},\bZ_{4,l})\right\vert\calF_0,\tilde{\bZ}_{111}}\right\vert\calF_0}
\label{EquivErrorEvent03HK}
\end{align}
\hrulefill
\vspace*{4pt}
\end{figure*}
Equation \eqref{EquivErrorEvent03HK} has the same form of the probability in \eqref{sameFormRCE}, which we already analyzed. Accordingly, we similarly obtain \eqref{showattop}, 
\begin{figure*}[!t]
\normalsize
\setcounter{MYtempeqncnt}{\value{equation}}
\setcounter{equation}{117}
\begin{align}
&\bE\ppp{\left.\left.\Pr\ppp{\sum_{l=0}^{M_{22}-1}P(\bY\vert \tilde{\bZ}_{111},\bZ_{4,l})\geq\sum_{l=0}^{M_{22}-1}P(\bY\vert \tilde{\bZ}_{0},\bZ_{4,l})\right\vert\calF_0,\tilde{\bZ}_{111}}\right\vert\calF_0}\\
&\exe  \exp\ppp{-n\min_{Q_{Z_1^3\vert Z_{4,0}Y}\in\calS_{\ppp{4}}(Q_{Z_{1,0}^4Y})}\pp{I_Q(Z_1^3;Z_{4,0},Y)+E_7(Q_{Z_1^3Z_{4,0}Y},Q_{Z_{1,0}^4Y})}}\\
&\triangleq \exp\ppp{-n\hat{E}_7^{(6)}(Q_{Z_{1,0}^4Y},R_{22})}\label{showattop}
\end{align}
\hrulefill
\vspace*{4pt}
\end{figure*}
where $E_7(\cdot,\cdot)$ is defined in \eqref{HKEef}, and $\calS_{\ppp{4}}(Q)$ is given in \eqref{HKSef}.

The other terms are handled in a similar fashion. Specifically, let $\hat{\bZ} \triangleq \ppp{Z_1,Z_2,Z_3}$, and define the sets $\calU = \ppp{1,2,3,12,13,23,123}$, and $\tilde{\calU} = \ppp{14,24,34,124,134,234,1234}$. Then, define for any\footnote{Note that $P_{e,7}^{(6)}$ correspond to $P_{e,1234}$ in \eqref{EquivErrorEvent03HK}.} $u\in\ppp{1,2,\ldots,7}$:
\begin{align}
P_{e,u}^{(6)}\triangleq M_{\calU(u)}\cdot\Pr\ppp{\bigcup_{\blt}\bU_{\tilde{\calU}(u)}\in\calB_{\blt,\tilde{\calU}(u)}},
\end{align}
where
\begin{align}
&M_{\calU(1)} \triangleq M_{11};\;M_{\calU(2)} = M_{12};\;M_{\calU(3)} = M_{21};\nonumber\\
&M_{\calU(4)} = M_{11}M_{12};\;M_{\calU(5)} \triangleq M_{11}M_{21};\nonumber\\
&M_{\calU(6)} = M_{12}M_{21}\;M_{\calU(7)} = M_{11}M_{12}M_{21}.
\end{align}
Accordingly, following \eqref{N1QDef}-\eqref{lastProRight3ere}, we get \eqref{sshowtop2}, shown at the top of the next page, 
\begin{figure*}[!t]
\normalsize
\setcounter{MYtempeqncnt}{\value{equation}}
\setcounter{equation}{122}
\begin{align}
P_{e,u}^{(6)} &\exe \exp\ppp{-n\min_{Q_{Z_1^3\vert Z_{4,0}Y}\in\calS_{\ppp{4}}(Q_{Z_{1,0}^4Y})}\pp{I_Q(\hat{\bZ}_{{\calU}(u)};Z_{4,0},Y\vert \hat{\bZ}_{123\setminus{\calU}(u)})+E_u(Q_{Z_1^3Z_{4,0}Y},Q_{Z_{1,0}^4Y})}}\nonumber\\
&\triangleq \exp\ppp{-n\hat{E}_u^{(6)}(Q_{Z_{1,0}^4Y},R_{22})}\label{sshowtop2}
\end{align}
\hrulefill
\vspace*{4pt}
\end{figure*}
where $E_u(\cdot,\cdot)$ is defined in \eqref{HKEef}. Note that the mutual information term in the above exponent is due to the averaging over $\hat{\bZ}_{{\calU}(u)}$, and it is resulted by using the method of types as in \eqref{methodoftypes}. This concludes the analysis for possibilities $\ppp{14,24,34,124,134,234,1234}$, and we left with possibility $\ppp{4}$, which is very similar to \eqref{weirdposs}. Accordingly, using the same arguments in \eqref{lastProRight2}-\eqref{lastProRight3ere}, we obtain
\begin{align}
&P_{e,8}^{(6)}\triangleq \Pr\ppp{\bigcup_{\blt}\bU_4\in\calB_{\blt,4}}\nonumber\\
&\exe \exp\ppp{-n\min_{Q_{Z_1^3\vert Z_{4,0}Y}\in\calS_{\ppp{4}}(Q_{Z_{1,0}^4Y})}E_0(Q_{Z_1^3Z_{4,0}Y},Q_{Z_{1,0}^4Y})}\nonumber\\
&\triangleq \exp\ppp{-n\hat{E}^{(6)}_8(Q_{Z_{1,0}^4Y},R_{22})}
\end{align}
where $E_0(\cdot,\cdot)$ is, again, defined in \eqref{HKEef}. Wrapping up, using \eqref{innerProHK}, \eqref{EquivErrorEvent2HK}, and the last results, after averaging w.r.t. $\calF_0$, we get \eqref{dfdfds}, shown at the top of the next page, 
\begin{figure*}[!t]
\normalsize
\setcounter{MYtempeqncnt}{\value{equation}}
\setcounter{equation}{124}
\begin{align}
P_e^{(6)}&\exeq \bE\ppp{\min\ppp{1,\min_{u\in\ppp{1:7}} e^{-n\pp{\hat{E}_{u}^{(6)}(Q_{Z_{1,0}^4Y},R_{22})-n^{-1}\log M_{\tilde{\calU}(u)}}},e^{-n\hat{E}_8^{(6)}(Q_{Z_{1,0}^4Y},R_{22})}}}\\
& = \bE\ppp{\min\ppp{\min_{u\in\ppp{1:7}}e^{-n\pp{\hat{E}_{u}^{(6)}(Q_{Z_{1,0}^4Y},R_{22})-n^{-1}\log M_{\tilde{\calU}(u)}}_+},e^{-n\hat{E}_8^{(6)}(Q_{Z_{1,0}^4Y},R_{22})}}}\\
& = \bE\ppp{\exp\pp{-n\max\ppp{\max_u\pp{\hat{E}_{u}^{(6)}(Q_{Z_{1,0}^4Y},R_{22})-\frac{1}{n}\log M_{\tilde{\calU}(u)}}_+,\hat{E}_8^{(6)}(Q_{Z_{1,0}^4Y},R_{22})}}}\\
&\exe \exp\ppp{-n\pp{\min_{Q_{Y|Z_{1,0}^4}}\pp{D(Q_{Y|Z_{1,0}^4}||W_{Y|Z_{1,0}^4}\vert P_{Z_{1,0}^4})+E_{\text{HK}}^{(6)}(Q_{Z_{1,0}^4Y})}}}\\
E_{\text{HK}}^{(6)}(Q_{Z_{1,0}^4Y}) &\triangleq \max\ppp{\max_{u\in\ppp{1:7}}\pp{\hat{E}_{u}^{(6)}(Q_{Z_{1,0}^4Y},R_{22})-\calR_{u}}_+,\hat{E}_8^{(6)}(Q_{Z_{1,0}^4Y},R_{22})}\label{dfdfds}
\end{align}
\hrulefill
\vspace*{4pt}
\end{figure*}
where $\calR_u$ for $u=1,2,\ldots,7$ is defined in \eqref{ratesSums}. 

This concludes the analysis of the error event $(\hat{m}_{11}\neq0,\hat{m}_{12}\neq0,\hat{m}_{21}\neq0)$ in \eqref{overalProb}. The other types of errors are analyzed in a similar manner. For $(\hat{m}_{11}\neq0,\hat{m}_{12}=0,\hat{m}_{21}=0)$, the average probability of error, associated with the decoder in \eqref{optDecHK} is given in \eqref{shown3}, given at the top of the next page, 
\begin{figure*}[!t]
\normalsize
\setcounter{MYtempeqncnt}{\value{equation}}
\setcounter{equation}{129}
\begin{align}
P_e^{(1)} &= \Pr\pp{\bigcup_{i=1}^{M_{11}-1}\ppp{\sum_{l=0}^{M_{22}-1}P(\bY\vert \tilde{\bZ}_{i00},\bZ_{22,l})\geq\sum_{l=0}^{M_{22}-1}P(\bY\vert \tilde{\bZ}_{0},\bZ_{22,l})}}\\
& = \bE\ppp{\left.\Pr\pp{\bigcup_{i=1}^{M_{11}-1}\ppp{\sum_{l=0}^{M_{22}-1}P(\bY\vert \tilde{\bZ}_{i00},\bZ_{22,l})\geq\sum_{l=0}^{M_{22}-1}P(\bY\vert \tilde{\bZ}_{0},\bZ_{22,l})}\right\vert\calF_0}}\label{shown3}
\end{align}
\hrulefill
\vspace*{4pt}
\end{figure*}
where $\calF_0 \triangleq(\tilde{\bZ}_0,\bZ_{22,0},\bY)$. Thus, due to the fact that $(\bZ_{12,0},\bZ_{21,0})$ are now fixed, they play a same role as $\bY$ and $\bZ_{22,0}$. Accordingly, following the same steps as in \eqref{lastProRight2}-\eqref{lastProRight3ere}, we get \eqref{pe1}, presented at the top of the next page, 
\begin{figure*}[!t]
\normalsize
\setcounter{MYtempeqncnt}{\value{equation}}
\setcounter{equation}{131}
\begin{align}
P_e^{(1)}&\exeq \exp\ppp{-n\pp{\min_{Q_{Y|Z_{1,0}^4}}\pp{D(Q_{Y|Z_{1,0}^4}||W_{Y|Z_{1,0}^4}\vert P_{Z_{1,0}^4})+E_{\text{HK}}^{(1)}(Q_{Z_{1,0}^4Y})}}}\label{pe1}\\
E_{\text{HK}}^{(1)}(Q_{Z_{1,0}^4Y}) &\triangleq \max\ppp{\pp{\hat{E}^{(1)}(Q_{Z_{1,0}^4Y},R_{22})-\calR_{1}}_+,\hat{E}^{(1)}_8(Q_{Z_{1,0}^4Y},R_{22})}\label{pe12}\\
\hat{E}^{(1)}(Q_{Z_{1,0}^4Y},R_{22}) &= \min_{Q_{Z_1\vert Z_{2,0}^4Y}\in\calS_{\ppp{2,3,4}}(Q_{Z_{1,0}^4Y})}\pp{I_Q(Z_1;Z_{2,0}^4,Y)+E_7(Q_{Z_1Z_{2,0}^4Y},Q_{Z_{1,0}^4Y})}\label{pe13}\\
\hat{E}_8^{(1)}(Q_{Z_{1,0}^4Y},R_{22}) &= \min_{Q_{Z_1\vert Z_{2,0}^4Y}\in\calS_{\ppp{2,3,4}}(Q_{Z_{1,0}^4Y})}E_6(Q_{Z_1Z_{2,0}^4Y},Q_{Z_{1,0}^4Y})\label{pe138}
\end{align}
\hrulefill
\vspace*{4pt}
\end{figure*}
where $E_6(\cdot,\cdot)$ and $E_7(\cdot,\cdot)$ are defined in \eqref{HKEef}, and $\calS_{\ppp{2,3,4}}(Q)$ is given in \eqref{HKSef}. Again, since $(\bZ_{12,0},\bZ_{21,0})$, which correspond to $(Z_{2,0},Z_{3,0})$), are fixed, they are conjugated to $(Z_{4,0},Y)$. The error exponent of $P_e^{(2)}$ which corresponds to $(\hat{m}_{11}=0,\hat{m}_{12}\neq0,\hat{m}_{21}=0)$ can be derived in the same way. We get that the exponent of $P_e^{(2)}$ is obtained by replacing the role of $Z_1$ with $Z_2$ and $\calR_1$ with $\calR_2$, in \eqref{pe1}-\eqref{pe13}. Similarly, $P_e^{(3)}$, corresponding to $(\hat{m}_{11}\neq0,\hat{m}_{12}\neq0,\hat{m}_{21}=0)$, is upper bounded by \eqref{shown4}, presented at the top of the next page,
\begin{figure*}[!t]
\normalsize
\setcounter{MYtempeqncnt}{\value{equation}}
\setcounter{equation}{135}
\begin{align}
P_e^{(3)}&\exeq \exp\ppp{-n\pp{\min_{Q_{Y|Z_{1,0}^4}}\pp{D(Q_{Y|Z_{1,0}^4}||W_{Y|Z_{1,0}^4}\vert P_{Z_{1,0}^4})+E_{\text{HK}}^{(3)}(Q_{Z_{1,0}^4Y})}}}\label{shown4}\\
E_{\text{HK}}^{(3)}(Q_{Z_{1,0}^4Y}) &\triangleq \max\ppp{\max_{u \in\ppp{1,2,4} }\pp{\hat{E}_u^{(3)}(Q_{Z_{1,0}^4Y},R_{22})-\calR_{u}}_+,\hat{E}_8^{(3)}(Q_{Z_{1,0}^4Y},R_{22})}\label{minindex}\\
\hat{E}_1^{(3)}(Q_{Z_{1,0}^4Y},R_{22}) &= \min_{Q_{Z_1^2\vert Z_{3,0}^4Y}\in\calS_{\ppp{3,4}}(Q_{Z_{1,0}^4Y})}\pp{I_Q(Z_1;Z_{3,0}^4,Y\vert Z_2)+E_5(Q_{Z_1^2Z_{3,0}^4Y},Q_{Z_{1,0}^4Y})}\\
\hat{E}_2^{(3)}(Q_{Z_{1,0}^4Y},R_{22}) &= \min_{Q_{Z_1^2\vert Z_{3,0}^4Y}\in\calS_{\ppp{3,4}}(Q_{Z_{1,0}^4Y})}\pp{I_Q(Z_2;Z_{3,0}^4,Y\vert Z_1)+E_6(Q_{Z_1^2Z_{3,0}^4Y},Q_{Z_{1,0}^4Y})}\\
\hat{E}_4^{(3)}(Q_{Z_{1,0}^4Y},R_{22}) &= \min_{Q_{Z_1^2\vert Z_{3,0}^4Y}\in\calS_{\ppp{3,4}}(Q_{Z_{1,0}^4Y})}\pp{I_Q(Z_1,Z_2;Z_{3,0}^4,Y)+E_7(Q_{Z_1^2Z_{3,0}^4Y},Q_{Z_{1,0}^4Y})}\\
\hat{E}_8^{(3)}(Q_{Z_{1,0}^4Y},R_{22}) &= \min_{Q_{Z_1^2\vert Z_{3,0}^4Y}\in\calS_{\ppp{3,4}}(Q_{Z_{1,0}^4Y})}E_3(Q_{Z_1^2Z_{3,0}^4Y},Q_{Z_{1,0}^4Y})
\end{align}
\hrulefill
\vspace*{4pt}
\end{figure*}
where $\calS_{\ppp{3,4}}(Q)$ is defined in \eqref{HKSef}.

Finally, the error exponents of $P_e^{(4)}$ and $P_e^{(5)}$, corresponding to $(\hat{m}_{11}\neq0,\hat{m}_{12}=0,\hat{m}_{21}\neq0)$ and $(\hat{m}_{11}=0,\hat{m}_{12}\neq0,\hat{m}_{21}\neq0)$, respectively, are obtained in the same way. The exponent of $P_e^{(4)}$ is obtained by replacing the role of $Z_2$ with $Z_3$, and changing the minimization in \eqref{minindex} to over the indexes $\ppp{1,3,5}$, and the exponent of $P_e^{(5)}$ is obtained by replacing the role of $Z_1$ with $Z_3$, and changing the minimization in \eqref{minindex} to over the indexes $\ppp{2,3,6}$.

\appendices
\numberwithin{equation}{section}

\section{Definitions for Theorem~\ref{th:2}}
\label{app:0}
In this appendix, we give the definitions of the various parameters appearing in Theorem~\ref{th:2}. For simplicity of notation, in the following, we use the indexes $\ppp{1,2,3,4}$ instead of $\ppp{11,12,21,22}$, respectively. Let $\bZ \triangleq (Z_1,Z_2,Z_3)$, and $\calU \triangleq \ppp{1,2,3,12,13,23,123}$. For $u\in\ppp{0,1,2,\ldots,7}$, $\bZ_{\calU(u)}$ is a random vector consisting of the RVs corresponding to the indexes in $\calU(u)$, for example, $\bZ_1 \triangleq \bZ_{\calU(1)} = Z_1$, $\bZ_{12} \triangleq\bZ_{\calU(4)} = (Z_1,Z_2)$, $\bZ_{123} \triangleq\bZ_{\calU(7)} = (Z_1,Z_2,Z_3)$, and so on, where we define $\bZ_{\calU(0)} = \emptyset$. Let also $\tilde{\bZ}\triangleq \ppp{Z_1,Z_2,Z_3,Z_4}$, $\calI\subseteq\ppp{1,2,3,4}$, and $\tilde{\bZ}_\calI$ be the restriction of the entries of $\tilde{\bZ}$ on the set $\calI$. Then, let
\begin{align}
\calS_{\calI}(Q)\triangleq\ppp{\tilde Q:\;\tilde Q_{\tilde{\bzt}_\calI Y_1}=Q_{\tilde{\bzt}_\calI Y_1},\;\tilde Q_{Z_i}=P_{Z_i},\;\text{for }i\in\calI^c}.\label{HKSef}
\end{align}
Define
\begin{align}
f(Q_{Z_1^4Y_1}) &\triangleq \bE_Q\pp{\log W_{Y_1\vert X_1X_2}(Y_1\vert X_1(Z_{1},Z_{2}),X_2(Z_{3},Z_{4}))}
\end{align}
and let
\begin{align}
r_0(Q_{Z_1^3Y_1}) \triangleq R_{22}+\max_{\substack{\hat{Q}:\;\hat{Q}\in\calS_{\ppp{1,2,3}}(Q)\\ I_{\hat{Q}}(Z_4;Z_1^3,Y_1)\leq R_{22}}}f(\hat{Q})-I_{\hat{Q}}(Z_4;Z_1^3,Y_1).
\end{align}
For $u\in\ppp{0,1,2,\ldots,7}$, define
\begin{align}
&E_{u}(\tilde{Q}_{Z_1^4Y_1},Q_{Z_1^4Y_1}) \triangleq\nonumber\\
&\ \ \ \ \ \ \ \ \ \ \ \min_{\substack{\hat{Q}:\;\hat{Q}\in\calS_{\ppp{1,2,3}}(\tilde{Q})\\ \hat{Q}\in\calD_u(\tilde{Q}_{Z_1^4Y_1},Q_{Z_1^4Y_1})}}\pp{I_{\hat{Q}}(Z_4;\bZ_{\calU(u)},Y_1)-R_{22}}_+,\label{HKEef}
\end{align}
where
\begin{align}
&\calD_u(\tilde{Q}_{Z_1^4Y_1},Q_{Z_1^4Y_1}) \triangleq \left\{\vphantom{\max\ppp{f(\tilde{Q}_{Z_1^4Y_1}),f(\hat{Q})+\pp{R_{22}-I_{\hat{Q}}(Z_4;Z_1^3,Y_1)}_+}}\hat{Q}:\;\max\pp{r_0(Q_{Z_1^4Y_1}),f(Q_{Z_1^4Y_1})}\right.\nonumber\\
&\left.\leq\max\ppp{f(\tilde{Q}_{Z_1^4Y_1}),f(\hat{Q})+\pp{R_{22}-I_{\hat{Q}}(Z_4;\bZ_{\calU(u)},Y_1)}_+}\right\}
\end{align}
Finally, we let
\begin{align}
&\calR_{1} \triangleq R_{11};\;\calR_{2} \triangleq R_{12};\;\calR_{3} \triangleq R_{21};\;\calR_{4} \triangleq R_{11}+R_{12};\nonumber\\
&\calR_{5} \triangleq R_{11}+R_{21};\;\calR_{6} \triangleq R_{12}+R_{21};\nonumber\\
&\calR_{7} \triangleq R_{11}+R_{12}+R_{21}.\label{ratesSums}
\end{align}
Using all the above definition, we define \eqref{appendixA0}-\eqref{appendixA}, shown at the top of page~\pageref{appendixA}.
\begin{figure*}[!t]
\normalsize
\setcounter{MYtempeqncnt}{\value{equation}}
\setcounter{equation}{6}
\begin{align}
&\hat{E}^{(1)}(Q_{Z_{1}^4Y_1},R_{22}) \triangleq \min_{\tilde{Q}:\;\tilde Q\in\calS_{\ppp{2,3,4}}(Q)}\pp{I_{\tilde{Q}}(Z_1;Z_2^{4},Y_1)+E_7(\tilde{Q}_{Z_1^4Y_1},Q_{Z_{1}^4Y_1})},\label{appendixA0}\\
&\hat{E}^{(2)}(Q_{Z_{1}^4Y_1},R_{22}) \triangleq \min_{\tilde{Q}:\;\tilde Q\in\calS_{\ppp{1,3,4}}(Q)}\pp{I_{\tilde{Q}}(Z_2;Z_1,Z_3^{4},Y_1)+E_7(\tilde{Q}_{Z_1^4Y_1},Q_{Z_{1}^4Y_1})},\\
&\hat{E}_8^{(1)}(Q_{Z_{1}^4Y_1},R_{22}) \triangleq \min_{\tilde{Q}:\;\tilde Q\in\calS_{\ppp{2,3,4}}(Q)}E_6(\tilde{Q}_{Z_1^4Y},Q_{Z_{1}^4Y_1}),\\
&\hat{E}_8^{(2)}(Q_{Z_{1}^4Y_1},R_{22}) \triangleq \min_{\tilde{Q}:\;\tilde Q\in\calS_{\ppp{1,3,4}}(Q)}E_5(\tilde{Q}_{Z_1^4Y},Q_{Z_{1}^4Y_1}),\\
&\hat{E}_{1}^{(3)}(Q_{Z_{1}^4Y_1},R_{22}) \triangleq \min_{\tilde{Q}:\;\tilde Q\in\calS_{\ppp{3,4}}(Q)}\pp{I_{\tilde{Q}}(\bZ_{{\calU}(1)};Z_{3}^4,Y_1\vert \bZ_{12\setminus{\calU}(1)})+E_5(\tilde{Q}_{Z_1^4Y_1},Q_{Z_{1}^4Y})},\\
&\hat{E}_{2}^{(3)}(Q_{Z_{1}^4Y_1},R_{22}) \triangleq \min_{\tilde{Q}:\;\tilde Q\in\calS_{\ppp{3,4}}(Q)}\pp{I_{\tilde{Q}}(\bZ_{{\calU}(2)};Z_{3}^4,Y_1\vert \bZ_{12\setminus{\calU}(2)})+E_6(\tilde{Q}_{Z_1^4Y_1},Q_{Z_{1}^4Y})},\\
&\hat{E}_{4}^{(3)}(Q_{Z_{1}^4Y_1},R_{22}) \triangleq \min_{\tilde{Q}:\;\tilde Q\in\calS_{\ppp{3,4}}(Q)}\pp{I_{\tilde{Q}}(\bZ_{{\calU}(4)};Z_{3}^4,Y_1\vert \bZ_{12\setminus{\calU}(4)})+E_7(\tilde{Q}_{Z_1^4Y_1},Q_{Z_{1}^4Y})},\\
&\hat{E}_8^{(3)}(Q_{Z_{1}^4Y_1},R_{22}) \triangleq \min_{\tilde{Q}:\;\tilde Q\in\calS_{\ppp{3,4}}(Q)}E_3(\tilde{Q}_{Z_1^4Y},Q_{Z_{1}^4Y_1}),\\
&\hat{E}_{1}^{(4)}(Q_{Z_{1}^4Y_1},R_{22}) \triangleq \min_{\tilde{Q}:\;\tilde Q\in\calS_{\ppp{2,4}}(Q)}\pp{I_{\tilde{Q}}(\bZ_{{\calU}(1)};Z_2,Z_4,Y_1\vert \bZ_{13\setminus{\calU}(1)})+E_4(\tilde{Q}_{Z_1^4Y_1},Q_{Z_{1}^4Y_1})},\\
&\hat{E}_{3}^{(4)}(Q_{Z_{1}^4Y_1},R_{22}) \triangleq \min_{\tilde{Q}:\;\tilde Q\in\calS_{\ppp{2,4}}(Q)}\pp{I_{\tilde{Q}}(\bZ_{{\calU}(3)};Z_2,Z_4,Y_1\vert \bZ_{13\setminus{\calU}(3)})+E_6(\tilde{Q}_{Z_1^4Y_1},Q_{Z_{1}^4Y_1})},\\
&\hat{E}_{5}^{(4)}(Q_{Z_{1}^4Y_1},R_{22}) \triangleq \min_{\tilde{Q}:\;\tilde Q\in\calS_{\ppp{2,4}}(Q)}\pp{I_{\tilde{Q}}(\bZ_{{\calU}(5)};Z_2,Z_4,Y_1\vert \bZ_{13\setminus{\calU}(5)})+E_7(\tilde{Q}_{Z_1^4Y_1},Q_{Z_{1}^4Y_1})},\\
&\hat{E}_8^{(4)}(Q_{Z_{1}^4Y_1},R_{22}) \triangleq \min_{\tilde{Q}:\;\tilde Q\in\calS_{\ppp{2,4}}(Q)}E_2(\tilde{Q}_{Z_1^4Y_1},Q_{Z_{1}^4Y_1}),\\
&\hat{E}_{2}^{(5)}(Q_{Z_{1}^4Y_1},R_{22}) \triangleq \min_{\tilde{Q}:\;\tilde Q\in\calS_{\ppp{1,4}}(Q)}\pp{I_{\tilde{Q}}(\bZ_{{\calU}(2)};Z_1,Z_4,Y_1\vert \bZ_{23\setminus{\calU}(2)})+E_4(\tilde{Q}_{Z_1^4Y_1},Q_{Z_{1}^4Y_1})},\\
&\hat{E}_{3}^{(5)}(Q_{Z_{1}^4Y_1},R_{22}) \triangleq \min_{\tilde{Q}:\;\tilde Q\in\calS_{\ppp{1,4}}(Q)}\pp{I_{\tilde{Q}}(\bZ_{{\calU}(3)};Z_1,Z_4,Y_1\vert \bZ_{23\setminus{\calU}(3)})+E_5(\tilde{Q}_{Z_1^4Y_1},Q_{Z_{1}^4Y_1})},\\
&\hat{E}_{6}^{(5)}(Q_{Z_{1}^4Y_1},R_{22}) \triangleq \min_{\tilde{Q}:\;\tilde Q\in\calS_{\ppp{1,4}}(Q)}\pp{I_{\tilde{Q}}(\bZ_{{\calU}(6)};Z_1,Z_4,Y_1\vert \bZ_{23\setminus{\calU}(6)})+E_7(\tilde{Q}_{Z_1^4Y_1},Q_{Z_{1}^4Y_1})},\\
&\hat{E}_8^{(5)}(Q_{Z_{1}^4Y_1},R_{22}) \triangleq \min_{\tilde{Q}:\;\tilde Q\in\calS_{\ppp{1,4}}(Q)}E_1(\tilde{Q}_{Z_1^4Y_1},Q_{Z_{1}^4Y_1}),\\
&\hat{E}_{u}^{(6)}(Q_{Z_{1}^4Y_1},R_{22}) \triangleq \min_{\tilde{Q}:\;\tilde Q\in\calS_{\ppp{4}}(Q)}\pp{I_{\tilde{Q}}(\bZ_{{\calU}(u)};Z_{4},Y_1\vert \bZ_{123\setminus{\calU}(u)})+E_u(\tilde{Q}_{Z_1^4Y_1},Q_{Z_{1}^4Y_1})},\;u\in\ppp{1,\ldots,7},\label{Eu6}\\
&\hat{E}_8^{(6)}(Q_{Z_1^4Y_1},R_{22}) \triangleq \min_{\tilde{Q}:\;\tilde Q\in\calS_{\ppp{4}}(Q)}E_0(\tilde{Q}_{Z_1^4Y_1},Q_{Z_{1}^4Y_1}),\\
&E_{\text{HK}}^{(u)}(Q_{Z_1^4Y_1}) \triangleq \max\ppp{\pp{\hat{E}^{(u)}(Q_{Z_{1}^4Y_1},R_{22})-\calR_{u}}_+,\hat{E}_8^{(u)}(Q_{Z_{1}^4Y_1},R_{22})},\;u\in\ppp{1,2},\\
&E_{\text{HK}}^{(3)}(Q_{Z_1^4Y_1}) \triangleq \max\ppp{\max_{u\in\ppp{1,2,4}}\pp{\hat{E}_{u}^{(3)}(Q_{Z_{1}^4Y_1},R_{22})-\calR_{u}}_+,\hat{E}_8^{(3)}(Q_{Z_{1}^4Y_1},R_{22})},\\
&E_{\text{HK}}^{(4)}(Q_{Z_1^4Y_1}) \triangleq \max\ppp{\max_{u\in\ppp{1,3,5}}\pp{\hat{E}_{u}^{(4)}(Q_{Z_{1}^4Y_1},R_{22})-\calR_{u}}_,\hat{E}_8^{(4)}(Q_{Z_{1}^4Y_1},R_{22})},\\
&E_{\text{HK}}^{(5)}(Q_{Z_1^4Y_1}) \triangleq \max\ppp{\max_{u\in\ppp{2,3,6}}\pp{\hat{E}_{u}^{(5)}(Q_{Z_{1}^4Y_1},R_{22})-\calR_{u}}_+,\hat{E}_8^{(5)}(Q_{Z_{1}^4Y_1},R_{22})},\\
&E_{\text{HK}}^{(6)}(Q_{Z_1^4Y_1}) \triangleq \max\ppp{\max_{u\in\ppp{1:7}}\pp{\hat{E}_{u}^{(6)}(Q_{Z_{1}^4Y_1},R_{22})-\calR_{u}}_+,\hat{E}_8^{(6)}(Q_{Z_1^4Y_1},R_{22})},\label{ErrorUserHK}\\
&\tilde E_{\text{HK}}(R_{11},R_{12},R_{21},R_{22}) \triangleq\min_{\substack{Q_{Y_1|Z_{1}^4}:\\ Q_{Z_i} = P_{Z_i},\;1\leq i\leq 4}}\pp{D(Q_{Y_1|Z_{1}^4}||W_{Y_1|Z_{1}^4}\vert P_{Z_{1}^4})+\min_{u\in\ppp{1:6}}E_{\text{HK}}^{(u)}(Q_{Z_1^4Y_1})}.\label{appendixA}
\end{align}
\hrulefill
\vspace*{4pt}
\end{figure*}

\section{Proof of Lemma \ref{lem:Union3}}
\label{app:1}

We prove a generalized version of Lemma \ref{lem:Union3}, where we consider random sequences, $\ppp{V_2\p{i}}_{i=1}^{L_2},\ldots,\ppp{V_K\p{i}}_{i=1}^{L_2}$, rather than single RVs $V_2,\ldots,V_K$. Lemma \ref{lem:Union3} is then obtained on substituting $L_2=1$. We start with the following result which can be thought of as an extension of \cite[Lemma 2]{scarletNew}.

\begin{lemma}\label{lem:scar3}
Let $\ppp{V_1\p{i}}_{i=1}^{L_1},\ppp{V_2\p{i}}_{i=1}^{L_2},\ldots,\ppp{V_K\p{i}}_{i=1}^{L_2}$ be independent sequences of independently and identically distributed (i.i.d.) RVs on the alphabets $\calV_1\times\calV_2\times\ldots\times\calV_K$, respectively, with $V_1\p{i}\sim P_{V_1},V_2\p{i}\sim P_{V_2},\ldots,V_K\p{i}\sim P_{V_K}$. Fix a sequence of sets $\ppp{\calA_{i,1}}_{i=1}^N,\ppp{\calA_{i,2}}_{i=1}^N,\ldots,\ppp{\calA_{i,K-1}}_{i=1}^N$, where $\calA_{i,j}\subseteq\calV_1\times\calV_{j+1}$, for $1\leq j\leq K-1$ and for all $1\leq i\leq N$. Also, fix a set $\ppp{\calA_{i,0}}_{i=1}^N$ where $\calA_{i,0}\subseteq\calV_1$ for all $1\leq i\leq N$, and another sequence of sets $\ppp{\calG_{i,2}}_{i=1}^N,\ppp{\calG_{i,3}}_{i=1}^N,\ldots,\ppp{\calG_{i,K}}_{i=1}^N$, where $\calG_{i,j}\subseteq\calV_{j}$, for $2\leq j\leq K$ and for all $1\leq i\leq N$. We have \eqref{res3s0}, shown at the top of page~\pageref{res3s0}, 
\begin{figure*}[!t]
\normalsize
\setcounter{MYtempeqncnt}{\value{equation}}
\setcounter{equation}{0}
\begin{align}
&\Pr\ppp{\bigcup_{i,j}\ppp{\bigcup_{l=1}^N\ppp{V_1(i)\in\calA_{l,0},\;\bigcap_{k=1}^{K-1}\p{V_1(i),V_{k+1}(j)}\in\calA_{l,k},\;\bigcap_{k=2}^{K}V_k(j)\in\calG_{l,k}}}}\nonumber\\
&\leq\min\left\{1,L_1\bE\pp{\min\ppp{1,L_2\Pr\ppp{\left.\bigcup_{l=1}^N\ppp{V_1\in\calA_{l,0},\;\bigcap_{k=1}^{K-1}\p{V_1,V_{k+1}}\in\calA_{l,k},\;\bigcap_{k=2}^{K}V_k\in\calG_{l,k}}\right\vert V_1}}},\right.\nonumber\\
&\left. L_2\bE\pp{\min\ppp{1,L_1\Pr\ppp{\left.\bigcup_{l=1}^N\ppp{V_1\in\calA_{l,0},\;\bigcap_{k=1}^{K-1}\p{V_1,V_{k+1}}\in\calA_{l,k},\;\bigcap_{k=2}^{K}V_k\in\calG_{l,k}}\right\vert \ppp{V_k}_{k=2}^K}}}\right\}\label{res3s0}
\end{align}
\hrulefill
\vspace*{4pt}
\end{figure*}
with $\p{V_1,\ldots,V_K}\sim P_{V_1}\cdots\times P_{V_K}$.
\end{lemma}

\begin{proof}[Proof of Lemma \ref{lem:scar3}]
The second term in \eqref{res3s0} follows by first applying the union bound over $i$ as in \eqref{uB1}, shown at the top of page~\pageref{uB1},
\begin{figure*}[!t]
\normalsize
\setcounter{MYtempeqncnt}{\value{equation}}
\setcounter{equation}{1}
\begin{align}
&\Pr\ppp{\bigcup_{i,j}\ppp{\bigcup_{l=1}^N\ppp{V_1(i)\in\calA_{l,0},\;\bigcap_{k=1}^{K-1}\p{V_1(i),V_{k+1}(j)}\in\calA_{l,k},\;\bigcap_{k=2}^{K}V_k(j)\in\calG_{l,k}}}}\nonumber\\
&\leq L_1\Pr\ppp{\bigcup_{j}\ppp{\bigcup_{l=1}^N\ppp{V_1\in\calA_{l,0},\;\bigcap_{k=1}^{K-1}\p{V_1,V_{k+1}(j)}\in\calA_{l,k},\;\bigcap_{k=2}^{K}V_k(j)\in\calG_{l,k}}}}\nonumber\\
&\leq L_1\bE\ppp{\Pr\ppp{\left.\bigcup_{j}\ppp{\bigcup_{l=1}^N\ppp{V_1\in\calA_{l,0},\;\bigcap_{k=1}^{K-1}\p{V_1,V_{k+1}(j)}\in\calA_{l,k},\;\bigcap_{k=2}^{K}V_k(j)\in\calG_{l,k}}}\right\vert V_1}}\label{uB1}
\end{align}
\hrulefill
\vspace*{4pt}
\end{figure*}
and then we apply the truncated union bound to the union over $j$, and obtain \eqref{uB2}. 
\begin{figure*}[!t]
\normalsize
\setcounter{MYtempeqncnt}{\value{equation}}
\setcounter{equation}{2}
\begin{align}
&\Pr\ppp{\bigcup_{i,j}\ppp{\bigcup_{l=1}^N\ppp{V_1(i)\in\calA_{l,0},\;\bigcap_{k=1}^{K-1}\p{V_1(i),V_{k+1}(j)}\in\calA_{l,k},\;\bigcap_{k=2}^{K}V_k(j)\in\calG_{l,k}}}}\nonumber\\
&\leq L_1\bE\pp{\min\ppp{1,L_2\Pr\ppp{\left.\bigcup_{l=1}^N\ppp{V_1\in\calA_{l,0},\;\bigcap_{k=1}^{K-1}\p{V_1,V_{k+1}}\in\calA_{l,k},\;\bigcap_{k=2}^{K}V_k\in\calG_{l,k}}}\right\vert V_1}}\label{uB2}
\end{align}
\hrulefill
\vspace*{4pt}
\end{figure*}
The third term is obtained similarly by applying the union bounds in the opposite order, and the upper bound of 1 is trivial. 
\end{proof}

We are now in a position to prove Lemma \ref{lem:Union3}.
\begin{proof}[Proof of Lemma \ref{lem:Union3}]
To obtain \eqref{EquivErrorEvent01} we weaken \eqref{res3s0} as follows. Let $\calF\triangleq\bigcup_{l=1}^N\ppp{V_1\in\calB_{l,1}}$. The second term in \eqref{EquivErrorEvent01} follows from \eqref{Ub3}, shown at the top of page \pageref{Ub3}, 
\begin{figure*}[!t]
\normalsize
\setcounter{MYtempeqncnt}{\value{equation}}
\setcounter{equation}{3}
\begin{align}
&\min\ppp{1,L_2\Pr\ppp{\left.\bigcup_{l=1}^N\ppp{V_1\in\calA_{l,0},\;\bigcap_{k=1}^{K-1}\p{V_1,V_{k+1}}\in\calA_{l,k},\;\bigcap_{k=2}^{K}V_k\in\calG_{l,k}}\right\vert V_1}}\nonumber\\
&= \calI\ppp{\calF}\min\ppp{1,L_2\Pr\ppp{\left.\bigcup_{l=1}^N\ppp{V_1\in\calA_{l,0},\;\bigcap_{k=1}^{K-1}\p{V_1,V_{k+1}}\in\calA_{l,k},\;\bigcap_{k=2}^{K}V_k\in\calG_{l,k}}\right\vert V_1}}\nonumber\\
&\ +\calI\ppp{\calF^c}\min\ppp{1,L_2\Pr\ppp{\left.\bigcup_{l=1}^N\ppp{V_1\in\calA_{l,0},\;\bigcap_{k=1}^{K-1}\p{V_1,V_{k+1}}\in\calA_{l,k},\;\bigcap_{k=2}^{K}V_k\in\calG_{l,k}}\right\vert V_1}}\nonumber\\
& = \calI\ppp{\calF}\min\ppp{1,L_2\Pr\ppp{\left.\bigcup_{l=1}^N\ppp{V_1\in\calA_{l,0},\;\bigcap_{k=1}^{K-1}\p{V_1,V_{k+1}}\in\calA_{l,k},\;\bigcap_{k=2}^{K}V_k\in\calG_{l,k}}\right\vert V_1}}\nonumber\\
& \leq \calI\ppp{\calF}\label{Ub3}
\end{align}
\hrulefill
\vspace*{4pt}
\end{figure*}
where the second equality follows from the fact that the inner term in the expectation vanishes over $\bigcap_{l=1}^N\ppp{V_1\notin\calB_{l,1}}$, and the third inequality follows from the fact that $\min\ppp{1,x}\leq 1$. The third term in \eqref{EquivErrorEvent01} follows in a similar fashion, and the forth term follows from the fact that $\min\ppp{1,x}\leq x$, and thus we get \eqref{uB4}, 
\begin{figure*}[!t]
\normalsize
\setcounter{MYtempeqncnt}{\value{equation}}
\setcounter{equation}{4}
\begin{align}
&L_1\bE\pp{\min\ppp{1,L_2\Pr\ppp{\left.\bigcup_{l=1}^N\ppp{V_1\in\calA_{l,0},\;\bigcap_{k=1}^{K-1}\p{V_1,V_{k+1}}\in\calA_{l,k},\;\bigcap_{k=2}^{K}V_k\in\calG_{l,k}}\right\vert V_1}}}\nonumber\\
&\leq L_1L_2\Pr\ppp{\bigcup_{l=1}^N\ppp{V_1\in\calA_{l,0},\;\bigcap_{k=1}^{K-1}\p{V_1,V_{k+1}}\in\calA_{l,k},\;\bigcap_{k=2}^{K}V_k\in\calG_{l,k}}}\label{uB4}
\end{align}
\hrulefill
\vspace*{4pt}
\end{figure*}
which concludes the proof.
\end{proof}

\section{Computational Aspects of the Exponents}
\label{app:2}
In this appendix, we discuss the computation of \eqref{Estar12}, similarly as in \cite{ScarlettCog}. We start with an alternative formulation of \eqref{Estar12}. Recall that
\begin{align}
\tilde E_1(R_1,R_2) = \min_{Q}\ppp{D(Q||W)+E(Q,R_1,R_2)},\label{EAeff}
\end{align}
where
\begin{align}
E(Q,R_1,R_2) = \max\ppp{\pp{\hat{E}_1(Q,R_2)-R_1}_+,\hat{E}_2(Q,R_2)}.
\end{align}
In the following, for a given $Q_{Y_1|X_{1}X_{2}}$, we show that $\hat{E}_1(Q,R_2)$ and $\hat{E}_2(Q,R_2)$ can be calculated efficiently. For brevity, we let $\tilde{I}(\tilde Q)\equiv I_{\tilde Q}(X_1;X_2,Y_1)$ and $\hat{I}(\hat Q)\equiv I_{\hat Q}(X_2;X_1,Y_1)$. Recall that 
\begin{align}
\hat{E}_1(Q,R_2)&= \min_{\substack{\tilde{Q}\in\calS(Q),\hat{Q}\in\hat{S}(\tilde Q),\\ \hat{Q}\in\calL(\tilde Q,Q)}}\ppp{\tilde{I}(\tilde Q)+\pp{\hat{I}(\hat{Q})-R_2}_+},\label{0eq}\\
\hat{E}_2(Q,R_2)&= \min_{\substack{\tilde{Q}\in\calS(Q),\hat{Q}\in\hat{S}(\tilde Q),\\ \hat{Q}\in\hat{\calL}(\tilde Q,Q)}}\pp{\hat{I}_{\hat{Q}}(X_2;Y_1)-R_2}_+,\label{1eq}
\end{align}
where $\calL$ and $\hat{\calL}$ are defined in \eqref{LsetCases0} and \eqref{LsetCases}, respectively, $\calS(Q) = \{\tilde Q:\;\tilde{Q}_{X_1}=P_{X_1},\;\tilde{Q}_{X_2Y_1} = Q_{X_2Y_1}\}$, and $\hat{S}(\tilde Q) = \{\hat Q:\;\hat{Q}_{X_2}=P_{X_2},\;\hat{Q}_{X_1Y_1} = \tilde Q_{X_1Y_1}\}$. In \cite{ScarlettCog}, it was shown that $\hat{E}_1(Q,R_2)$ can be equivalently expressed in terms of the minimum between the following terms:
\begin{align}
\hat{E}'_{1}(Q,R_2)&\triangleq\min_{\substack{\tilde{Q}\in\calS(Q),\max\pp{t_0(Q),f(Q)}\leq f(\tilde Q)}}\tilde{I}(\tilde Q),\label{71eq3}\\
\hat{E}''_{1}(Q,R_2)&\triangleq\min_{\substack{\tilde{Q}\in\calS(Q),\hat{Q}\in\hat{S}(\tilde Q),\hat{I}(\hat{Q})\leq R_2\\ \max\pp{t_0(Q),f(Q)}\leq f(\hat Q)+R_2-\hat{I}(\hat Q)}}\tilde{I}(\tilde Q),\label{71eq2}\\
\hat{E}_{1}'''(Q,R_2)&\triangleq\min_{\substack{\tilde{Q}\in\calS(Q),\hat{Q}\in\hat{S}(\tilde Q)\\ t_0(Q),f(Q)\leq f(\hat Q)}}\ppp{\tilde{I}(\tilde Q)+\pp{\hat{I}(\hat Q)-R_2}_+}.\label{123eq2}
\end{align}
Using the same arguments as in \cite{ScarlettCog}, it can be shown that $\hat{E}_2(Q,R_2)$ can be equivalently be expressed as
\begin{align}
\hat{E}_{2}''(Q,R_2)&\triangleq\min_{\substack{\tilde{Q}\in\calS(Q),\hat{Q}\in\hat{S}(\tilde Q)\\ \max\pp{t_0(Q),f(Q)}\leq f(\hat Q)}}\pp{I_{\hat Q}(X_2;Y_1)-R_2}_+.\label{123eq2B}
\end{align}
Accordingly, from \eqref{71eq3}-\eqref{123eq2} and \eqref{123eq2B}, we see that \eqref{0eq} and \eqref{1eq} can be expressed in terms of convex optimization problems, namely, for a given $Q_{Y_1|X_{1}X_{2}}$, the terms $\hat{E}_1(Q,R_2)$ and $\hat{E}_2(Q,R_2)$ (i.e., the inner terms of the minimization problem in \eqref{EAeff}) can be calculated efficiently, as desired.

Finally, we discuss the computation of \eqref{EAeff}. Generally speaking, the minimization over $Q_{Y_1|X_1X_2}$ might not be a convex problem, and thus one should resort to global optimization methods (e.g., a simple algorithm is an exhaustive search over a fine grid of probability simplex). Nonetheless, in the following we somewhat simplify these optimizations. We first see that \eqref{EAeff} can be rewritten as
\begin{align}
\tilde E_1(R_1,R_2)& \triangleq\min\ppp{\tilde E_1'(R_1,R_2),\tilde E_1''(R_1,R_2)}
\end{align}
where
\begin{align}
\tilde E_1'(R_1,R_2)&=\inf_{Q:\;R_1<\hat{E}_1(Q,R_2)-\hat{E}_2(Q,R_2)}\nonumber\\
&\left\{D(Q||W)+\hat{E}_1(Q,R_2)-R_1\right\},
\end{align}
and
\begin{align}
\tilde E_1''(R_1,R_2)&=\inf_{Q:\;R_1\geq \hat{E}_1(Q,R_2)-\hat{E}_2(Q,R_2)}\nonumber\\
&\ \ \ \ \ \ppp{D(Q||W)+\hat{E}_2(Q,R_2)}.
\end{align}
Let us analyze $\tilde E_1'(R_1,R_2)$ as a function of $R_1$. For $R_1=0$, we have
\begin{align}
\tilde E_1'(0,R_2)=\inf_{Q}\ppp{D(Q||W)+\hat{E}_1(Q,R_2)}.
\end{align}
Now, letting the minimizer be $\left.Q^*_{Y_1|X_1X_2}\right|_{R_1=0}$, and defining the critical rate $R_{1,\text{crit}} = \hat{E}_1(\left.Q^*_{Y_1|X_1X_2}\right|_{R_1=0},R_2)-\hat{E}_2(\left.Q^*_{Y_1|X_1X_2}\right|_{R_1=0},R_2)$, it is easily noticed that for $R_1\leq R_{1,\text{crit}}$, the exponent is an affine function
\begin{align}
\tilde E_1'(R_1,R_2) = \tilde E_1'(0,R_2)-R_1.
\end{align}
Furthermore, for $R_1\leq R_{1,\text{crit}}$, it is readily seen that
\begin{align}
\tilde E_1(R_1,R_2) = \tilde E_1'(R_1,R_2).
\end{align}
For $R_1> R_{1,\text{crit}}$, however, since the optimization of $Q$ is not convex, the term cannot be simplified anymore. 

\section{Proof of \eqref{ordinregion}}
\label{app:3}
First, note that the following region:
\begin{align}
&R_1<I_W(X_1;Y_1|X_2)\label{10},\\
&R_1<I_W(X_1;Y_1)+\pp{I_W(X_2;Y_1|X_1)-R_2}_+\label{101},
\end{align}
evaluated with $P_{X_1X_2Y_1} = P_{X_1}\times P_{X_2}\times W_{Y_1|X_1X_2}$, is equivalent to \eqref{ordinregion}. Thus, to show that \eqref{ordinregion} is achievable, it suffices to show that the above region is achievable. Now, recall that the ordinary random coding exponent, in the single-user setting, is given by
\begin{align}
E_r(R) = \min_Q\ppp{D(Q_{Y|X}||W|P_X)+\pp{I_Q(X;Y)-R}_+}.\nonumber
\end{align}
From \cite[Lemma 9]{Galg} it can be shown that $E_r(R)$ can be rewritten as
\begin{align}
E_r(R)  &=\min_{(Q,\tilde Q)\in \calD: f(\tilde{Q})\leq f(Q)}\left\{D(\tilde{Q}_{Y|X}||W|P_X)\right.\nonumber\\
&\ \ \ \ \ \ \ \ \ \ \ \ \ \ \ \ \ \ \ \ \ \ \left. \ \ \ \ \ +\pp{I_Q(X;Y)-R}_+\right\},
\end{align}
where $f(Q) = \bE_Q\ppp{\log W(Y|X)}$, and $\calD$ is the set of $(Q,\tilde Q)$ distributions such that $Q_Y = \tilde{Q}_Y$, and $Q_X = \tilde{Q}_X = P_X$. The last representation is very similar, in some sense, to the error exponent formula in Theorem \ref{th:1}. It can be seen that $E_r(R)$ is positive as long as
\begin{align}
R<\min_{Q: Q_Y = W_Y, Q_X = P_X, f(W)\leq f(Q)}I_Q(X;Y).\label{11}
\end{align}
Obviously, we should get that $R<I_W(X;Y)$, namely, the minimum in \eqref{11} should be equal to $I_W(X;Y)$. To see that this is indeed the case, note that since the above optimization problem is convex, the linear constraint is met with equality, and we note that $f(W) = -H_W(Y|X)$, $f(Q) = -D(Q||W|P_X)-H_Q(Y|X)$, and $I_Q(X;Y) = H_W(Y)-H_Q(Y|X)$. Using the last facts, we get
\begin{align}
R&<\min_{Q: Q_Y = W_Y, Q_X = P_X, f(W)= f(Q)}I_Q(X;Y)\nonumber\\
& = \min_{\substack{Q: Q_Y = W_Y, Q_X = P_X,\\ H_W(Y|X)= D(Q||W|P_X)+H_Q(Y|X)}}H_W(Y)-H_Q(Y|X)\nonumber\\
& = I_W(X;Y) \nonumber\\
&\ \ \ \ \ \ +\min_{\substack{Q: Q_Y = W_Y, Q_X = P_X,\\ H_W(Y|X)= D(Q||W|P_X)+H_Q(Y|X)}}D(Q||W|P_X)\nonumber\\
& = I_W(X;Y),\nonumber
\end{align}
as required. 

In our case, using the equivalent representation of our error exponent in \eqref{71eq3}-\eqref{123eq2}, we readily get that the error exponent in Theorem \ref{th:1} is positive if\footnote{To show that \eqref{ordinregion} is achievable, we consider a lower bound on $\tilde{E}(R_1,R_2)$, which ignores the contribution of $\hat{E}_2(Q,R_2)$, namely, $\min_Q\ppp{D(Q||W|P_X)+\pp{\hat{E}_1(Q,R_2)-R_1}_+}$.}:
\begin{align}
R_1&<\min_{\substack{\tilde{Q}\in\calS(W),\max\pp{t_0(W),f(W)}\leq f(\tilde Q)}}\tilde{I}(\tilde Q),\label{con1}\\
R_1&<\min_{\substack{\tilde{Q}\in\calS(W),\hat{Q}\in\hat{S}(\tilde Q),\hat{I}(\hat{Q})\leq R_2\\ \max\pp{t_0(W),f(W)}\leq f(\hat Q)+R_2-\hat{I}(\hat Q)}}\tilde{I}(\tilde Q),\label{con2}\\
R_1&<\min_{\substack{\tilde{Q}\in\calS(W),\hat{Q}\in\hat{S}(\tilde Q)\\ \max\pp{t_0(W),f(W)}\leq f(\hat Q)}}\ppp{\tilde{I}(\tilde Q)+\pp{\hat{I}(\hat Q)-R_2}_+}\label{con3}
\end{align}
where we recall that $\tilde{I}(\tilde Q)\equiv I_{\tilde Q}(X_1;X_2,Y_1)$ and $\hat{I}(\hat Q)\equiv I_{\hat Q}(X_2;X_1,Y_1)$. In the following, we show that \eqref{con1} and \eqref{con3} correspond  to \eqref{10} and \eqref{101}, respectively. Finally, we show that \eqref{con1} is dominated by \eqref{con1} and \eqref{con3}, and thus superfluous. Indeed, for \eqref{con1}, we have
\begin{align}
&\min_{\substack{\tilde{Q}\in\calS(W),\max\pp{t_0(W),f(W)}\leq f(\tilde Q)}}\tilde{I}(\tilde Q) \nonumber\\
&= \min_{\substack{\tilde{Q}_{Y|X_2}=W_{Y|X_2},\max\pp{t_0(W),f(W)}\leq f(\tilde Q)}}I_{\tilde{Q}}(X_1;X_2,Y)\nonumber\\
&\geq \min_{\substack{\tilde{Q}_{Y|X_2}=W_{Y|X_2},f(W)\leq f(\tilde Q)}}I_{\tilde{Q}}(X_1;X_2,Y).\label{min}
\end{align}
Now, as before, we note that (using the fact that the minimization over $\tilde Q$ in \eqref{min} is such that $\tilde{Q}_{Y|X_2}=W_{Y|X_2}$)
\begin{align}
I_{\tilde{Q}}(X_1;X_2,Y) &= I_{\tilde{Q}}(X_1;X_2)+I_{\tilde{Q}}(X_1;Y|X_2)\\
&= I_{\tilde{Q}}(X_1;X_2)+H_W(Y|X_2)\nonumber\\
&\ \ \ \ -H_{\tilde{Q}}(Y|X_1,X_2),\\
f(W)& = -H_W(Y|X_1,X_2),\label{propre1}
\end{align}
and
\begin{align}
f(\tilde Q) = -D(\tilde{Q}||W|P_X)-H_{\tilde{Q}}(Y|X_1,X_2).\label{propre2}
\end{align}
Thus, we have \eqref{CapDiv}, shown at the top of the next page, 
\begin{figure*}[!t]
\normalsize
\setcounter{MYtempeqncnt}{\value{equation}}
\setcounter{equation}{12}
\begin{align}
&\min_{\substack{\tilde{Q}_{Y|X_2}=W_{Y|X_2},f(W)\leq f(\tilde Q)}}I_{\tilde{Q}}(X_1;X_2,Y) \nonumber\\
&= \min_{\substack{\tilde{Q}_{Y|X_2}=W_{Y|X_2},\\ D(\tilde{Q}||W|P_X)+H_{\tilde{Q}}(Y|X_1,X_2)\leq H_W(Y|X_1,X_2)}}I_{\tilde{Q}}(X_1;X_2)+H_W(Y|X_2)-H_{\tilde{Q}}(Y|X_1,X_2) \nonumber\\
&\geq I_W(X_1;Y|X_2)+\min_{\substack{\tilde{Q}_{Y|X_2}=W_{Y|X_2},\\ D(\tilde{Q}||W|P_X)+H_{\tilde{Q}}(Y|X_1,X_2)\leq H_W(Y|X_1,X_2)}}I_{\tilde{Q}}(X_1;X_2)+D(\tilde{Q}||W|P_X)\nonumber\\
&= I_W(X_1;Y|X_2)\label{CapDiv}
\end{align}
\hrulefill
\vspace*{4pt}
\end{figure*}
where the inequality follows from the fact that $-H_{\tilde{Q}}(Y|X_1,X_2)\geq D(\tilde{Q}||W|P_X)-H_W(Y|X_1,X_2)$ induced by the optimization constraint, and the last equality is achieved by taking $\tilde{Q}=W$ and $\tilde{Q}_{X_1,X_2} = P_{X_1}P_{X_2}$. The constraint in \eqref{con3} is handled in a similar manner. Indeed, using the same manipulations, we get \eqref{supported}, shown at the top of the next page, 
\begin{figure*}[!t]
\normalsize
\setcounter{MYtempeqncnt}{\value{equation}}
\setcounter{equation}{13}
\begin{align}
&\min_{\substack{\tilde{Q}_{Y|X_2} = W_{Y|X_2},\hat{Q}_{Y|X_1} = \tilde{Q}_{Y|X_1}\\ D(\hat{Q}||W|P_X)+H_{\hat{Q}}(Y|X_1,X_2)\leq H_W(Y|X_1,X_2)}}\ppp{\tilde{I}(\tilde{Q})+\pp{I_{\hat{Q}}(X_1;X_2)+H_{\tilde{Q}}(Y|X_1)-H_{\hat Q}(Y|X_1,X_2)-R_2}_+}\nonumber\\
&\geq\min_{\substack{\tilde{Q}_{Y|X_2} = W_{Y|X_2},\hat{Q}_{Y|X_1} = \tilde{Q}_{Y|X_1}\\ D(\hat{Q}||W|P_X)+H_{\hat{Q}}(Y|X_1,X_2)\leq H_W(Y|X_1,X_2)}}\left\{\tilde{I}(\tilde{Q})+\left[I_{\hat{Q}}(X_1;X_2)+D(\hat{Q}||W|P_X)+H_{\tilde{Q}}(Y|X_1)\right.\right.\nonumber\\
&\left.\left.\ \ \  \ \ \ \ \ \ \ \ \  \ \ \ \ \ \ \ \ \  \ \ \ \ \ \ \ \ \  \ \ \ \ \ \ \ \ \ \ \ \ \ \ \ \ \ \ \ \ \ \ \ \ \ \ \ \ \ \vphantom{\left\{\tilde{I}(\tilde{Q})+\left[I_{\hat{Q}}(X_1;X_2)+D(\hat{Q}||W|P_X)+H_{\tilde{Q}}(Y|X_1)\right.\right.}-H_{W}(Y|X_1,X_2)-R_2\right]_+\right\}\nonumber\\
&=\min_{\substack{\tilde{Q}_{Y|X_2} = W_{Y|X_2},\tilde{Q}_{Y|X_1}=W_{Y|X_1}}}\left\{I_{\tilde{Q}}(X_1;Y)+I_{\tilde{Q}}(X_1;X_2|Y)+\left[H_{\tilde{Q}}(Y|X_1)-H_{W}(Y|X_1,X_2)-R_2\right]_+\right\}\nonumber\\
&=\min_{\substack{\tilde{Q}_{Y|X_2} = W_{Y|X_2},\tilde{Q}_{Y|X_1}=W_{Y|X_1}}}\left\{I_{W}(X_1;Y)+I_{\tilde{Q}}(X_1;X_2|Y)+\left[H_{W}(Y|X_1)-H_{W}(Y|X_1,X_2)-R_2\right]_+\right\}\nonumber\\
&=I_{W}(X_1;Y)+\left[I_{W}(X_2;Y|X_2)-R_2\right]_+,\label{supported}
\end{align}
\hrulefill
\vspace*{4pt}
\end{figure*}
where the inequality is due to the fact that $-H_{\tilde{Q}}(Y|X_1,X_2)\geq D(\tilde{Q}||W|P_X)-H_W(Y|X_1,X_2)$, the second equality follows by taking $\hat{Q}=W$, and the last equality follows by taking $\tilde{Q}$ such that $X_2-Y-X_1$ is a Markov chain. 

Finally, we show that the constraint in \eqref{con2} is superfluous. To this end, we will show that for $R_2<I_W(X_2;Y|X_1)$, the r.h.s. of \eqref{con2} reduces to $R_1+R_2<I(X_1,X_2;Y)$, which is dominated by \eqref{con1} and \eqref{con3} (or, equivalently, by \eqref{CapDiv} and \eqref{supported}), and for $R_2\geq I_W(X_2;Y|X_1)$, \eqref{con2} reduces to $R_1<I_{W}(X_1;Y)$, already supported by \eqref{con3} (see \eqref{supported}). Whence, \eqref{con2} is redundant. Indeed, for $R_2\geq I_W(X_2;Y|X_1)$, the r.h.s. of \eqref{con2} can be lower bounded as in \eqref{UpperAppD1}, presented at the top of the next page, 
\begin{figure*}[!t]
\normalsize
\setcounter{MYtempeqncnt}{\value{equation}}
\setcounter{equation}{14}
\begin{align}
&\min_{\substack{\tilde{Q}\in\calS(W),\hat{Q}\in\hat{S}(\tilde Q),\hat{I}(\hat{Q})\leq R_2\\ \max\pp{t_0(W),f(W)}\leq f(\hat Q)+R_2-\hat{I}(\hat Q)}}\tilde{I}(\tilde Q)\geq \min_{\substack{\tilde{Q}\in\calS(W),\hat{Q}\in\hat{S}(\tilde Q),\hat{I}(\hat{Q})\leq R_2\\ t_0(W)\leq f(\hat Q)+R_2-\hat{I}(\hat Q)}}\tilde{I}(\tilde Q)\nonumber\\
& \stackrel{(a)}{=} \min_{\substack{\tilde{Q}\in\calS(W),\hat{Q}\in\hat{S}(\tilde Q),\hat{I}(\hat{Q})\leq R_2\\ \max_{\hat{Q}:\;\hat{Q}\in\calS(W),\hat{I}(\hat{Q})\leq R_2}\pp{f(\hat{Q})-I_{\hat{Q}}(X_2;X_1,Y)}\leq f(\hat Q)-I_{\hat{Q}}(X_2;X_1,Y)}}I_{\tilde{Q}}(X_1;X_2,Y)\nonumber\\
&\stackrel{(b)}{\geq} \min_{\substack{\tilde{Q}\in\calS(W),\hat{Q}\in\hat{S}(\tilde Q),\hat{I}(\hat{Q})\leq R_2\\ f(W)-I_W(X_2;X_1,Y)\leq f(\hat Q)-I_{\hat{Q}}(X_2;X_1,Y)}}I_{\tilde{Q}}(X_1;X_2,Y)\nonumber\\
& \stackrel{(c)}{=} \min_{\substack{\tilde{Q}\in\calS(W),\hat{Q}\in\hat{S}(\tilde Q),\hat{I}(\hat{Q})\leq R_2\\ -H_W(Y|X_1)\leq -D(\hat{Q}||W|P_X)-I_{\hat{Q}}(X_1;X_2)-H_{\tilde{Q}}(Y|X_1)}}H_W(Y)-H_{\tilde{Q}}(Y|X_1)+I_{\tilde{Q}}(X_1;X_2|Y)\nonumber\\
& \stackrel{(d)}{\geq} I_{W}(X_1;Y)+\min_{\substack{\tilde{Q}\in\calS(W),\hat{Q}\in\hat{S}(\tilde Q),\hat{I}(\hat{Q})\leq R_2\\ -H_W(Y|X_1)\leq -D(\hat{Q}||W|P_X)-I_{\hat{Q}}(X_1;X_2)-H_{\tilde{Q}}(Y|X_1)}}D(\hat{Q}||W|P_X)+I_{\hat{Q}}(X_1;X_2)+I_{\tilde{Q}}(X_1;X_2|Y)\nonumber\\
& \stackrel{(e)}{=} I_{W}(X_1;Y)\label{UpperAppD1}
\end{align}
\hrulefill
\vspace*{4pt}
\end{figure*}
where in (a) we use the definition of $t_0(W)$ in \eqref{t0def}, (b) follows from the assumption that $R_2\geq I_W(X_2;Y|X_1)$, (c) is due to \eqref{propre1}-\eqref{propre2}, (d) follows from the fact that $-H_{\tilde{Q}}(Y|X_1)\geq D(\hat{Q}||W|P_X)+I_{\hat{Q}}(X_1;X_2)-H_W(Y|X_1)$ induced by the optimization constraint, and (e) is achieved by taking $\hat{Q}=W$, $\hat{Q}_{X_1X_2} = \hat{Q}_{X_1}\hat{Q}_{X_2}$, and $\tilde{Q}$ such that $X_2-Y-X_1$ is a Markov chain. Thus, for $R_2\geq I_W(X_2;Y|X_1)$, we obtained that $R_1<I_{W}(X_1;Y)$, as required. On the other hand, for $R_2<I_W(X_2;Y|X_1)$, the r.h.s. of \eqref{con2} can be lower bounded as shown in \eqref{UpperAppD2}, given in the next page, 
\begin{figure*}[!t]
\normalsize
\setcounter{MYtempeqncnt}{\value{equation}}
\setcounter{equation}{15}
\begin{align}
&\min_{\substack{\tilde{Q}\in\calS(W),\hat{Q}\in\hat{S}(\tilde Q),\hat{I}(\hat{Q})\leq R_2\\ \max\pp{t_0(W),f(W)}\leq f(\hat Q)+R_2-\hat{I}(\hat Q)}}\tilde{I}(\tilde Q)\geq \min_{\substack{\tilde{Q}\in\calS(W),\hat{Q}\in\hat{S}(\tilde Q),\hat{I}(\hat{Q})\leq R_2\\ f(W)\leq f(\hat Q)+R_2-\hat{I}(\hat Q)}}\tilde{I}(\tilde Q)\nonumber\\
&\stackrel{(a)}{=} \min_{\substack{\tilde{Q}\in\calS(W),\hat{Q}\in\hat{S}(\tilde Q),\hat{I}(\hat{Q})\leq R_2\\ -H_W(Y|X_1,X_2)\leq -D(\hat{Q}||W|P_X)-H_{\hat{Q}}(Y|X_1)-I_{\hat{Q}}(X_1;X_2)+R_2}}H_W(Y)-H_{\tilde{Q}}(Y|X_1)+I_{\tilde{Q}}(X_1;X_2|Y)\nonumber\\
&\stackrel{(b)}{=} \min_{\substack{\tilde{Q}\in\calS(W),\hat{Q}\in\hat{S}(\tilde Q),\hat{I}(\hat{Q})\leq R_2\\ -H_W(Y|X_1,X_2)\leq -D(\hat{Q}||W|P_X)-H_{\hat{Q}}(Y|X_1)-I_{\hat{Q}}(X_1;X_2)+R_2}}H_W(Y)-H_{\hat{Q}}(Y|X_1)+I_{\tilde{Q}}(X_1;X_2|Y)\nonumber\\
& \stackrel{(c)}{\geq} I_W(X_1,X_2;Y)-R_2\nonumber\\
&\ \ \ +\min_{\substack{\tilde{Q}\in\calS(W),\hat{Q}\in\hat{S}(\tilde Q),\hat{I}(\hat{Q})\leq R_2\\ -H_W(Y|X_1,X_2)\leq -D(\hat{Q}||W|P_X)-H_{\hat{Q}}(Y|X_1)-I_{\hat{Q}}(X_1;X_2)+R_2}}D(\hat{Q}||W|P_X)+I_{\hat{Q}}(X_1;X_2)+I_{\tilde{Q}}(X_1;X_2|Y)\nonumber\\
&\geq I_W(X_1,X_2;Y)-R_2\label{UpperAppD2}
\end{align}
\hrulefill
\vspace*{4pt}
\end{figure*}
where (a) is due to \eqref{propre1}-\eqref{propre2}, (b) is because $\hat{Q}\in\calS(\tilde{Q})$ and thus $H_{\tilde{Q}}(Y|X_1)=H_{\hat{Q}}(Y|X_1)$, and (c) follows from the fact that $-H_{\hat{Q}}(Y|X_1)\geq D(\hat{Q}||W|P_X)+I_{\hat{Q}}(X_1;X_2)-H_W(Y|X_1,X_2)-R_2$ induced by the optimization constraint. Whence, we obtained that $R_1+R_2\leq I_W(X_1,X_2;Y)$, as required. 

\section*{Acknowledgment}

The authors would like to thank the associate editor, Jun Chen, and the anonymous referees for their suggestions and comments which helped improving the content of this paper. The authors are also grateful to Jonathan Scarlett for very helpful comments on a previous version of this paper, and for drawing our attention to \cite{ScarlettCog} (and its extended version which is not published yet), which was unknown to the authors at the time of writing the original version of this paper.

\ifCLASSOPTIONcaptionsoff
  \newpage
	\newpage
\fi
\bibliographystyle{IEEEtran}
\bibliography{strings}
\end{document}